\documentclass[11pt]{article}
\usepackage[margin=1in]{geometry}
\usepackage{caption}
\usepackage{float}
\RequirePackage[authoryear]{natbib}
\usepackage{bbm}
\usepackage{mathrsfs}
\usepackage{abbrev}
\usepackage{color,xcolor}
\usepackage{graphicx}
\usepackage{lmodern}
\usepackage{anyfontsize}
\usepackage{algorithm}
\usepackage{algpseudocode}
\usepackage{comment}
\usepackage{longtable}
\usepackage{booktabs}
\usepackage{multicol}
\usepackage{enumitem}
\usepackage{hyperref}
\hypersetup{
  colorlinks,
  linkcolor={red!50!black},
  citecolor={blue},
  urlcolor={blue!80!black}
}
\allowdisplaybreaks

\usepackage{setspace}
\usepackage{authblk}

\theoremstyle{plain}
\newtheorem{theorem}{Theorem}
\newtheorem{proposition}{Proposition}

\newtheorem{lemma}{Lemma}
\newtheorem{corollary}{Corollary}

\theoremstyle{definition}
\newtheorem{assumption}{Assumption}
\newtheorem{definition}{Definition}

\newcommand{\sigmat}{\sigma_{\theta}}

\def\spacingset#1{\renewcommand{\baselinestretch}
{#1}\small\normalsize}

\spacingset{1}
\usepackage{authblk}
\title{\textbf{La}tency-\textbf{R}esponse \textbf{T}heory Model: Evaluating Large Language Models via Response Accuracy and Chain-of-Thought Length}
\author[1]{Zhiyu Xu}
\author[1]{Jia Liu}
\author[2]{Yixin Wang}
\author[1]{Yuqi Gu\footnote{\texttt{\{zx2488,jl6795\}@columbia.edu}; \texttt{yixinw@umich.edu}; \texttt{yuqi.gu@columbia.edu}. Corresponding author: Yuqi Gu.}}

\affil[1]{Department of Statistics, Columbia University.}
\affil[2]{Department of Statistics, University of Michigan}

\date{}
\begin{document}
\maketitle
\vspace{-5mm}
\begin{abstract}
The proliferation of Large Language Models (LLMs) necessitates valid evaluation methods to provide guidance for both downstream applications and actionable future improvements. The Item Response Theory (IRT) model with Computerized Adaptive Testing has recently emerged as a promising framework for evaluating LLMs via their response accuracy. Beyond simple response accuracy, LLMs' chain of thought (CoT) lengths serve as a vital indicator of their reasoning ability. To leverage the CoT length information to assist LLM evaluation, we propose the \textbf{La}tency-\textbf{R}esponse \textbf{T}heory (LaRT) model, which jointly models both the response accuracy and CoT length by introducing a key correlation parameter between the latent ability and the latent speed. We derive an efficient stochastic approximation Expectation-Maximization algorithm for parameter estimation. We establish rigorous identifiability results for the latent ability and latent speed parameters to ensure the statistical validity of their estimation. Through both theoretical asymptotic analyses and simulation studies, we demonstrate LaRT's advantages over IRT in terms of superior estimation accuracy and shorter confidence intervals for latent trait estimation. To evaluate LaRT in real data, we collect responses from diverse LLMs on popular benchmark datasets. We find that LaRT yields different LLM rankings than IRT and outperforms IRT across multiple key evaluation metrics including predictive power, item efficiency, ranking validity, and LLM evaluation efficiency. Code and data are available at \url{https://github.com/Toby-X/Latency-Response-Theory-Model}.
\end{abstract}

\section{Introduction}

As large language models (LLMs) continue to advance across diverse tasks, it has become increasingly crucial to effectively and reliably evaluate their abilities \citep{liang2023holistic}. A common approach to evaluating an LLM’s capability is using a corresponding benchmark, such as GSM8K for mathematical reasoning \citep{cobbe2021training}, HumanEval for code generation \citep{chen2021evaluating}, or MMLU for general knowledge reasoning.
The LLM is run on the benchmark items, its outputs are scored using task-specific metrics, and the resulting score is used to compare or rank LLMs. However, benchmark scores alone provide only a coarse measure of performance and offer limited insight into how or why LLMs differ, as they do not account for variation in item characteristics or model-specific response patterns. \textcolor{black}{To address this, item characterization has gained significant traction for LLM evaluation, including application in selecting high-quality test questions from the increasingly saturated pool of LLM evaluation benchmarks.}

Recently, Item Response Theory (IRT), together with Computerized Adaptive Testing (CAT), from the psychometrics literature has emerged as a promising framework for evaluating LLMs in a more principled and interpretable manner \citep{laloretal2016building, zhuang2023efficiently, polo2024tinybenchmarks, kipnis2025metabench, castleman2025rethinking, hofmann2025fluid}. Unlike raw benchmark scoring, which treats all items as equally informative, IRT explicitly models variation in item difficulty and discrimination.
\textcolor{black}{This granular item-level information not only provides a latent ability estimate that is comparable across models and datasets, but also yields a deeper understanding of the benchmark datasets themselves.} In this context, IRT can be used in two complementary ways. First, in a \emph{static evaluation setting}, IRT is applied to existing benchmark responses to estimate each model’s latent ability and to produce model rankings. Second, IRT can also support an \emph{active evaluation setting}, where item parameters are first estimated from a pool of problems and then used to infer a model’s ability based on selectively chosen items \citep{eindoretal2020active}. This enables tracking how a model’s ability evolves during training or under limited evaluation budgets, and closely relates to CAT procedures that aim to estimate ability efficiently by administering only the most informative items \citep{chang1996global}.

Existing IRT-based evaluations of LLMs consider only the final response outcome. Yet, the response-generation process inherently contains information about the model’s reasoning behavior. In particular, \emph{chain of thought} (CoT) is a vital indicator of reasoning in LLMs. \textcolor{black}{The seminal paper \citet{wei2022chain} defines chain-of-thought as \textit{"A series of intermediate natural language reasoning steps that lead to the final output."}} \citet{wei2022chain} demonstrated that prompts eliciting CoT reasoning significantly enhance LLMs' ability to resolve complex tasks.
Given the autoregressive nature of LLMs, extending the CoT length necessitates additional decoding steps, thereby increasing the computation at test time.
\citet{snell2025scaling} discovered that scaling the test-time compute via extending CoT length is often more effective than scaling model parameters. Consequently, recent training strategies for reasoning models prioritize extended CoT generation. For instance, DeepSeek-R1 targets longer CoT sequences in training to improve performance \citep{guo2025deepseek}. OpenAI’s o1 similarly emphasizes eliciting CoT during training \citep{jaech2024openai}. Therefore, CoT length can serve
as a process-level indicator of how much intermediate reasoning the model performs prior to producing an answer, \textcolor{black}{and longer CoT length is associated with stronger reasoning ability}. Incorporating both accuracy and CoT length into a unified modeling and evaluation framework can therefore yield more reliable and discriminative ability estimates, which motivates this work.
CoT length closely parallels the role of response time (RT) in student assessment, where the duration of problem solving is used as an auxiliary indicator of cognitive effort beyond correctness alone (e.g., \citealp{vanderlinden2007Hierarchical, meng2015Conditional, klein2009evaluating}).
\color{black}
Furthermore, within the specific context of LLM evaluation, it is important to recognize that models autonomously dictate their own computational budgets during inference. Unlike traditional algorithms where compute is a strictly exogenous constraint, an LLM's generated reasoning length is an endogenous, self-determined action. Therefore, explicitly incorporating this autonomous behavioral dimension will capture a more holistic representation of an LLM's nature, while simultaneously providing quantitative insight into the expected marginal gain in ability associated with their self-determined computational budgets.
\color{black}

\color{black}
The joint modeling of response accuracy and response time has been studied in psychometrics (see \citealp{schnipke2005exploring, de2019overview, kyllonen2016use}, for overviews). Among these approaches, the hierarchical IRT–RT framework \citep{vanderlinden2007Hierarchical} is particularly suitable for our setting. The hierarchical approach does not commit to a particular cognitive mechanism and treats response time as an observable process signal linked to latent traits. This leads to a flexible and assumption-light measurement framework.

Despite these developments, existing methodologies in psychometrics  face limitations that hinder their direct application to LLM evaluation. From a statistical perspective,  these models \emph{lack strict theoretical identifiability guarantees}. Identifiability ensures that no two distinct parameter configurations yield the same marginal distribution of observed data, and it is essential for consistency, interpretability, and reliable parameter recovery.
From a computational perspective, existing models are typically \emph{computationally expensive}, hence difficult to scale to LLM evaluation settings, where benchmark datasets typically contain a large number of items. Specifically, existing approaches often estimate the joint models using Markov chain Monte Carlo (MCMC; e.g. \citealt{fox2016joint, bolsinova2019modeling}) algorithms due to the intractable integral in the marginal likelihood.
Yet, MCMC-based estimation is computationally intensive,
as each iteration requires repeatedly sampling all latent traits and item parameters, leading to slow convergence when the item pool is large.
\color{black}

\color{black}
This article makes the following contributions. \emph{First}, we propose an identifiable Latency–Response Theory (LaRT) model through jointly modeling response accuracy and chain-of-thought length. LaRT introduces a key unknown correlation parameter $\rho\in(-1,1)$ between latent ability and the latent speed, allowing the two signals to complement one another in estimating model proficiency. We establish rigorous identifiability guarantees to ensure reliable inference for the LLM evaluation. Furthermore, our theoretical asymptotic analysis and simulation studies demonstrate that LaRT yields shorter confidence intervals and significantly higher estimation accuracy than IRT.

\emph{Second}, we develop an efficient stochastic approximation Expectation-Maximization (SAEM) algorithm to estimate population parameters and individual latent traits of the LLMs. To accelerate and stabilize SAEM iterations, we design an MCMC-free latent-trait sampler.
After estimating the population parameters, individual latent traits are obtained via maximum-a-posteriori (MAP) estimation, which reduces to a convex optimization problem.

\emph{Third}, to evaluate LaRT in practical settings, we conduct comprehensive empirical data analysis on real data by collecting responses from diverse LLMs on popular benchmark datasets. \textcolor{black}{Qualitatively},
we observe an interpretable and consistent trend  across all benchmark datasets: LLMs with greater latent ability exhibit lower latent speed (longer CoT).
Quantitatively, we compare the performance of LaRT against IRT on the real-world data.
LaRT  outperforms IRT  across multiple key evaluation metrics in LLM evaluation: \textit{predictive power}, \textit{validity}, \textit{item efficiency}, and \textit{LLM efficiency}.
\color{black}

The rest of this paper is organized as follows. Section~\ref{sec:motivation_desiderata} introduces the motivating dataset for this paper and several desiderata valued by LLM evaluation. Section~\ref{sec:model} introduces the Latency-Response Theory model. Section \ref{sec:algorithm} presents an efficient SAEM algorithm with data-driven initialization for estimating population parameters, and a maximum a posteriori (MAP) algorithm to estimate the latent traits. Section \ref{sec:theory} presents the identifiability results for LaRT and the asymptotic distribution of the latent traits. In Section \ref{sec:simulation}, we perform simulations to validate the performance of our proposed methodology. Finally, Section \ref{sec:application} applies LaRT to real-world LLM evaluation settings, offering a comprehensive qualitative and quantitative comparison against standard IRT across multiple performance dimensions. Section \ref{sec:discussion} concludes and discusses future directions. Code and data are available at \url{https://github.com/Toby-X/Latency-Response-Theory-Model}.

\color{black}
\section{Motivating Datasets and Evaluation Desiderata}
\label{sec:motivation_desiderata}
In this paper, we focus on binary (correct/incorrect) LLM responses, a setting that encompasses a wide range of application scenarios, including mathematical reasoning \citep{hendrycksmath2021}, reading comprehension \citep{clark2019boolq}, hallucination detection \citep{li2023halueval}, and natural language inference \citep{wang2019superglue}.

As a primary application, we evaluate mathematical reasoning abilities using four benchmark datasets: MATH500 \citep{hendrycksmath2021} (500 questions), AMC23 (40 questions), AIME24 (30 questions), and AIME25 (30 questions). These benchmarks comprise advanced high-school competition problems of increasing difficulty, with MATH500 being the easiest and AIME25 the hardest. We evaluate over 80 open-source LLMs ranging from 0.6 billion to 32 billion parameters (detailed in Appendix F of the Supplementary Material; \citealp{xu2026supplement}). Following \citet{castleman2025rethinking}, we employ both zero-shot and one-shot chain-of-thought (CoT) prompting to broaden the scope of evaluated models (see Appendix F of the Supplementary Material for prompt details). \textcolor{black}{To generate the LLM responses, we bypass standard conversational interfaces and query the models directly via API. This setup ensures that each question is processed independently, strictly isolating the models from any conversational context or memory of prior prompts.}

Beyond producing qualitative rankings, LLM researchers have established several desiderata for evaluating benchmarking methodologies \citep{polo2024tinybenchmarks, hofmann2025fluid}. We evaluate the LaRT framework in these four key criteria. The first two desiderata address active evaluation, while the final two address static evaluation.

\begin{enumerate}
    \item \textbf{Predictive Power}: Refers to the ability to accurately predict LLM performance on unseen items using a limited set of questions \citep{polo2024tinybenchmarks}. A method with higher predictive power extracts more information from the response data. Held-out predictive power is a standard measure of model fit in statistics \citep{gelman1996posterior}.
    \item \textbf{Item Efficiency}: Measures the number of items required to accurately estimate an LLM's latent ability. IRT with CAT has shown better efficiency in LLM ranking than other methods \citep{zhuang2023efficiently, hofmann2025fluid}.
    \item \textbf{Validity}: Measures the consistency of LLM rankings across different datasets. High consistency is a desirable trait of a robust benchmarking method. IRT-based methods have previously shown strong performance in this area \citep{hofmann2025fluid}.
    \item \textbf{LLM Efficiency}: Measures the number of LLMs required to accurately estimate population parameters for active evaluation. Reducing this number is vital, as full evaluations are resource-intensive, potentially costing thousands of GPU hours for each LLM \citep{liang2023holistic}.
\end{enumerate}

In the real data application, we will not only test on the qualitative properties of our LLM evaluation methods, but compare our proposed LaRT framework to baseline quantitatively in these key metrics valued by LLM practitioners.

\color{black}

\section{Latency-Response Theory Model (LaRT)}
\label{sec:model}
\subsection{Notation}

We introduce some notations used throughout this paper. Denote the response accuracy matrix as $\Rb = (R_{ij})\in \{0,1\}^{N\times J}$, collecting responses of $N$ LLMs to $J$ test questions, where $R_{ij}=1$ or 0 indicates whether LLM $i$ answers question $j$ correctly or not. Denote the CoT length matrix as $\Tb = (T_{ij})\in \NN_+^{N\times J}$, where $T_{ij}$ is a positive number recording the CoT length when LLM $i$ responds to question $j$. Denote the $K\times K$ identity matrix as $\Ib_K$. Standard multivariate Gaussian probability density function of $K$ dimensions is denoted as $\phi_K(\bx)$, and its
\color{black}
cumulative distribution function
\color{black}
is denoted as $\Phi_K(\bx)$. When $K=1$, we omit the subscript $K$. \textcolor{black}{For a set $B\subseteq \RR^K$, let $\Phi_K(B) = P(Z\in B)$, where $Z\sim N(\bm{0}_K, \Ib_K)$.}

\subsection{The LaRT Framework}

We propose a hierarchical LaRT framework that jointly models response accuracy and CoT variables.
In this framework, the binary response accuracy is modeled using an item response model with a probit link, while CoT is modeled through a log-normal distribution with a latent speed variable. The subject-level latent traits underlying these two components are jointly specified by a two-dimensional multivariate normal distribution, where the correlation captures the dependence between latent traits.
These model specifications draw inspiration from hierarchical frameworks in psychometrics that jointly model response accuracy and response time  \citep{vanderlinden2007Hierarchical, entink2009multivariate, wang2015mixture},
while being motivated by the parallel role of CoT processes and response times as indicators of intermediate reasoning.
\begin{subequations}
\label{eq:model}
\begin{align}
R_{ij} &\sim \text{Bernoulli}(\Phi(a_j\theta_i + b_j)),
    \label{eq:model-R} \\
\log T_{ij} &\sim N(\omega_j - \varphi_j\tau_i, \lambda_j),
    \label{eq:model-T} \\
(\theta_i,\tau_i)^\top &\sim N(0,\bSigma),
    \label{eq:model-latent} \\
\bSigma &=
\begin{pmatrix}
1 & \rho\\
\rho & 1
\end{pmatrix},
    \label{eq:model-Sigma}
\end{align}
\end{subequations}

In the IRT model in (\ref{eq:model-R}), $\theta_i$ denotes the latent ability of LLM $i$, representing its position on the underlying proficiency scale. The discrimination parameter $a_j$ describes how strongly item $j$ differentiates between LLMs with different ability levels. The parameter $b_j$ represents the difficulty of item $j$, with smaller values indicating more difficult items, i.e. items that require higher ability for a correct response.
Among the two commonly used link functions for IRT models (the logit link and the probit link), we employ the latter, as its desirable mathematical properties support developing both theoretical identifiability results and computational estimation methods, as detailed later in Sections \ref{sec:algorithm} and \ref{sec:identifiability}.

In (\ref{eq:model-T}), we use a log-normal distribution to model the
length of CoT. Here, $\tau_i$ denotes a latent speed variable representing the CoT-related trait of LLM $i$.
$\omega_j$ represents the CoT-intensity of item $j$, with
larger values indicating items that require more
intermediate reasoning. The coefficient $\varphi_j$ reflects the
CoT-discrimination of item $j$, measuring the sensitivity of CoT lengths to
differences in LLMs' latent speed $\tau_i$. $\lambda_j$ is the
residual variance of the log-CoT ($\log T_{ij}$).
The use of log-normal distribution for CoT length is justified by
both empirical observations and theoretical considerations. Empirically, CoT
lengths produced by LLMs are typically large (often in the range of hundreds to
thousands), and in our applications fewer than 1\% of the CoT are shorter than~10.
Theoretically, the log-normal distribution serves as an asymptotic
approximation to common count distributions, such as the Poisson and negative
binomial, when the counts are large. Thus, for large $T_{ij}$, the log-normal is
a suitable choice for modeling CoT length. Moreover, as shown later, this specification offers mathematical conveniences that, together with the probit link, support rigorous identifiability analysis and
efficient estimation algorithms.

Finally, in Equations (\ref{eq:model-latent})-(\ref{eq:model-Sigma}), the diagonal entries of $\bSigma$ are fixed to $1$ to identify the scale
of the latent traits, and the off-diagonal entry of $\bSigma$ is denoted by
$\rho$, so that $\bSigma$ is essentially a correlation matrix. We provide
theoretical identifiability results in Section~\ref{sec:identifiability}.
For notational simplicity, we denote the collection of population parameters (i.e., parameters not depending on individual LLMs) by
$\bOmega = \{\ba, \bb, \bomega, \bvarphi, \blambda, \rho\}$.

\color{black}
Figure~\ref{fig:LaRT_DAG} gives a graphical model representation of
the LaRT model. The current framework relies on two key conditional independence assumptions. First, the prior distributions of the latent traits $\bxi_i=(\theta_i,\tau_i)$ are assumed to be independent across all evaluated LLMs $i\in[N]$. Second, the observed responses are conditionally independent given the latent traits and item parameters. While these assumptions are standard in latent variable modeling, we explore potential relaxations and extensions to this structure in the discussion.

\begin{figure}[h!]
    \centering
    \begin{minipage}[c]{0.40\linewidth}
        \centering
        \includegraphics[width=\linewidth]{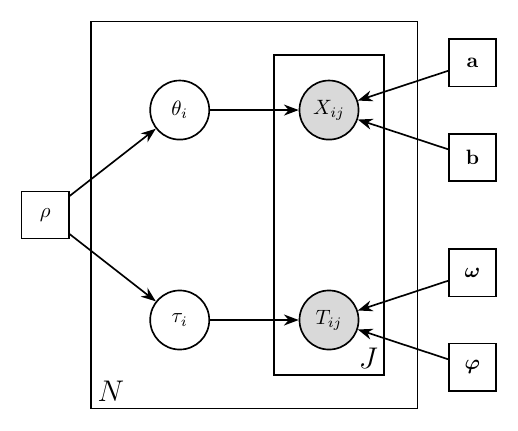}
    \end{minipage}\hfill
    \begin{minipage}[c]{0.55\linewidth}
        \centering
        \begin{tabular}{ll | ll}
            \toprule
            \multicolumn{2}{c|}{\textbf{LLM Parameters}} & \multicolumn{2}{c}{\textbf{Item Parameters}} \\
            \midrule
            $\theta_i$ & Latent ability & $a_j $ & Accuracy Discrim. \\
            $\tau_i$ & Latent speed & $b_j$ & Difficulty level \\
            $\rho$ & Ability-speed cor. & $\varphi_j$ & Speed Discrim. \\
            & & $\omega_j$ & Baseline CoT \\
            & & $\lambda_j$ & Residual variance \\
            \bottomrule
        \end{tabular}
    \end{minipage}

    \vspace{0.5cm}
    \caption{\textcolor{black}{Directed acyclic graph (DAG) illustrating the conditional independence structure of the LaRT model. Unshaded circular nodes represent latent variables, while shaded circular nodes denote observed variables. Square nodes indicate parameters that are treated as fixed within the LaRT framework. The table on the right summarizes the interpretation for the model parameters.}}
    \label{fig:LaRT_DAG}
\end{figure}

A primary advantage of this hierarchical modeling strategy lies in the interpretability and theoretical consistency. LaRT employs a two-level modeling strategy. At the individual level, the measurement model for the response accuracy is mathematically identical to a standard IRT model. The definition of ability is not directly altered by the CoT length. At the population level, LaRT hypothesizes that across the population of LLMs, stronger reasoning ability correlates with longer CoT. We capture this relationship through the population correlation parameter $\rho$. By leveraging the collateral information provided by CoT length, the model achieves a more precise estimation of this underlying ability, a property we detail both theoretically and empirically in subsequent sections.
\color{black}

\section{Estimation Algorithm}
\label{sec:algorithm}
Parameter estimation is crucial for applying the LaRT framework to LLM evaluation.
\textcolor{black}{In the current setting, the observed data likelihood depends on both population parameters and latent traits, with the latter treated as random effects.  A common approach is to first estimate the population parameters $\bOmega$ by maximizing the observed-data likelihood, obtained by integrating out the latent traits $(\btheta, \btau)$ with respect to their assumed distribution.  This step focuses on capturing the measurement characteristics of the items at the population level while accounting for variability in latent traits through marginalization.
Subsequently, given the estimated population parameters, the latent traits of individual LLMs are inferred via MAP estimation.}

The main estimation challenge for the population parameters arises from the intractability of the observed-data likelihood:
\begin{equation}
    \label{eqn:marginal_likelihood}
    \begin{aligned}
    P(\Rb,\Tb;\bOmega) &= \int \prod_{i=1}^N\prod_{j=1}^J \left[\frac{\Phi((2R_{ij}-1)(a_j\theta_i+b_j))}{\sqrt{2\pi\lambda_j}}\exp\left\{- \frac{(\log T_{ij}-\omega_j+\varphi_j\tau_i)^2}{2\lambda_j}\right\}\right]\\
    &\prod_{i=1}^N \frac{\exp\left\{-\bxi_i^\top\bSigma^{-1}\bxi_i/2\right\}}{2\pi(1-\rho^2)^{1/2}}d\btheta d\btau
    \end{aligned}
\end{equation}
\color{black}
Direct likelihood maximization is computationally intractable due to the integral over the latent traits $(\btheta,\btau)$. Although the expectation-maximization (EM) algorithm bypasses this intractable integral with the complete-data log-likelihood, the E-step requires evaluating conditional expectations without closed-form solutions, and remains challenging. To circumvent this, we employ Stochastic Approximation EM \citep[][SAEM]{delyon1999Convergence},
which replaces the intractable E-step with Monte Carlo approximations updated through a Robbins–Monro scheme \citep{robbins1951stochastic}. However, successfully scaling SAEM for the LaRT framework relies heavily on designing an efficient latent-variable sampler to ensure rapid computation and algorithmic stability.
\color{black}

\color{black}
The remainder of this section is organized as follows. First, we propose an efficient SAEM algorithm for estimating the population parameters tailored for LaRT in Subsection~\ref{sec:saem}. Next, in Subsection~\ref{sec:map}, we describe how to estimate the latent variables of LLMs once the population parameters have been obtained, under both static and active evaluation settings. Additionally, in Appendix B of the Supplementary Material \citep{xu2026supplement}, we introduce an effective spectral-based initialization method for the proposed SAEM algorithm.
\color{black}

\subsection{SAEM Algorithm for Population Parameter Estimation}
\label{sec:saem}
The SAEM algorithm addresses the intractable E-step of EM through a stochastic approximation scheme. It alternates three steps: at each iteration, latent variables are sampled in an S-step and used to approximate the conditional expectation in an SA-step, followed by an M-step that updates the population parameters. The SAEM algorithm operates on the following complete-data log-likelihood:
\begin{equation}
\label{eqn:log_complete_post}
    \begin{aligned}
    L(\bOmega,\btheta,\btau\mid \Rb, \Tb) &= \sum_{i,j} \left[ \log \Phi\left((2R_{ij}-1)(a_j\theta_i+b_j)\right) - \frac{1}{2\lambda_j}(\log T_{ij}-\omega_j+\varphi_j\tau_i)^2 \right] \\
    &\quad - N\sum_{j}\log \lambda_j -\frac{N}{2}\log (1-\rho^2) -\frac{1}{2} \sum_{i} \bxi_i^{\top}\bSigma^{-1}\bxi_i.
\end{aligned}
\end{equation}

Next, we introduce the three steps separately and describe
their detailed derivation and implementation.

\textbf{S-step.} The S-step requires drawing samples of the LLM latent variables $(\theta_i, \tau_i)$ from
their posterior distribution given the data and the current population parameters:
\begin{equation}
    \label{eqn:posterior_decom_1}
    \begin{aligned}
        P_{\bOmega}\left(\theta_i,\tau_i\mid \Rb_{i,:}, \Tb_{i,:} \right)
        &= {\biggl[P_{\bOmega}(\theta_i) \prod_{j=1}^JP_{\bOmega}\left(R_{ij}\mid \theta_i\right)\biggr]} {\biggl[ P_{\bOmega} \left(\tau_i\mid \theta_i\right)\prod_{j=1}^J P_{\bOmega}\left(T_{ij}\mid \tau_i\right) \biggr]}.
    \end{aligned}
\end{equation}
Direct joint sampling of $(\theta_i,\tau_i)$ from
(\ref{eqn:posterior_decom_1}) is infeasible. Fortunately, LaRT's structure, specifically the probit link for response accuracy and the log-normal model for CoT, enables an efficient two-step sampling procedure that avoids MCMC-based samplers such as Metropolis–Hastings or Gibbs. In particular, because both the conditional prior $P(\tau_i\mid\theta_i;\bOmega)$ and the
likelihood $P(T_{ij}\mid\tau_i;\bOmega)$ are normal, the conditional posterior
$P(\tau_i\mid\theta_i,\Tb_{i,:};\bOmega)$ also remains normal. Thus, once a sample of
$\theta_i$ is obtained, $\tau_i$ can be drawn directly from a normal distribution.
The remaining task is thus to sample $\theta_i$ from its marginal posterior
$P(\theta_i\mid\Rb_{i,:},\bOmega)$, obtained by integrating out $\tau_i$ in
(\ref{eqn:posterior_decom_1}).
Benefiting from the structure induced by the probit likelihood and Gaussian prior,
this marginal posterior falls within the family of unified skew-normal (SUN) distributions.
Leveraging the analytical characterization developed in \citet{durante2019Conjugate},
we explicitly characterize its posterior distribution, as presented in the following lemma.

\color{black}
\begin{lemma}
    \label{lem:posterior}
    The posterior distribution of $(\theta_i,\tau_i)$ given $\Rb_{i,:}$, $\Tb_{i,:}$, and $\bOmega$ follows the following distribution,
    \begin{align*}
        \theta_i\mid \Rb_{i,:},\Tb_{i,:};\bOmega \sim \mathsf{SUN}_{1,J}\left(\mu_{\theta}^{(i)},\sigma_{\theta}^2, \Delta_{i,\text{post}}, \gamma_{i,\text{post}}, \Gamma_{i,\text{post}}  \right),\quad
        \tau_i\mid \theta_i,\Tb_{i,:}; \bOmega \sim N\left(\mu_{\tau}^{(i)}, \sigma_{\tau}^{2}\right),
    \end{align*}
    where $\mathsf{SUN}$ represents the unified skew-normal distribution, and the specification of each term is deferred to Appendix A.1 of the Supplementary Material \citep{xu2026supplement}.
\end{lemma}
\color{black}
\color{black}
Lemma~\ref{lem:posterior} provides the explicit distributional expressions discussed
above for sampling $\theta_i$ and $\tau_i$. Furthermore, the SUN
distribution can be sampled efficiently, which can be
implemented via a linear combination of samples from a multivariate normal distribution and a truncated normal distribution. We provide the corollary from \citet{li2025Sparse} in Appendix A.3 of the Supplementary Material \citep{xu2026supplement} for completeness.

Taken together, these results yield a two-step sampler for $(\theta_i,\tau_i)$. First, $\theta_i$ is drawn from the SUN distribution in Lemma~\ref{lem:posterior}. Second, $\tau_i$ is then sampled from the normal distribution in Lemma~\ref{lem:posterior}. Since all target distributions admit straightforward sampling, the S-step can be implemented efficiently.
\color{black}

\textbf{SA-step.}
The SA-step updates the current estimate of the conditional expectation of the log complete posterior  using a stochastic
approximation scheme. At each iteration $t$, it first uses the $C$ samples of
$(\theta_i, \tau_i)$, denoted by $\bxi_i^{(c)}=(\theta_i^{(c)}, \tau_i^{(c)})$,
$c \in [C]$, obtained in the S-step to form a Monte Carlo approximation of this
conditional expectation:
\begin{equation}
    \label{eqn:obj_new}
    \begin{aligned}
    Q_t^{\text{(new)}}(\bOmega) &= \frac{1}{C}\sum_{c,i,j} \left[ \log \Phi\!\left((2R_{ij}-1)(a_j\theta_i^{(c)} +b_j)\right) - \frac{\bigl(\log T_{ij}-\omega_j+\varphi_j\tau_i^{(c)}\bigr)^2}{2\lambda_j} \right] \\
    &\quad - N\sum_{j}\log \lambda_j -\frac{N}{2}\log (1-\rho^2) -\frac{1}{2C}\sum_{c,i} \bxi_i^{(c)\top}\bSigma^{-1}\bxi_i^{(c)}.
\end{aligned}
\end{equation}
The constant $C$ denotes the number of Monte Carlo samples drawn at each SA-step. Although
a larger $C$ can reduce Monte Carlo variability, it also increases computational cost.
Following standard SAEM practice, we set $C=1$, which is sufficient for stable estimation
in our setting.
Then, this Monte Carlo estimate is incorporated into the running approximation of the
conditional expectation through a Robbins–Monro stochastic approximation update:
\[
    \tilde Q_t(\bOmega)
    = (1-\alpha_t)\,\tilde Q_{t-1}(\bOmega)
      + \alpha_t \, Q_t^{\text{(new)}}(\bOmega),
\]
with initialization $\tilde Q_1(\bOmega)=Q_1^{\text{(new)}}(\bOmega)$.
As shown, the new estimate $\tilde Q_t(\bOmega)$ is a weighted mixture of the current
Monte Carlo estimate $Q_t^{\text{(new)}}(\bOmega)$ and the previous estimate
$\tilde Q_{t-1}(\bOmega)$. The sequence $\{\alpha_t\}$ denotes the step sizes and is required to satisfy the
Robbins–Monro conditions $\sum_{t}\alpha_t=\infty$ and $\sum_{t}\alpha_t^2<\infty$ to
ensure convergence. A typical choice that meets these requirements is $\alpha_t = 1/t$.
Through this averaging, the stochastic approximation update stabilizes the sequence of
expectation estimates by weighting new updates $Q_t^{\text{(new)}}(\bOmega)$ with decreasing step sizes, allowing the
algorithm to converge even when only a small number of samples is used at each step
\citep{robbins1951stochastic,delyon1999Convergence}.

\textbf{M-step.} In the M-step, the population parameters are updated by maximizing the current
$\tilde Q_t(\bOmega)$. Importantly, the concavity of $Q_t^{\text{(new)}}$ in $\bOmega$ is
preserved under the Robbins–Monro averaging step, ensuring that $Q_t$ remains concave and
can therefore be maximized efficiently. To this end, we employ the L-BFGS algorithm
\citep{liu1989limited}, a widely used and efficient quasi-Newton method for smooth convex
optimization, to perform the M-step at each iteration. The complete SAEM procedure is
outlined in Algorithm \ref{alg:SAEM}.

\begin{algorithm}[h!]
    \caption{SAEM algorithm for the estimation of population parameters.}
    \label{alg:SAEM}
    \begin{algorithmic}[1]
        \Require Binary response matrix $\Rb\in \{0,1\}^{N\times J}$, CoT length matrix $\Tb \in \RR_{+}^{N\times J} $, Initialization $\hat{\bOmega}_0$, Number of Monte Carlo samples $C$, stochastic approximation weights $\{\alpha_t\}_{t\in \NN}$.
        \Ensure Estimated $\hat{\bOmega}$.
        \State Initialize $t\gets 0$, $Q_0\gets 0$.
        \While{not converge}
            \State $t\gets t+1$.
            \State Draw $C$ samples of $(\btheta^{(c)}, \btau^{(c)})$ following Lemma \ref{lem:posterior} and the corollary in Appendix A.3 of the Supplementary Material \citep{xu2026supplement}.
            \State Compute $Q_t^{\text{(new)}}$ by (\ref{eqn:obj_new}), and $\tilde{Q}_{t}\gets (1-\alpha_t)\tilde{Q}_{t-1}+\alpha_tQ_t^{\text{(new)}} $.
            \State $\hat{\bOmega}_t\gets \argmax_{\bOmega}\tilde{Q}_t(\bOmega)$ with a valid convex optimization algorithm.
        \EndWhile
        \State Return $\hat{\bOmega}_t$.
    \end{algorithmic}
\end{algorithm}

\subsection{Maximum-a-posterior estimation for individual latent ability and speed}
\label{sec:map}
After obtaining the estimates of the population parameters, we estimate the individual latent variables $(\btheta,\btau)$ based on maximum a posteriori (MAP) estimation:
\begin{equation}
    \label{eqn:res_indi_param}
    (\btheta,\btau) = \argmax_{(\btheta,\btau)}L(\hat{\bOmega}, \btheta, \btau \mid \Rb, \Tb).
\end{equation}
Given the population parameters $\bOmega$, the log complete posterior
(\ref{eqn:log_complete_post}) is concave with respect to
$\bxi=(\btheta,\btau)$. This observation, together with the independence across the
$\bxi_i$'s, reduces the optimization problem to $N$ convex optimization problems with
two-dimensional parameters, each of which can be solved efficiently using L-BFGS.

Finally, the proposed estimation procedure for population parameters and  individual latent traits applies to both static and active evaluation settings.  In static evaluation, the population parameters are first estimated by applying SAEM to  fit a LaRT model to the LLM response data, after which the MAP estimates of the LLMs'  latent traits are obtained by solving (\ref{eqn:res_indi_param}). In active  evaluation, LLMs are administered a set of items drawn from a calibrated item pool, where  the items can be calibrated using an appropriate estimator such as the proposed SAEM.  As each LLM answers items, we update its latent trait estimates by solving (\ref{eqn:res_indi_param}) with the current response data.

\section{Theoretical Results}
\label{sec:theory}
This section presents theoretical results for LaRT. We do not make explicit model assumptions about the item parameters for the questions; instead, we treat them as fixed, unknown parameters, as in most existing IRT literature.

Next, we first prove identifiability of the model parameters to ensure statistical validity of inferences about LLM latent traits. Second, we derive asymptotic distributions that characterize estimation precision and reveal when LaRT achieves smaller asymptotic variances than standalone IRT.

\subsection{Identifiability}
\label{sec:identifiability}
We next present rigorous identifiability result for the population parameters in LaRT. We begin by defining identifiability for LaRT.
\begin{definition}
    \label{def:identifiability}
    The LaRT model is identifiable at $\bOmega$ if, for any other set of parameters $\bOmega^{\prime}$ that gives rise to the same marginal distribution of the observed data, i.e., $P(\Rb, \Tb; \bOmega) = P(\Rb, \Tb; \bOmega')$, then $\bOmega = \bOmega^{\prime}$ must hold.
\end{definition}

The probit link for modeling the response accuracy in LaRT enables rigorous identifiability analyses and guarantees for LaRT. For single-modal IRT model with a probit link, \citet{fang2021identifiability} build their identifiability analysis on a key proposition showing that establishing identifiability reduces to verifying whether the probit thresholds and tetrachoric correlations admit a unique parameterization.
While their results do not directly apply to LaRT's hierarchical bimodal structure, they inspire our approach.
\color{black}
Specifically, by leveraging both the probit and log-normal link functions, we can characterize the joint marginal distribution of $(\Rb,\Tb)$ entirely through its means, cross-correlations, and tetrachoric correlations. This analytical tractability allows us to distill the full marginal distribution of $(\Rb,\Tb)$ conditional on $\bOmega$ into a concrete system of algebraic equalities. The precise mathematical formulation of this characterization is detailed in Appendix C.1 of the Supplementary Material \citep{xu2026supplement}.
\color{black}

Importantly, checking identifiability of LaRT via the equalities in Proposition 1 of the Supplementary Material \citep{xu2026supplement} parallels the proof strategy used for linear factor models. Let $\tilde{\Ab}$ denote the discrimination matrix (factor loadings) and $\tilde{\bSigma}$ the covariance matrix of the latent variables. In that setting, identifiability reduces to verifying the equality
$\tilde{\Ab}\tilde{\Ab}^{\top}+\tilde{\bSigma}
= \tilde{\Ab}^{\prime}\tilde{\Ab}^{\prime\top}+\tilde{\bSigma}^{\prime}$.
\citet{anderson1956statistical} show the identifiability of linear factor models holds under mild conditions.
Building on this analogy and the foundational results of \citet{anderson1956statistical},
we establish identifiability for LaRT in Theorem~\ref{thm:identifiability}.

\begin{theorem}
    \label{thm:identifiability}
    If (1) there are at least 3 non-zero entries in both $\ba$ and $\bvarphi$, (2) $\sum_{j=1}^Ja_{j}>0$, $\sum_{j=1}^J \varphi_{j} >0$, then LaRT is identifiable.
\end{theorem}

Condition 1 is standard in the identifiability analysis of generalized linear factor models.
For a general $K$-dimensional setting with factor loading matrix
$\tilde{\Ab}=[\ba_1,\ldots,\ba_J]^\top \in \RR^{J\times K}$, it requires that, after deleting any row of $\tilde{\Ab}$, the remaining matrix still has two rank $K$ submatrices.
In our evaluation context, this condition implies that there must be at least three items that differentiate the latent ability and at least three items that differentiate the latent speed of the LLMs, which is a very mild requirement.

Condition 1 ensures LaRT identifiability only up to a sign indeterminacy arising from the bilinear terms $a_j\theta_i$ and $\varphi_j\tau_i$. To resolve this, we impose the additional constraints $\sum_{j} a_j > 0$ and $\sum_{j}\varphi_j > 0$.
These constraints enforce interpretable parameter orientations: a higher latent ability corresponds to a higher overall probability of correct responses, and a higher latent speed corresponds to faster responses. In addition, requiring only the sums of $a_j$ and $\varphi_j$ to be positive is a weaker restriction than the common psychometric assumption that all item discriminations are strictly positive \citep{hambleton2013item}.
The proof of Theorem~\ref{thm:identifiability} is in Appendix C.3 of the Supplementary Material \citep{xu2026supplement}.

\subsection{Asymptotic Posterior Normality}
We now establish the asymptotic normality of the individual latent-variable estimators $(\theta_i,\tau_i)$. The asymptotic variances can provide guidance on the estimation precision of the latent-variable estimates as the number of items increases, which is particularly valuable in high-stakes large-scale evaluations such as LLM evaluation, where ranking outcomes can influence user adoption and system deployment decisions.
Moreover, as shown later in this section, they offer insight into when LaRT can outperform a standalone IRT model by achieving smaller asymptotic variances.

For general IRT models, \citet{chang1993Asymptotic,kornely2022Asymptotic} present results on the asymptotic distribution of latent ability. In what follows, we derive the asymptotic distribution of the latent variables $(\theta_i,\tau_i)$ in LaRT, with additional considerations arising from the CoT component of the model. Since $(\theta_i,\tau_i)$ are independent across LLMs conditional on the population parameters, we drop the subscript $i$ and use the generic notation $\bxi = (\theta,\tau)$. We further denote its parameter space by $\Theta$ and assume $\Theta = \bar{\RR} \times \bar{\RR}$, which is closed and convex. In the following, we write the log posterior of LaRT as
\begin{align*}
    \tilde{l}^{(J)}(\theta,\tau\mid \Rb^{(J)}, \Tb^{(J)};\bOmega) &= \sum_{j=1}^J\left[\log P(R_j^{(J)}\mid \theta ;\bOmega) + \log P(T_j^{(J)}\mid \tau;\bOmega) \right]-\frac{1}{2}\bxi^\top \bSigma^{-1} \bxi.
\end{align*}
The corresponding Fisher information matrices for the likelihood and posterior are denoted by $\tilde{\cI}_J(\bxi) = -\EEE_{\bxi}\bigl[\nabla_{\bxi}^2\tilde{l}^{(J)}(\bxi\mid \Rb^{(J)}, \Tb^{(J)};\bOmega)\bigr]$.

We require three regularity conditions on the response-time component for establishing the asymptotic distribution, following \citet{chang1993Asymptotic,kornely2022Asymptotic}. \textcolor{black}{We defer the detailed characterization of the assumptions and discussions to Appendix D.1 of the Supplementary Material \citep{xu2026supplement}.}

\begin{theorem}
    \label{thm:asymptotic}
    Let $\bz\sim N_2(0, \Ib_2)$, $\{R_j,T_j\}_{j\in\NN} \sim \cP(\bxi_0)$, $\bxi_0\in \Theta\setminus \partial \Theta$. Under Assumptions D.1--D.3 in the Supplementary Material \citep{xu2026supplement}, for all $B\in \cB(\Theta)$, for a fixed $\bxi_0$,
    \begin{equation*}
        P\left(\tilde{\cI}_J(\tilde{\bxi}_J)^{1/2} (\bxi-\tilde{\bxi}_J)\in B\mid \Rb^{(J)}, \Tb^{(J)} \right) \stackrel{P_{\bxi_0}}{\to} P(\bz\in B),\quad J\to \infty.
    \end{equation*}
    If $\bxi_0\sim \cG$, where $\cG$ is an absolutely continuous proper distribution whose support is within $\Theta$,
    \begin{equation*}
        P\left(\tilde{\cI}_J(\tilde{\bxi}_J)^{1/2} (\bxi-\tilde{\bxi}_J)\in B\mid \Rb^{(J)}, \Tb^{(J)} \right) \stackrel{p}{\to} P(\bz\in B),\quad J\to\infty.
    \end{equation*}
\end{theorem}

The proof of Theorem~\ref{thm:asymptotic} is in Appendix D of the Supplementary Material \citep{xu2026supplement}.
To further analyze the factors contributing to the variance of $\theta$, we present the explicit expression for its inverse variance, i.e., the precision $\tilde{\cI}_J(\theta)$, as follows:
\begin{equation}
    \label{eqn:fisher_info_theta}
    \tilde{\cI}_J(\theta) = \frac{1}{1-\rho^2} + \sum_{j=1}^J \frac{a_j^2\phi(a_j\theta_i+b_j)^2}{\Phi(a_j\theta_i+b_j)[1-\Phi(a_j\theta_i+b_j)]}.
\end{equation}
Here, the second term corresponds to the Fisher information, while the first term depends on the correlation between the latent ability and latent speed. In particular, $\tilde{\cI}_J(\theta)$ increases as $|\rho|$ increases, implying that the variance decreases as $|\rho|$ increases and
is maximized at $\rho=0$, which corresponds to the IRT case where CoT information is ignored. Consequently, the asymptotic estimation precision of $\theta$ under LaRT exceeds that of IRT whenever the latent ability and latent speed are correlated.

\section{Simulation Study}
\label{sec:simulation}
We conduct a simulation study that emulates the characteristics of the real-world application. These simulations serve two primary objectives. First, we validate the performance of our proposed SAEM algorithm with the data-driven initialization and the MAP estimate of latent traits. Second, we demonstrate the superiority of LaRT over standard IRT in terms of finite-sample estimation accuracy of the latent ability. For the IRT baseline, we estimate the population parameters $\boldsymbol{a}$ and $\boldsymbol{b}$ using the SAEM approach described by \citet{li2025Sparse}, and estimate the latent ability $\boldsymbol{\theta}$ using the similar MAP method detailed in Section \ref{sec:map}. Unless otherwise noted, all subsequent implementations of IRT follow this procedure.

The simulation design is as follows.
Each entry of $\ba$ is drawn from $\text{Unif}(0.5,1)$, $\bb$ from $N(0,0.5)$, $\bomega$ from $N(0,1)$, $\bvarphi$ from $\text{Unif}(0.5,1.5)$, $\blambda$ from $\text{Unif}(0.5,2)$. We set $\rho=-0.8$. We fix $J=50$, and let $N\in\{100,200,500\}$. For each simulation setting, we perform $200$ independent replications.

The simulation results are presented in Figure \ref{fig:com_irt_lart_8} and \ref{fig:lart_8}. Figure \ref{fig:com_irt_lart_8} shows that LaRT achieves better estimation accuracy than IRT for all $\btheta$, $\ba$, and $\bb$. For other parameters $\bvarphi$, $\bomega$, $\blambda$, and $\rho$, the simulation results confirm that as $N$ increases, the estimation error decreases.

\begin{figure}[h!]
    \centering
    \includegraphics[width=\linewidth]{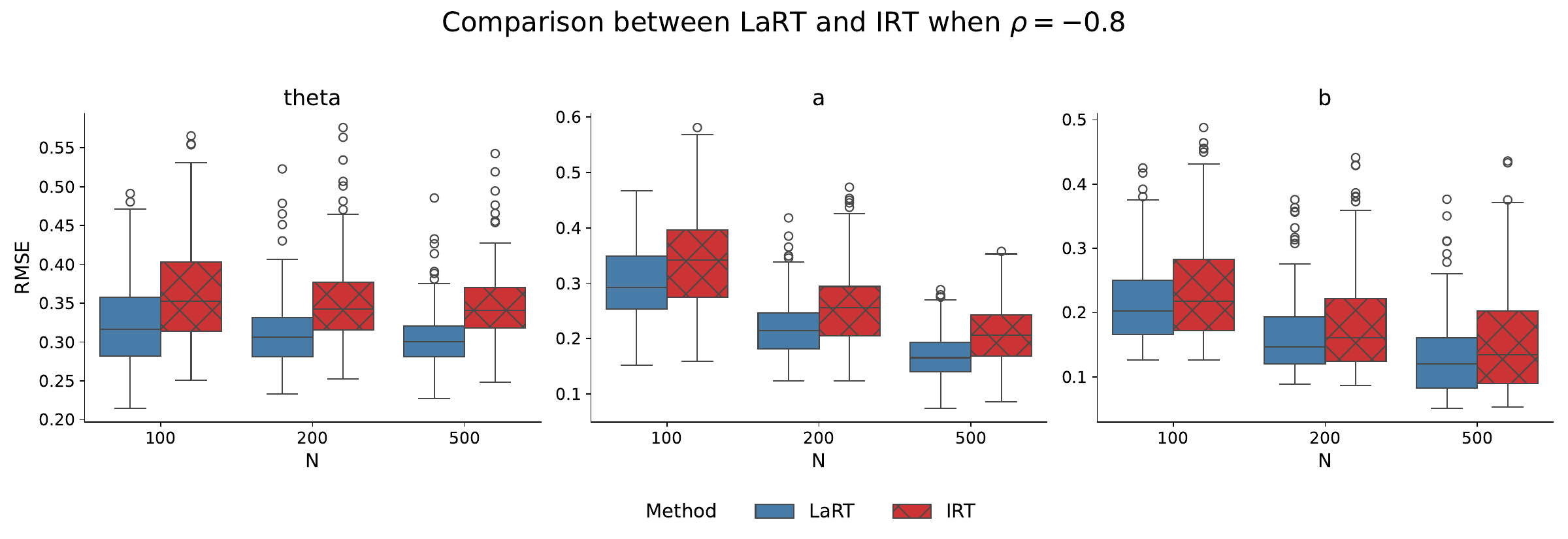}
    \caption{RMSEs of IRT and LaRT when $\rho=-0.8$. LaRT performs uniformly better than IRT. As $N$ grows, RMSE of $\hat{\ba}$ and $\hat{\bb}$ decreases when $J$ is fixed.}
    \label{fig:com_irt_lart_8}
\end{figure}

\begin{figure}[h!]
    \centering
    \includegraphics[width=0.6\linewidth]{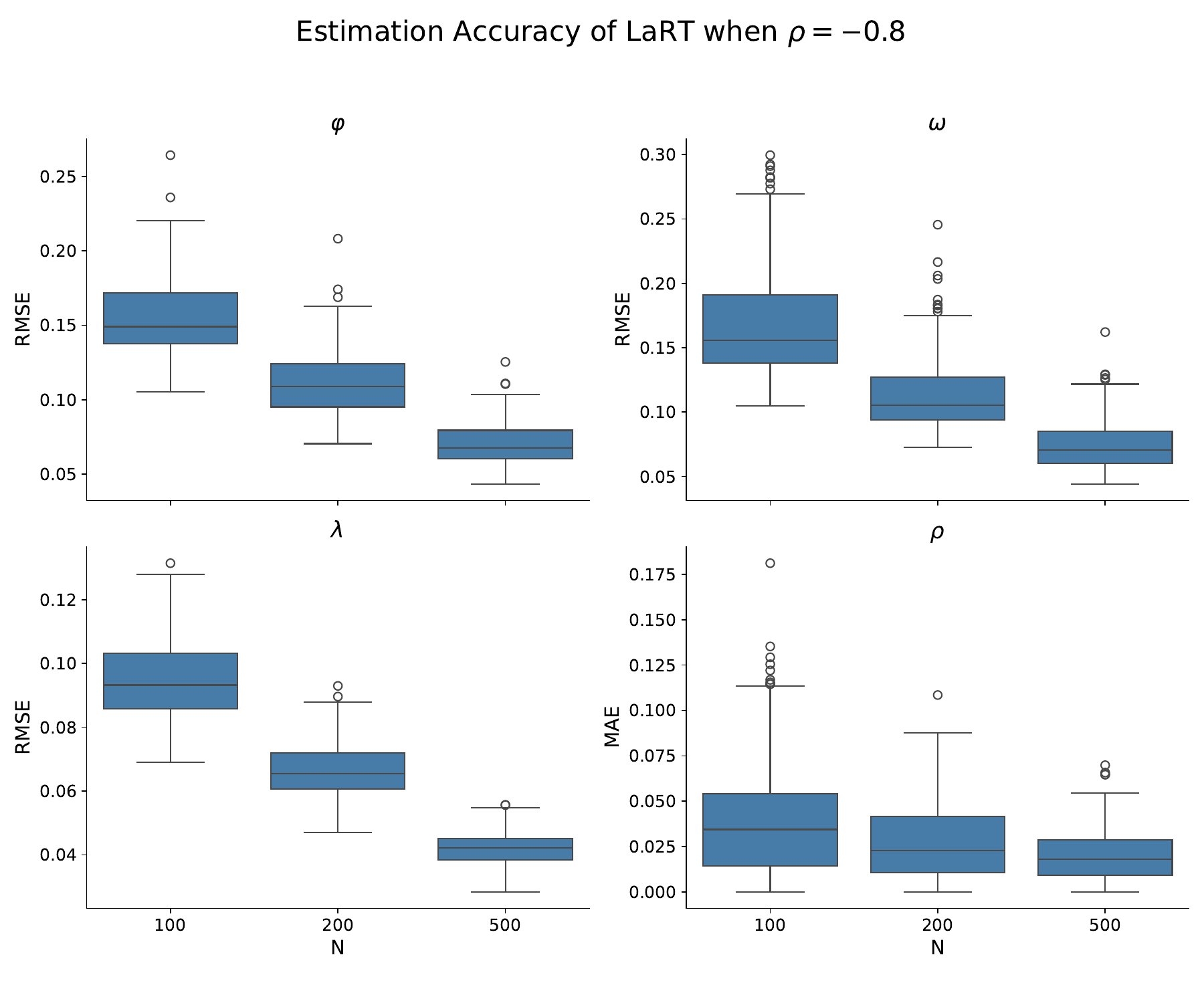}
    \caption{Boxplot of estimation accuracy of other parameters of LaRT when $\rho=-0.8$ in RMSE and MAE. The metric is presented in the plot. As $N$ grows, the estimation error of all parameters decreases.}
    \label{fig:lart_8}
\end{figure}

\section{Applications to Real LLM Data}
\label{sec:application}
\color{black}
Before presenting the real-data analysis, we first establish our operational definition of Chain-of-Thought (CoT) length. Throughout this section, CoT length is calculated as the \textit{total number of tokens in the LLM's generated reasoning sequence prior to producing the final answer}. The LLMs evaluated and benchmark datasets are presented in Section~\ref{sec:motivation_desiderata}.
\color{black}

When collecting response data from LLMs, we set the maximum CoT length to 10,240 tokens. This limit is generous for mathematical reasoning, as ground truth solutions in the MATH dataset rarely exceed 1,000 tokens \citep{hendrycksmath2021}. By setting the limit to 10,240, we allow LLMs to utilize ten times the token budget of human references. Furthermore, this setting aligns with the ``medium-to-high'' reasoning regime defined by \citet{agarwal2025gpt}, who evaluated gpt-oss across limits of 1,000, 6,000, and 16,000 tokens. Detailed hyperparameter configurations are in Appendix F of the Supplementary Material \citep{xu2026supplement}.
After obtaining LLMs' responses to each math problem, we use each LLM's tokenizer to count the length of their CoT. We delete LLMs who failed to answer any question correctly in these four benchmarks. After data preprocessing, there are 138 LLMs for further evaluation.

In this section, we present a comprehensive analysis that jointly examines LLM performance, benchmark characteristics, and the effectiveness of LaRT across the math datasets introduced above.  First, we visualize evaluation outcomes based on estimated latent ability and latent speed across benchmarks, providing insights into LLM ranking, and the relationship between latent ability and latent speed. Second, leveraging the estimated population parameters, we investigate general trends in item characteristics within and across benchmarks, and illustrate how LaRT improves discriminative utility beyond accuracy-based evaluation. Third, we compare the LLM rankings produced by IRT and LaRT: we not only illustrate representative ranking shifts and explain their causes, but also conduct comprehensive quantitative experiments to demonstrate that LaRT offers more convincing evaluation results with respect to the desiderata of predictive power, item efficiency, validity, and LLM efficiency.

\subsection{Qualitative Evaluation}
\label{sec:general_eval}

Figure \ref{fig:speed_ability} presents the scatter plot of estimated latent speed against latent ability, as well as the estimated correlation $\rho$ between the latent ability and latent speed. Across all datasets, $\rho$ is strongly negative, confirming that LLMs with stronger math reasoning ability have longer CoT, a highly intuitive result. Moreover, the absolute value of $\rho$ increases as the questions in the dataset becomes more difficult, indicating that harder questions require more test-time compute. As $|\rho|$ gets larger, the CoT modeling contributes more to the estimation of the latent ability for each LLM. Thus, this phenomenon implies LaRT may offer greater advantages on more challenging benchmarks.

\begin{figure}[h!]
    \centering
    \includegraphics[width=0.62\linewidth]{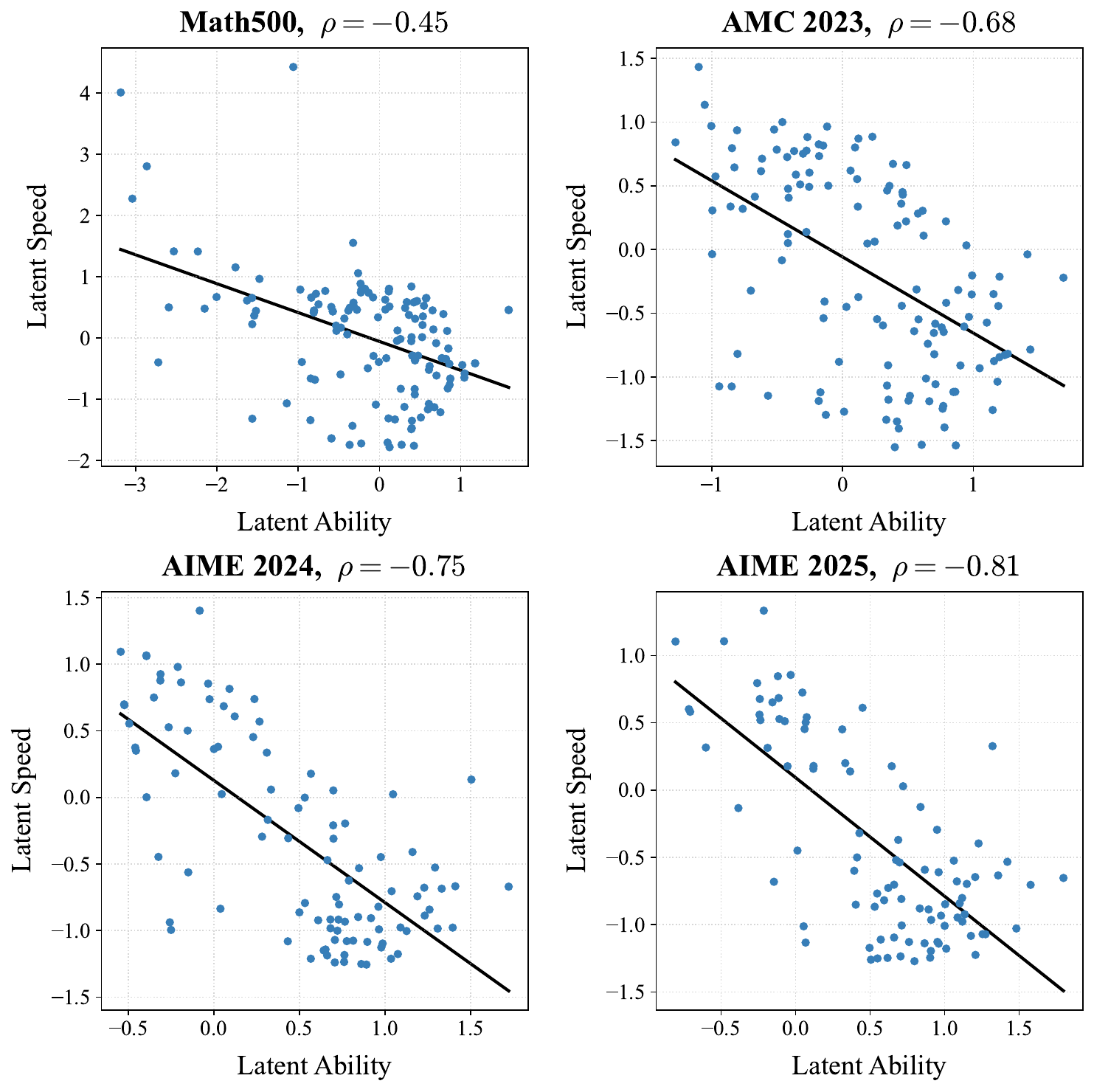}
    \caption{Scatter plot of estimated latent speed against latent ability for each dataset. The estimated correlation for each dataset is presented in the subtitle. LLMs with stronger latent ability have smaller latent speed (longer CoT length). As the dataset becomes more difficult, the estimated correlation $\rho$ increases in absolute value.}
    \label{fig:speed_ability}
\end{figure}

For the estimated population parameters, we present the boxplots of estimated parameters for accuracy and latency respectively in Figure \ref{fig:difficulties} and \ref{fig:latency}. Specifically, we transforms $\bb$ into $\tilde{\bb}=-\bb/\ba$. In this way, the probability of LLM $i$ answering item $j$ correctly is $\Phi(a_j(\theta_i-\tilde{b}_j))$, and $\tilde{b}_j$ can be interpreted as the difficulty level of item $j$. Similarly, we transforms $\bomega$ into $\tilde{\bomega}=\bomega/\bvarphi$. Since there are items whose discrimination $\ba$ is close to 0 or negative, we delete all entries whose discrimination $a_j$ is less than $0.05$ to better present the boxplot of $\tilde{\bb}$.

For $\ba$ and $\bb$, as the benchmark dataset becomes more difficult, the estimated difficulty parameters increase monotonically. Given its relatively large size, MATH500 captures a broad spectrum of difficulty, including some highly challenging items. This shows the validity of estimated population parameters by LaRT. For the discriminative parameters $\ba$, as the dataset becomes more difficult, the discriminative power increases, except AIME25. Since we are testing relatively small LLMs, AIME25 can be too difficult for most of the tested LLMs. Thus, in general, items in AIME25 have less discriminative power in comparison with the easier datasets. Additionally, there are a small proportion of items in these 4 datasets that have close to 0 or negative discriminative power. This is a common phenomenon in LLM benchmarking datasets. This can result from the wrong answers of question items, or the wrong grading \citep{gema2025we}. Through learning the discriminative parameters of question items, LaRT can automatically adjust for this issue. Furthermore, from Theorem \ref{thm:identifiability}, the estimated parameters satisfy the identifiability conditions. Our mild identifiability conditions do not require that every item should have a positive discriminative parameter, which adapts well to the misgrading issue in LLM ranking.

For $\bvarphi$, $\tilde{\bomega}$, and $\blambda$, all the estimated parameters are strictly positive. The identifiability conditions in Theorem \ref{thm:identifiability} are satisfied. For the basic latency level, as the difficulty of the dataset increases, the base CoT length increases, except for AIME25. This may result from the same issue as explained in the previous paragraph about the discriminative parameters.
Since AIME25 is too difficult for most of the tested LLMs, some LLMs fail to provide a valid CoT, and indeed answers the question incorrectly with a shorter CoT length. For the discriminative parameter of latency, there is no significant trend as the benchmark dataset becomes more difficult.

Figure \ref{fig:app_dis} presents the scatterplot of estimated discrimination of response accuracy $a$ against the discrimination of latent speed $\varphi$. We only present the question items whose discriminative power is less than 0.5. Several items exhibit minimal discriminative power for latent ability yet possess significant discriminative power for latent speed. Conceptually, extremely easy or difficult items yield uniform response accuracy, resulting in a small value of $a$, because few LLMs answer them correctly/wrong. However, joint modeling of accuracy and CoT length allows us to additionally evaluate CoT quality on these items. Consequently, the overall discriminative utility of questions with low discriminative power in accuracy is enhanced.

\begin{figure}[h!]
    \centering
    \includegraphics[width=0.6\linewidth]{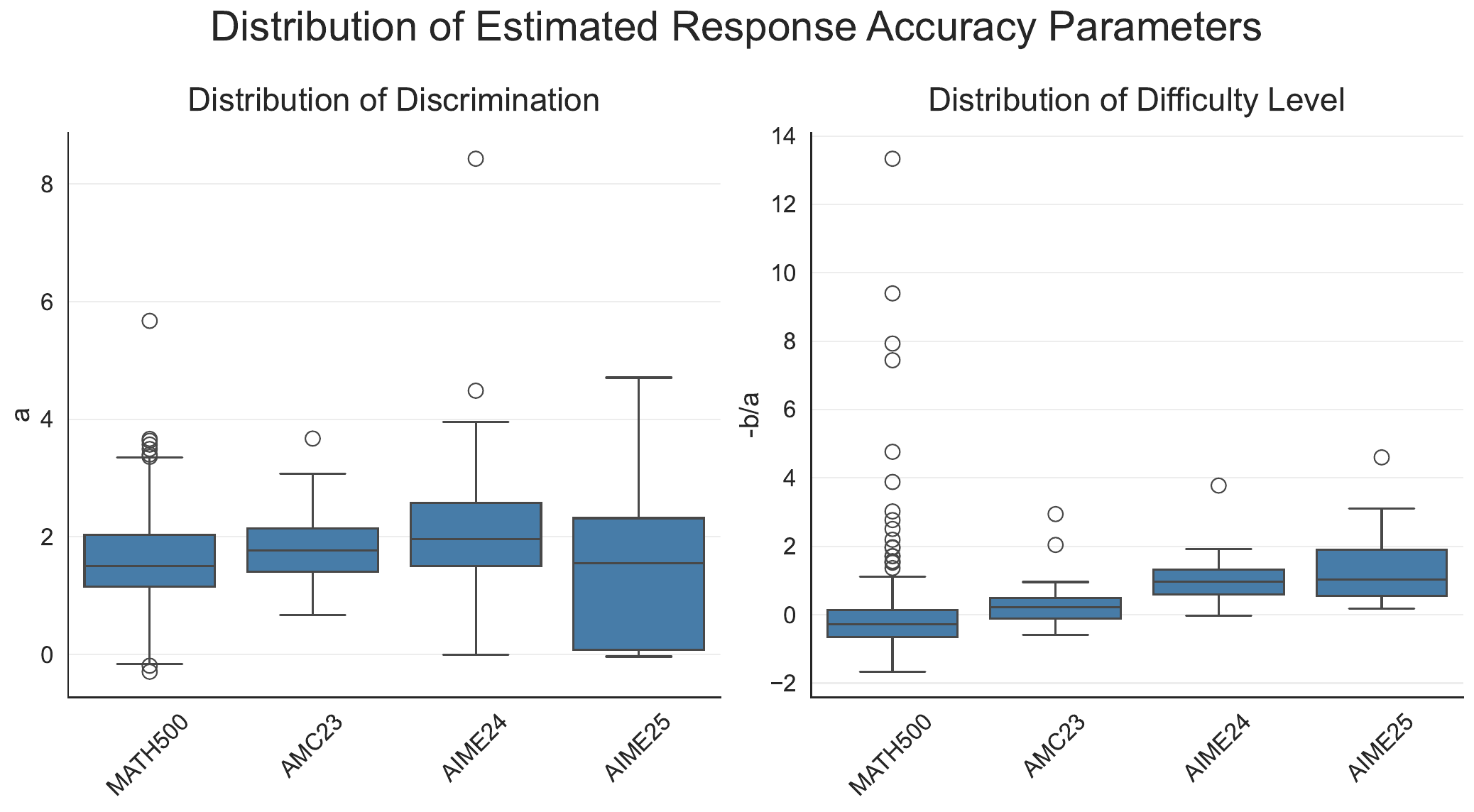}
    \caption{Boxplots of $\ba$ and $\bb$ for 4 datasets estimated by LaRT. The plot on the left is the boxplot of $\ba$. The plot on the right is the boxplot of $-\bb/\ba$. Except the most difficult AIME25, as the dataset becomes more difficult, the average discriminative power of accuracy increases.}
    \label{fig:difficulties}
\end{figure}

\begin{figure}[h!]
    \centering
    \includegraphics[width=0.9\linewidth]{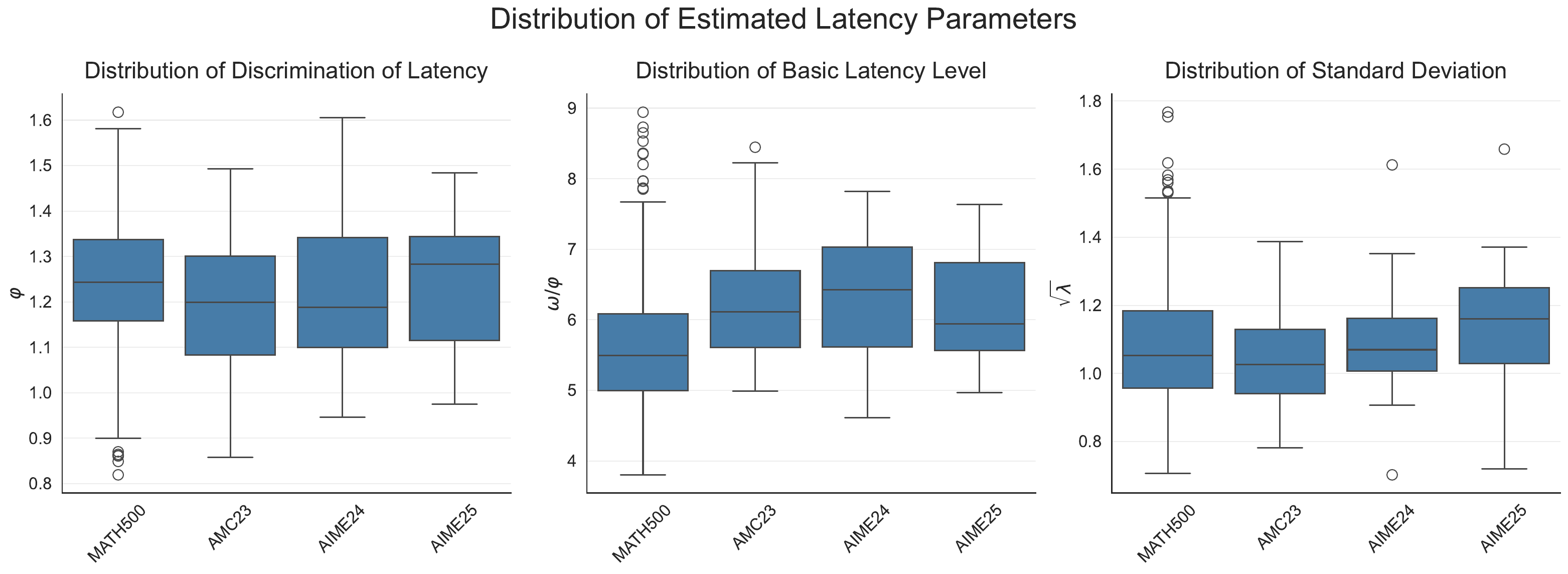}
    \caption{Boxplots of $\bvarphi$, $\bomega$, and $\blambda$ for 4 datasets estimated by LaRT. The left plot is the distribution of $\bvarphi$. The middle plot is the boxplot of $\bomega/\bvarphi$. The right plot is the boxplot of $\blambda$. Except AIME25, as the dataset becomes more difficult, the basic latency level increases.}
    \label{fig:latency}
\end{figure}

\begin{figure}[h!]
    \centering
    \includegraphics[width=0.5\linewidth]{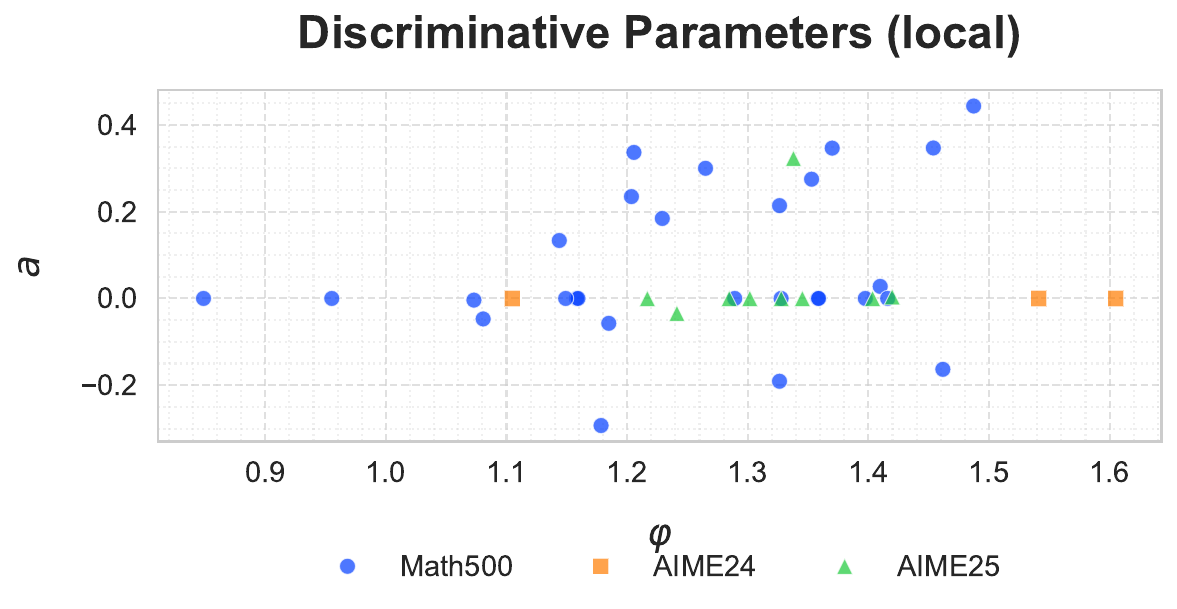}
    \caption{
    Scatterplot of estimated discrimination of the latent ability $a$ against the discrimination of latent speed $\varphi$ for questions with small discriminative ability. The colors of the points represent different datasets as shown in the legend. There are many questions with minimal discriminative power in accuracy, but significant positive discriminative power in CoT length.
    }
    \label{fig:app_dis}
\end{figure}

To give a direct comparison between the rankings given by IRT and LaRT, Figure~\ref{fig:ranking_shift_aime25} presents LLMs whose rankings given by these two models differ. Ranking differences for other datasets are presented in Appendix G of the Supplementary Material \citep{xu2026supplement}.
Figure~\ref{fig:ranking_shift_aime25} shows that many LLMs' rankings differ between IRT and LaRT. The difference results from LLMs that correctly answer similar numbers of items but differ in CoT length. In general, LLMs with longer CoT lengths have higher rankings under LaRT than under IRT.

Specifically, we look at the shift in rankings in Qwen3-30B-A3B-Thinking-2507 (shorthand as Qwen-Thinking), Qwen3-30B-A3B-Instruct-2507 (shorthand as Qwen-Instruct), and Qwen3-32B in the zero-shot prompting. Qwen-Thinking, Qwen-Instruct, and Qwen3-32B correctly answer 15, 14, and 14 questions, respectively, and have average CoT lengths of 5168, 5922, and 8333, respectively. In the ranking by IRT, Qwen-Thinking is better than Qwen-Instruct, and Qwen-Instruct is better than Qwen3-32B. In the ranking by LaRT, Qwen-Instruct is better than Qwen-Thinking, and Qwen-Thinking is better than Qwen3-32B. The larger test-time compute for Qwen-Instruct helps it surpass Qwen-Thinking in the ranking by LaRT. However, even though the CoT length for Qwen3-32B is significantly larger than both Qwen-Instruct and Qwen-Thinking, the ranking of Qwen3-32B remains the same. In particular, Qwen3-32B answers the same number of items correctly as Qwen-Instruct.
This illustrates how LaRT's probabilistic modeling automatically balances difficulty, time requirements, and discrimination in both accuracy and latency. We demonstrate quantitatively that LaRT produces more convincing rankings in the next subsection.

\begin{figure}[h!]
    \centering
    \includegraphics[width=\linewidth]{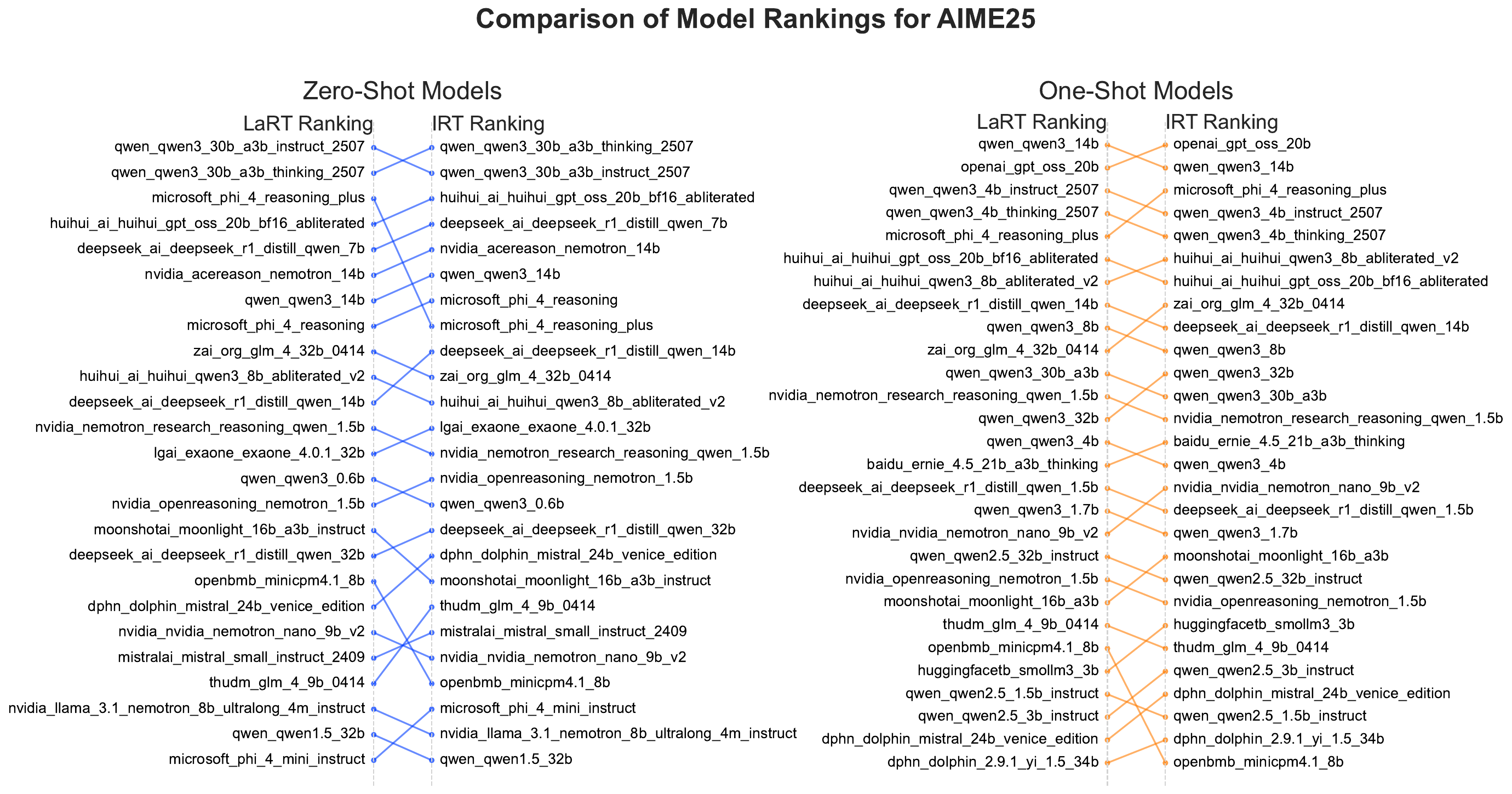}
    \caption{Differences in LLM rankings for both zero-shot models and one-shot models for AIME25. The left panel is for zero-shot models, and the right panel is for one-shot models. For each panel, rankings by LaRT are on the left, and rankings by IRT are on the right. LLMs higher in the plot have higher rankings. The lines connect the same models with different rankings by LaRT and IRT.}
    \label{fig:ranking_shift_aime25}
\end{figure}

\subsection{Quantitative Comparison}
\label{sec:quantitative_compare}
\subsubsection{Predictive Power}
\label{sec:predictive_power}
We test the predictive power of LaRT in comparison with IRT under the active evaluation setting. Prior works on predictive accuracy also studied the active evaluation regime \citep{polo2024tinybenchmarks}. We concatenate AMC23, AIME24, and AIME25 datasets together. We delete LLMs that fail to answer any of the questions correctly. After preprocessing, there are $N=128$ LLMs and $J=100$ question items. We randomly select $N_t=100$ LLMs as training set to estimate the population parameters in LaRT. For the test set, we examine the prediction performance via five-fold cross-validation. We randomly partition the $J=100$ questions into five disjoint sets. For each fold, we first obtain the LLMs' latent abilities with four sets of questions. Then, we predict the responses of each LLM for the questions in the remaining test set. Since both the correct/wrong responses and predicted probability of each item are within $[0,1]$, we use mean absolute error as the metric for the prediction accuracy.

The result for each fold and average mean absolute error is presented in Table \ref{tab:prediction}. LaRT significantly outperforms IRT in all folds. This illustrates that LaRT is able to extract more information from the data in comparison with IRT.

\begin{table}[h!]
\centering
\caption{Prediction Accuracy of LaRT and IRT in Mean Absolute Error. LaRT uniformly predicts responses in unseen entries more accurately than IRT.}
\label{tab:prediction}
\begin{tabular}{ccccccc}
\hline
Fold & 1              & 2              & 3              & 4              & 5              & Average        \\ \hline
LaRT & \textbf{0.235} & \textbf{0.177}          & \textbf{0.183} & \textbf{0.160} & \textbf{0.161} & \textbf{0.183} \\
IRT  & 0.391          & 0.190 & 0.268          & 0.257          & 0.238          & 0.269          \\ \hline
\end{tabular}
\end{table}

\subsubsection{Item Efficiency}
For item efficiency, we measure how many questions are needed to estimate the latent ability of each LLM accurately. We take on the computerized adaptive testing \citep[CAT,][]{meijer1999computerized, wainer2000computerized, chang2015psychometrics} perspective to determine which question to add at each time. CAT has gained great popularity in LLM evaluation with IRT \citep{zhuang2023efficiently, hofmann2025fluid}. Specifically, we first give every LLM the first 10 questions. According to their answers, we can estimate their latent abilities. For the estimated latent ability of each LLM, we choose an item that maximizes the Fisher information among the rest of the question items. Statistically, by choosing questions that maximize the Fisher information, we can adaptively obtain the least variance in estimating the latent ability \citep{reckase200618}. For LaRT, at each time, we choose the item that maximizes the Fisher information for $\theta_i$, whose form is presented in (\ref{eqn:fisher_info_theta}).

In the experiment, similarly as Section \ref{sec:predictive_power}, we concatenate AMC23, AIME24, and AIME25 together, and delete LLMs that fail to answer any of the questions correctly. We assume the ground truth of the latent ability $\btheta$ is the $\hat{\btheta}$ estimated with all the items in the concatenated dataset. Let $\hat{\btheta}_j$ be the estimated latent ability with $j$ question items for each LLM. Note that for each LLM, the $j$-th question items can be different. We compare the scaled Euclidean distance between $\hat{\btheta}_j$ and $\hat{\btheta}_J$ for each $j$. We plot how the scaled Euclidean distance evolve as the number of questions increase in Figure \ref{fig:com_efficiency}.

Figure \ref{fig:com_efficiency} shows that, in most questions, LaRT estimates the latent ability more accurately than IRT. By incorporating the CoT information, LaRT essentially has more data to estimate the latent ability. Therefore, even with a limited number of items, LaRT is able to estimate the latent ability more accurately.

\begin{figure}[h!]
    \centering
    \includegraphics[width=0.5\linewidth]{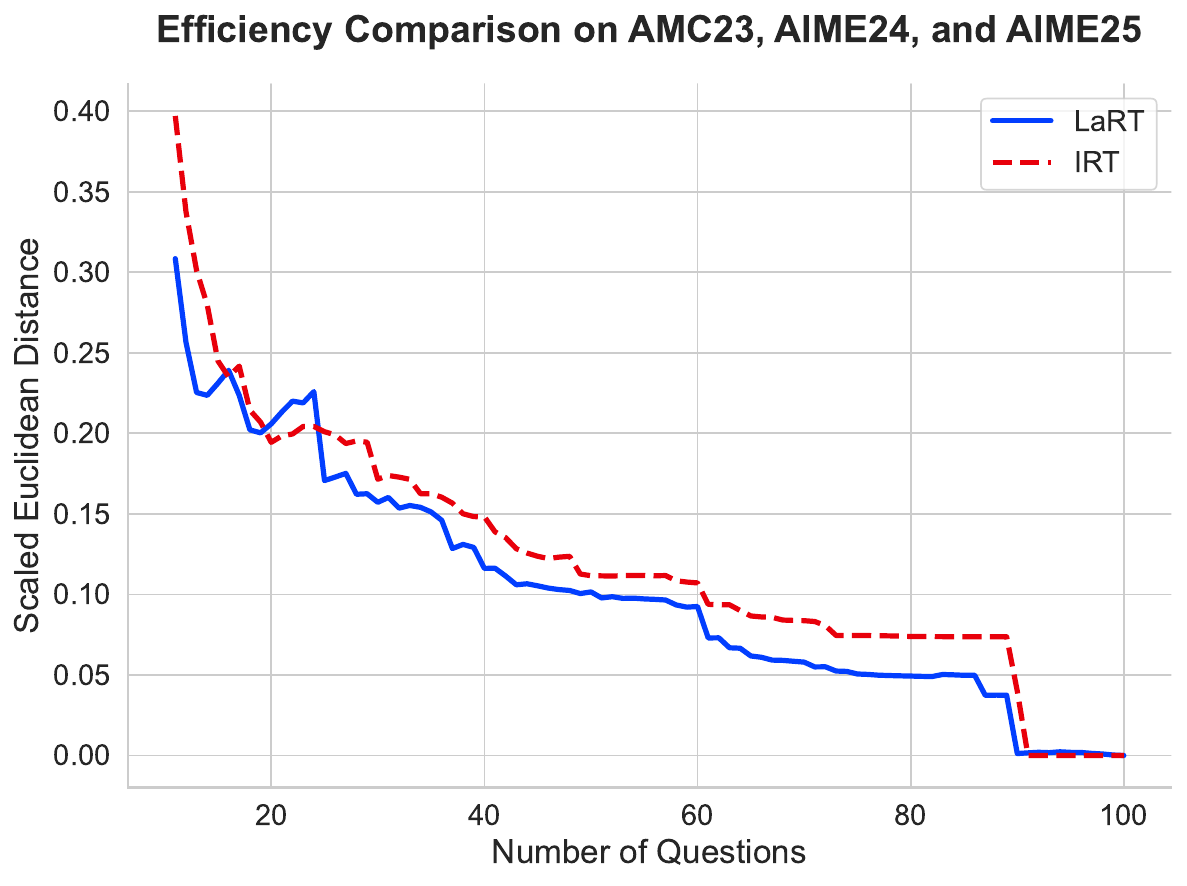}
    \caption{Comparison of Item Efficiency between LaRT and IRT. LaRT can estimate the latent abilities of LLMs accurately with fewer question items than IRT.}
    \label{fig:com_efficiency}
\end{figure}

\subsubsection{Validity}
Prior works have shown IRT-based methods have significant advantage in validity compared with other ranking methods \citep{hofmann2025fluid}. Therefore, if we observe that LaRT outperforms IRT, then it is reasonable to conclude that LaRT enjoys better validity over other ranking methods.

In this experiment, we randomly partition the questions in MATH500 into 5 non-overlapping sets, where each set contains 100 questions. We apply both LaRT and IRT to each set of questions, and estimate the latent ability of each LLM. We calculate the variance of the estimated latent ability for every LLM, and take the sum of the variances over the 138 LLMs.

The variance of the latent ability estimated by LaRT over five non-overlapping sets is 2.0130, compared to 2.3423 for IRT. LaRT significantly reduces the variance of the estimated latent ability across all subsets of MATH500. Figure \ref{fig:speed_ability} show that, MATH500 exhibits the smallest $|\rho|$, which suggests that CoT modeling provides the least assistance for the latent ability estimation compared to other benchmark datasets. LaRT's ability to significantly decrease variance even for the MATH500 dataset suggests that LaRT may offer greater validity improvements over IRT when applied to more challenging benchmark datasets.

\subsubsection{LLM Efficiency}
For LaRT estimation efficiency, we evaluate how many LLMs are needed to yield a good estimate of the population parameters. The accuracy of population parameters is crucial to the active evaluation. First, it determines how well we can estimate the latent ability, since the population parameters are fixed when performing MAP estimation for the individual latent variables. Second, we select the question to assign to each LLM adaptively based on the Fisher information. If the population parameters are poorly estimated, the efficacy of this adaptive procedure will be hampered.

In this specific experiment, we select 100 question items in MATH500 with the largest Fisher information averaged over all LLMs. We randomly permute the 140 LLMs and train LaRT and IRT using cumulative subsets of size $N\in \{50, 75, 100, 125, 140\}$. These subsets are nested, such that the set of LLMs used for a smaller sample size is strictly contained within the set for any larger sample size. We treat $\hat{\ba}$ and $\hat{\bb}$ estimated with the full 140 LLMs as the ground truth. We evaluated convergence by the scaled Euclidean distance between the estimates $\hat{\ba}_N$, $\hat{\bb}_N$ at each $N$, and the ground truth. The result is presented in Table \ref{tab:eff_LLMs}.

Table \ref{tab:eff_LLMs} shows that LaRT outperforms IRT in most of the different $N$.
Specifically, the drastic improvement when $N=50$ is particularly striking. When $N$ is small, two possible scenarios may happen that result in a bad estimate of $\ba$ or $\bb$. For $\bb$, suppose there is item $j$ that almost every LLM answers correctly (or almost every LLM answers incorrectly), and hence $\hat{b}_j$ is close to $\infty$ (or $-\infty$). For $\ba$, suppose there is item $j$ that all LLMs with $\theta_i>0$ answer it correctly, and all LLMs with $\theta_i<0$ answer it wrong. Then, $\hat{a}_j$ will be $\infty$. Due to these two major difficulties for IRT, when $N=50$, IRT has unusable estimates of $\ba$ and $\bb$, while LaRT leverages CoT length for valid estimation.

\begin{table}[h!]
\centering
\caption{Comparison between IRT and LaRT in LLM efficiency. LaRT can estimate the population parameters more accurately with fewer LLMs.}
\label{tab:eff_LLMs}
\begin{tabular}{ccccccccc}
\hline
N     & \multicolumn{2}{c}{50}   & \multicolumn{2}{c}{75}   & \multicolumn{2}{c}{100}  & \multicolumn{2}{c}{125}  \\
Method            & IRT    & LaRT            & IRT             & LaRT   & IRT    & LaRT            & IRT    & LaRT            \\ \hline
$\ba$      & 28.8248 & \textbf{3.8537} & \textbf{5.5137} & 5.6749 & 0.9078 & \textbf{0.8896} & 0.6288 & \textbf{0.5825} \\
$\bb$     & 18.5782 & \textbf{2.3181} & \textbf{2.7583} & 3.0197 & 0.7409 & \textbf{0.7105} & 0.7464 & \textbf{0.7010} \\ \hline
\end{tabular}

\end{table}

\section{Discussion}
\label{sec:discussion}
We proposed the \textbf{La}tency-\textbf{R}esponse \textbf{T}heory model (LaRT) for LLM evaluation in both static and active evaluation settings, with efficient SAEM estimation using spectral initialization.
We justify our modeling choice via rigorous theoretical results in both identifiability and asymptotic distribution. In simulation studies, we confirm LaRT achieves better estimation accuracy and confidence interval coverage. In real-data application, we generate the responses and chain-of-thought (CoT) of LLMs in multiple popular datasets ourselves. Through comprehensive comparison in \textit{Predictive Power}, \textit{Item Efficiency}, \textit{Validity}, \textit{LLM Efficiency}, we find that LaRT has outstanding performance in all four dimensions.

\color{black}
We make the assumption of prior independence for the latent traits in the current LaRT framework. To account for the prior dependence, we can extend LaRT using a multilevel modeling approach \citep{fox2005multilevel, fox2007multilevel}. For example, one can assign a shared hyperprior to the prior distributions of the latent traits $(\theta_i,\tau_i)$ based on the base model family (e.g. Llama or Qwen). We expect that similar identifiability guarantees and algorithms can be derived with a suitable specification of the prior distribution.
\color{black}

This work opens up multiple directions for future research.
First, our algorithmic and theoretical results for LaRT extend naturally to multidimensional latent abilities. If the goal of LLM evaluation is to predict the response for future question items, extending $\ba$, $\btheta$, $\bvarphi$, $\btau$ from one dimension to multiple dimensions can help extract more information from the data. \citet{polo2024tinybenchmarks} applies Multidimensional Item Response Theory models to predict the responses of future question items. Based on the real data application, we believe extending LaRT to multiple dimensions will achieve even better predictive performance.

Second, LaRT's flexibility allows incorporating other covariates that are highly correlated with latent ability. For example, the CoT reasoning themselves contain rich information about an LLM's reasoning ability. Integrating this data directly into the LaRT framework could substantially enhance ranking performance. Prior work has used LLM graders to evaluate step-by-step CoT reasoning, essentially yielding binary tensor data \citep{xia2025evaluating}. The step-by-step grading can serve as another summary of the information contained in the CoT. The LaRT framework can model these binary correct/wrong outcomes, using a probit link, to further refine the estimation of LLM latent abilities. Other considerations such as environmental impact \citep{liang2023holistic} can also be summarized as an additional covariate to modify the rankings of LLMs.

By jointly modeling response accuracy and CoT length, LaRT significantly advances key evaluation desiderata. Nevertheless, deviations from this general trend exist. For instance, on challenging items, certain LLMs may engage in extended but erroneous reasoning processes, yielding long CoT sequences despite incorrect outcomes. This behavior may violate LaRT's assumption of a homogeneous correlation $\rho$ between latent ability $\theta_i$ and the length parameter $\tau_i$ across all LLMs. Future work could address this heterogeneity via mixture modeling to disentangle CoT length distributions conditional on response correctness, possibly by extending the approach in \citep{wang2015mixture}.

Finally, beyond LLM evaluation, this work contributes new methodology with sound theoretical guarantee to the field of psychometrics. Our methods can be readily deployed in educational assessment applications and used to analyze student test takers' response data.

\bibliographystyle{apalike}
\bibliography{ref}

\clearpage
\appendix
\setcounter{equation}{0}
\setcounter{figure}{0}
\setcounter{table}{0}
\renewcommand{\theequation}{S.\arabic{equation}}
\renewcommand{\thefigure}{S\arabic{figure}}
\renewcommand{\thetable}{S\arabic{table}}
\renewcommand{\theHequation}{S.\arabic{equation}}
\renewcommand{\theHfigure}{S.\arabic{figure}}
\renewcommand{\theHtable}{S.\arabic{table}}

The supplementary material provides technical derivations, additional theoretical details, simulation results, and implementation information supporting the main paper. Appendix A derives the SAEM algorithm, including the posterior characterization of the latent traits and the sampling steps used in the stochastic approximation procedure. Appendix B presents the spectral initialization method used to improve the stability and efficiency of SAEM. Appendix C contains the identifiability analysis, including auxiliary results and proofs of the main identifiability theorem. Appendix D provides the assumptions, lemmas, and proof of asymptotic posterior normality for the latent trait estimators. Appendix E reports additional simulation studies, including robustness to link-function misspecification, settings in which LaRT and IRT exhibit convergent behavior, and comparisons with traditional SAEM implementations. Appendix F gives additional details on the real-data implementation, including model lists, prompting protocols, and hyperparameter settings. Appendix G presents supplementary empirical results for the LLM application, including additional confidence intervals, model behavior summaries, and ranking comparisons across benchmark datasets.

\section{Derivation of SAEM Algorithm}
\label{append:alg}
\subsection{Complete form of Lemma 1 in the main article}
\label{append:post_complete}
\begin{lemma}
    \label{lem:posterior_complete}
    The posterior distribution of $(\theta_i,\tau_i)$ given $\Rb_{i,:}$, $\Tb_{i,:}$, and $\bOmega$ follows the following distribution,
    \begin{align}
        \theta_i\mid \Rb_{i,:},\Tb_{i,:};\bOmega &\sim \mathsf{SUN}_{1,J}\left(\mu_{\theta}^{(i)},\sigma_{\theta}^2, \Delta_{i,\text{post}}, \gamma_{i,\text{post}}, \Gamma_{i,\text{post}}  \right)\label{SUN1}\\
        \tau_i\mid \theta_i,\Tb_{i,:}; \bOmega &\sim N\left(\mu_{\tau}^{(i)}, \sigma_{\tau}^{2}\right),\label{SUN2}
    \end{align}
    where $\mathsf{SUN}$ represents the unified skew-normal distribution, and
    \begin{align*}
        \sigma_{\tau}^{2} = \biggl(\frac{1}{1-\rho^2} + \sum_{j=1}^J\frac{\varphi_j^2}{\lambda_j}\biggr)^{-1}, \quad &\mu_{\tau}^{(i)} = \sigma_{\tau}^2\biggl(\frac{\rho\theta_i}{1-\rho^2} - \sum_{j=1}^J \frac{(\log T_{ij} - \omega_j)\varphi_j}{\lambda_j}\biggr),\\
        \sigma_{\theta}^2 = \biggl(\frac{1}{1-\rho^2} - \sigma_{\tau}^2\frac{\rho^2}{1-\rho^2}\biggr)^{-1},\quad &\mu_{\theta}^{(i)} = \sigma_{\theta}^2 \biggl(-\sum_{j=1}^J\frac{(\log T_{ij}-\omega_j)\varphi_j}{\lambda_j}\biggr) \frac{\sigma_{\tau}^2 \rho}{1-\rho^2}\\
        \Delta_{i,\text{post}}=\sigma_{\theta}\Db_{i, 1}^\top \Sbb_i^{-1},\quad &\gamma_{i,\text{post}}=\Sbb_i^{-1}\left(\Db_{i,1}\mu_{\theta}^{(i)} +  \Db_{i, 2}\right),\\  \Gamma_{i,\text{post}}=\Sbb_i^{-1}&\left( \sigmat^2 \Db_{i,1}\Db_{i,1}^\top + \Ib_J \right)\Sbb_i^{-1},
    \end{align*}
and
\begin{align*}
    \Db_{i,1} &= \diag\left( 2R_{i1}-1,\dots,2R_{iJ}-1 \right)\ba, \quad \Db_{i, 2} = \diag\left( 2R_{i1}-1,\dots,2R_{iJ}-1 \right)\bb,\\
    \Sbb_i &= \diag\left\{ (\sigmat^2 \Db_{i,11}^\top\Db_{i,11}+1 )^{1/2},\dots, (\sigmat^2 \Db_{i,1J}^\top\Db_{i,1J}+1 )^{1/2} \right\}\in \RR^{J\times J},
\end{align*}
where $\Db_{i,1j}$ is the $j$th row of $\Db_{i,1}$.
\end{lemma}

\subsection{Proof of Lemma 1 in the main article}
In this section, we derive the complete conditional distribution of $\bxi_i=(\theta_i,\tau_i)$ given $\bOmega$, $\Rb$, and $\Tb$. Due to the conditional independence structure, given $\bOmega$, $\Rb$, and $\Tb$, $\bxi_i$'s are independent. Hence, we only need to derive the complete conditional of $\bxi_i$. The complete conditional is,
\begin{align*}
    P(\theta_i,\tau_i\mid \Rb_{i,:}, \Tb_{i,:};\bOmega)
    &\propto P(\theta_i,\tau_i; \bOmega)\prod_{j=1}^J P(R_{ij}\mid \theta_i;\bOmega) \times \prod_{j=1}^J P(T_{ij}\mid \tau_i;\bOmega)\\
    &= \left[P(\theta_i; \bOmega) \prod_{j=1}^JP(R_{ij}\mid \theta_i; \bOmega)\right]\left[ P(\tau_i\mid \theta_i;\bOmega)\prod_{j=1}^J P(T_{ij}\mid \tau_i;\bOmega) \right].
\end{align*}

First, we focus on the conditional distribution of $P(\tau_i\mid\theta_i,\Tb_{i,:};\bOmega)$. Since $(\theta_i,\tau_i)\sim N(0,\bSigma) $, $\tau_i\mid\theta_i,\bOmega\sim N(\rho\theta_i,1-\rho^2 )$. The likelihood $P(T_{ij}\mid \tau_i;\bOmega)$ is also a normal distribution. Hence, the conditional distribution is still normal.
\begin{align*}
    P(\tau_i\mid \theta_i, \Tb_{i,:};\bOmega) &\propto p(\tau_i\mid \theta_i;\bOmega) \prod_{j=1}^J P(T_{ij}\mid \tau_i;\bOmega)\\
    &= \exp\left\{ -\frac{(\tau_i-\rho\theta_i)^2}{2(1-\rho^2)} \right\}\cdot\exp\left\{ -\sum_{j=1}^J \frac{(\log t_{ij}-\omega_j+\varphi_j\tau_i)^2}{2\lambda_j} \right\}\\
    &\propto \exp\left\{-\frac{1}{2\check{\sigma}_{\tau}^{(i)2}} \left(\tau_i-\check{\mu}_{\tau}^{(i)}\right)^2\right\},
\end{align*}
where
\begin{align*}
    \check{\sigma}_{\tau}^{(i)2}
    &= \left(\frac{1}{1-\rho^2} + \sum_{j=1}^J\frac{\varphi_j^2}{\lambda_j}\right)^{-1},\\
    \check{\mu}_{\tau}^{(i)}
    &= \left(\frac{1}{1-\rho^2} + \sum_{j=1}^J\frac{\varphi_j^2}{\lambda_j}\right)^{-1}
       \left(\frac{\rho\theta_i}{1-\rho^2} - \sum_{j=1}^J \frac{(\log t_{ij} - \omega_j)\varphi_j}{\lambda_j}\right).
\end{align*}
Note that $\check{\sigma}_{\tau}^{(i)}$ is independent of the index $i$. For simplicity, we will denote it as $\check{\sigma}_{\tau}$ from now on. Thus, the conditional distribution of $\tau_i$ is,
\begin{equation*}
    \tau_i\mid \theta_i,\Tb_{i,:};\bOmega \sim N\left(\check{\mu}_{\tau}^{(i)}, \check{\sigma}_{\tau}^2\right).
\end{equation*}

Then, for the marginal distribution $P(\theta_i\mid \Rb_{i,:},\Tb_{i,:};\bOmega)$, note that by marginalizing out $\tau_i$, normalizing constants containing $\theta_i$ will contribute to the marginal posterior of $\theta_i$. In the sequel, we consider the normalizing constant "twist" $P(\theta_i;\bOmega)$, and compute the twisted prior. First, for the normalizing constant concerning $\theta_i$,
\begin{align*}
    &\int_{\RR}P(\tau_i\mid \theta_i;\bOmega)\prod_{j=1}^J P(T_{ij}\mid \tau_i;\bOmega)d\tau_i\\ &\propto \int_{\RR}\exp\left\{ -\frac{1}{2(1-\rho^2)}(\tau_i-\rho\theta_i)^2 \right\}\exp\left\{ -\sum_{j=1}^J \frac{1}{2\lambda_j}(\log t_{ij}-\omega_j+\varphi_j\tau_i)^2 \right\}d\tau_i\\
    &=\exp\left\{ -\frac{\rho^2}{2(1-\rho^2)}\theta_i^2 \right\}\exp\left\{\frac{1}{2\check{\sigma}_{\tau}^2}\left(\check{\mu}_{\tau}^{(i)}\right)^2 \right\}\int_{\RR} \exp\left\{-\frac{1}{2\check{\sigma}_{\tau}^{(i)2}} \left(\tau_i-\check{\mu}_{\tau}^{(i)}\right)^2\right\} d\tau_i\\
    &\propto \exp\left\{ -\frac{\rho^2}{2(1-\rho^2)}\theta_i^2 \right\} \exp\left\{ \frac{\rho^2}{2(1-\rho^2)^2}\left( \frac{1}{1-\rho^2} + \sum_{j=1}^J \frac{\varphi_j^2}{\lambda_j} \right)^{-1}\theta_i^2 \right.\\
    &\left.- \left( \frac{1}{1-\rho^2} + \sum_{j=1}^J \frac{\varphi_j^2}{\lambda_j} \right)^{-1}\left(\sum_{j=1}^J \frac{(\log T_{ij}-\omega_j)\varphi_j}{\lambda_j} \right)\frac{\rho}{1-\rho^2}\theta_i \right\}.
\end{align*}

Note that
\begin{align*}
    \frac{1}{1-\rho^2}\left( \frac{1}{1-\rho^2} + \sum_{j=1}^J \frac{\varphi_j^2}{\lambda_j} \right)^{-1} = \left( 1 + (1-\rho^2)\sum_{j=1}^J\frac{\varphi_j^2}{\lambda_j} \right)^{-1} \leq 1.
\end{align*}

Thus, for the normalizing constant concerning $\theta_i$, inside the exponential, it is still a quadratic form whose quadratic term has negative coefficient. Note that $P(\theta_i;\bOmega)=N(0,1)$. Denote the twisted prior of $\theta_i$ as $\tilde{\phi}(\theta_i\mid \Tb_{i,:} ;\bOmega)$, which is
\begin{align*}
    \tilde{\phi}(\theta_i\mid \Tb_{i,:}; \bOmega)
    &\propto \exp\left\{-\frac{1}{2}\theta_i^2\right\}
    \exp\left\{ -\frac{\rho^2}{2(1-\rho^2)}\theta_i^2 \right\}\\
    &\qquad\times \exp\left\{ \frac{\rho^2}{2(1-\rho^2)^2}
    \left( \frac{1}{1-\rho^2} + \sum_{j=1}^J \frac{\varphi_j^2}{\lambda_j} \right)^{-1}\theta_i^2 \right.\\
    &\left.\qquad- \left( \frac{1}{1-\rho^2} + \sum_{j=1}^J \frac{\varphi_j^2}{\lambda_j} \right)^{-1}
    \left(\sum_{j=1}^J \frac{(\log T_{ij}-\omega_j)\varphi_j}{\lambda_j} \right)
    \frac{\rho}{1-\rho^2}\theta_i \right\}\\
    &\propto \exp\left\{ -\frac{1-\rho^2\left(1/(1-\rho^2)+\sum_j\varphi_j^2/\lambda_j\right)^{-1}}{2(1-\rho^2)} \theta_i^2 \right.\\
    &\left.- \left( \frac{1}{1-\rho^2} + \sum_{j=1}^J \frac{\varphi_j^2}{\lambda_j} \right)^{-1}\left(\sum_{j=1}^J \frac{(\log T_{ij}-\omega_j)\varphi_j}{\lambda_j} \right)\frac{\rho}{1-\rho^2}\theta_i  \right\}\\
    &\propto \exp\left\{ -\frac{1}{2\sigma_{\theta}^{(i)2}}(\theta_i-\mu_{\theta}^{(i)})^2 \right\},
\end{align*}
where
\begin{equation*}
    \sigma_{\theta}^2 = \biggl(\frac{1}{1-\rho^2} - \sigma_{\tau}^2\frac{\rho^2}{1-\rho^2}\biggr)^{-1},\quad \mu_{\theta}^{(i)} = \sigma_{\theta}^2 \biggl(-\sum_{j=1}^J\frac{(\log t_{ij}-\omega_j)\varphi_j}{\lambda_j}\biggr) \frac{\sigma_{\tau}^2 \rho}{1-\rho^2}.
\end{equation*}

\color{black}
\subsection{Sampling from SUN distribution}
\label{append:sample_sun}

\begin{corollary}[Corollary 4.3 in \citet{li2025Sparse}]
    \label{coro:posterior}
    If $(\theta_i,\tau_i)$ follows the distribution specified in Lemma 1 of the main article, then
    \begin{equation}
        \theta_i\mid \Rb_{i,:},\Tb_{i,:};\bOmega \stackrel{d}{=} \mu_{\theta}^{(i)} + \sigma_{\theta}\left[ U_0 + \sigma_{\theta}\Db_{i,1}^{\top} \left( \sigma_{\theta}^2 \Db_{i,1}\Db_{i,1}^{\top} + \Ib_J \right)^{-1}\Sbb_i U_1 \right],\label{eq_coro}
    \end{equation}
    where $U_0\sim N\Bigl(0,1-\sigma_{\theta}^2\Db_{i,1}^{\top}\bigl( \sigma_{\theta}^2 \Db_{i,1} \Db_{i,1}^{\top} + \Ib_J\bigr)^{-1}\Db_{i,1} \Bigr)$, $U_1$ follows a truncated normal distribution with mean $0$ and variance $\Sbb_i^{-1}\bigl( \sigma_{\theta}^2 \Db_{i,1}\Db_{i,1}^{\top}+\Ib_J\bigr)\Sbb_i^{-1} $, truncated with lower bound at $-\Sbb_i^{-1}\bigl(\Db_{i,1}\mu_{\theta}^{(i)}+\Db_{i,2}\bigr)$, and $U_0$ is independent of $U_1$.
\end{corollary}

\section{Initialization for the SAEM algorithm}
\label{sec:initialization}
Effective initialization is critical for SAEM's efficiency and stability. We present a spectral-based initialization that is non-iterative, data-driven, and statistically consistent.
This approach exploits the fact that the probit link for response accuracy and
the log-normal specification for CoT fall within the scope of generalized linear factor
models (GLFMs) or linear factor models, allowing us to adapt the spectral-based method of
\citet{zhang2020note}, originally developed for GLFM settings, to initialize the LaRT
framework. Specifically, we extend their procedure—which was designed for single-modality
response-accuracy data—to accommodate the bimodal structure of response accuracy and CoT,
introducing additional steps to initialize the CoT-related parameters under the LaRT
hierarchical model. Algorithm~\ref{alg:initialize} summarizes the full algorithm.

\begin{algorithm}[h!]
    \caption{Nonlinear Spectral Initialization for the SAEM algorithm.}
    \label{alg:initialize}
    \begin{algorithmic}[1]
        \Require Binary response matrix $\Rb\in \{0,1\}^{N\times J}$, CoT length matrix $\Tb \in \RR_{+}^{N\times J} $, thresholding parameter $\epsilon_{N,J}$.
        \Ensure Spectral estimates of $\check{\bOmega}$, $\check{\theta}_i$, $\check{\tau}_i$ for $i\in [N]$.

        \State Perform a full SVD of $\Rb=\sum_{i=1}^{N\wedge J}\sigma_i\bu_i\bv_i^\top$.
        \State Let $\Xb = \sum_{k=1}^{\tilde{K}}\sigma_k\bu_k\bv_k^\top$, where $\tilde{K}=\max\bigl\{K+1,\argmax_k\{ \sigma_k\geq 1.01\sqrt{N\vee J} \}\bigr\}$.
        \State Let $\tilde{\Mb}=(\tilde{M}_{ij})_{N\times J}$ be
        \begin{equation*}
            \hat{M}_{ij}=\left\{
            \begin{aligned}
                &\Phi^{-1} (\epsilon_{N,J}),\quad \text{if } x_{ij} < \epsilon_{N,J},\\
                &\Phi^{-1}(x_{ij}), \quad \text{if } \epsilon_{N,J} \leq x_{ij} \leq 1-\epsilon_{N,J},\\
                &\Phi^{-1}(1-\epsilon_{N,J}),\quad \text{if } x_{ij} > 1-\epsilon_{N,J}.
            \end{aligned}
            \right.
        \end{equation*}
        \State Let $\check{b}_j=\sum_{i=1}^N\tilde{m}_{ij}/N$, $\check{\omega}_k=\sum_{i=1}^N\log T_{ij}/N$ $\forall j \in [J]$.
        \State Perform top-1 SVD on $\hat{\Mb}=(\tilde{M}_{ij}-\check{b}_j)_{N\times J}$ and $\log \hat{\Tb} = (\log T_{ij} - \check{\omega}_j)_{N\times J}$ for respectively $\check{\sigma}_1\check{\bu}_1\check{\bv}_1^\top$ and $\tilde{\sigma}_1\tilde{\bu}_1\tilde{\bv}_1^{\top} $.
        \State Let $\check{\btheta}=\sqrt{N}\check{\bu}_1$, $\check{\ba}= \check{\sigma}_1\check{\bv}_1/\sqrt{N}$, $\check{\btau}=\sqrt{N}\tilde{\bu}_1$, and $\check{\bvarphi} = \tilde{\sigma}_1\tilde{\bv}_1/\sqrt{N}$.
        \State Let $\check{\rho}=\sum_{i=1}^N\check{\theta}_i\check{\tau}_i/N$, and $\check{\lambda}_j = \sum_{i=1}^N(\log \hat{T}_{ij}-\check{\tau}_i \check{\varphi}_j)^2/N$.
    \end{algorithmic}
\end{algorithm}

The algorithm proceeds in three stages. The algorithm begins by first performing an SVD on the response accuracy matrix to extract its
dominant latent structure and reduce noise, followed by an inverse link transformation to
obtain an approximately linear latent representation (steps 1–3). In step 3,
$\epsilon_{N,J}$ serves as a threshold for truncating $\Xb$ to the range of $\Rb$, and is
set to $10^{-9}$ in this work, following the general guidelines in \citet{zhang2020note}.
The estimates of the difficulty parameters $\check{b}_j$ and the intensity parameters
$\check{\omega}_k$ are obtained by averaging the corresponding entries in the transformed data $\tilde{\Mb}$,  respectively, as these parameters
serve as intercepts in their linear factor model forms (step 4). Then, SVD is applied to
the centered $\tilde{\Mb}$ and $\log(\Tb)$, and the estimates of the discrimination
parameters $\check{\ba}$ and $\check{\bvarphi}$, as well as the latent ability variables
$\check{\btheta}$ and $\check{\btau}$, are obtained by extracting their corresponding
components from the SVD (steps 5–6). Finally, the log-CoT residual variance estimates $\check{\lambda}_j$ and the correlation
estimate $\check{\rho}$ are computed using their closed-form expressions,
$\sum_{i=1}^N (\log \hat{T}_{ij}-\check{\tau}_i \check{\varphi}_j)^2 / N$ and
$\sum_{i=1}^N \theta_i \tau_i / N$, respectively (step 7).
For more details on the implementation, please refer to \citet{zhang2020note}.

This informed initialization enables the algorithm to adopt a decaying step size $\alpha_t$ from the very first iteration. In contrast, standard SAEM implementations typically require an initial \emph{burn-in} phase, during which a large step size (e.g., $\alpha_t = 1$) must be used for many iterations \citep{lavielle2014improved, kuhn2004coupling, camilli2019stochastic}. This phase is crucial when using random initializations, which often start far from the optimum and therefore require large updates to move into the optimal region. However, it also introduces an additional tuning burden: the appropriate length of the burn-in phase varies with the initial values, making the procedure more unstable and sensitive to initialization.
This initialization strategy bypasses this burn-in phase entirely by replacing random initialization with an efficient, non-iterative, data-driven strategy, while also enjoying favorable statistical consistency properties \citep{zhang2020note}. Our spectral initialization starts near the optimum, allowing SAEM to use a decaying step size from iteration one. This yields both faster convergence and improved estimation accuracy compared with traditional SAEM implementations, as demonstrated in Appendix~\ref{append:SAEM}.

\color{black}

\section{Proof of Identifiability}
\color{black}
\subsection{Additional Results of Identifiability}
\label{append:iden_prop}
We first establish Proposition~\ref{prop:iden},
which reduces the identifiability analysis to checking a set of more tractable conditions.

\begin{proposition}
    \label{prop:iden}
    Two sets of parameters $(\ba, \bbb, \bomega, \bvarphi, \blambda, \rho)$ and $(\ba^{\prime} , \bbb^{\prime}, \bomega^{\prime}, \bvarphi^{\prime}, \blambda^{\prime}, \rho^{\prime}) $ give rise to the same marginal distribution if and only if
    \begin{equation*}
        \frac{b_j}{\sqrt{a_j^2+1}} = \frac{b_j^{\prime}}{\sqrt{a_{j}^{\prime 2}+1}},\quad \omega_{j} = \omega_{j}^{\prime},\quad \varphi_j^2+\lambda_j = \varphi_j^{\prime 2}+\lambda_j^{\prime},
    \end{equation*}
    for all $j\in[J]$,
    \begin{equation*}
        \frac{a_{j_1}a_{j_2}}{\sqrt{a_{j_1}^2+1}\sqrt{a_{j_2}^2+1}} = \frac{a_{j_1}^\prime a_{j_2}^\prime }{\sqrt{a_{j_1}^{\prime2} + 1} \sqrt{a_{j_2}^{\prime 2}+1}},\quad \varphi_{j_1}\varphi_{j_2} = \varphi_{j_1}^{\prime}\varphi_{j_2}^\prime,
    \end{equation*}
    for all $j_1\neq j_2$,
    \begin{equation*}
        \frac{\rho a_{j_1}\varphi_{j_2}}{\sqrt{a_{j_1}^2+1}} = \frac{\rho^{\prime} a_{j_1}^{\prime}\varphi_{j_2}^{\prime}}{\sqrt{a_{j_1}^{\prime 2}+1}},
    \end{equation*}
    for all $j_1,j_2\in [J]$.
\end{proposition}

For the parameters associated with response accuracy $\Rb$, this proposition
requires that the probit thresholds and tetrachoric correlations implied by the marginal
distribution of $\Rb$ admit a unique parameterization. For the parameters associated with
CoT $\Tb$, it further requires a unique parameterization of the mean and
variance of its marginal distribution. In addition, it imposes that the covariance
between $R_{i,j_1}$ and $T_{i,j_2}$ remains the same for all $j_1,j_2\in[J]$. The proof of
Proposition~\ref{prop:iden} is provided in Appendix~\ref{append:prop_iden}.
\color{black}

\subsection{Proof of Proposition \ref{prop:iden}}
\label{append:prop_iden}
This proof is similar to the proof of Proposition 3.1 in \citet{fang2021identifiability}.
First, we derive the marginal distribution of $R_{ij}$ and $\log T_{ij}$ knowing $\bOmega$. Let $\varepsilon_{j}\sim N(0,1)$ independently for all $j$. Then, for $(\theta_i,\tau_i) \sim N(0,\bSigma)$,
\begin{align*}
    P(R_{ij}=1) = \EEE_{\theta_i}\left[P(R_{ij}=1\mid \theta_i)\right] &= \EEE_{\theta_i} \left[\EEE_{\varepsilon_j}\left[ \mathbbm{1}(\varepsilon_j \leq b_j+a_j\theta_i) \right]\right] \\
    &= P(\varepsilon_j\leq b_j+a_j\theta_i)\\
    &= P\left(\sqrt{a_{j}^2+1}\eta_j+b_j\geq 0\right)\\
    &= \Psi \biggl( -\frac{b_j}{\sqrt{a_j^2+1}} \biggr),
\end{align*}
where $\eta_{j} = (a_j\theta_i-\varepsilon_j)/\sqrt{a_{j}^2+1}\sim N(0,1)$ and $\Psi(x)=1-\Phi(x)$ the complementary cumulative density function for a standard normal variable.

Similarly, let $e_{j}\sim N(0,\lambda_j)$ independently for all $j$,
\begin{align*}
    P\left(\log T_{ij} \geq \log t_{ij} \right) &= \EEE_{\tau_i}\left[P(\omega_j-\varphi_j\tau_i+e_j\geq \log t_{ij}\mid \tau_i)\right]\\
    &= \EEE_{\tau_i}\left[\EEE_{e_j}\left[\mathbbm{1}(\omega_j - \varphi_j\tau_i + e_j \geq \log t_{ij})\mid \tau_i\right]\right]\\
    &= P(\omega_j - \varphi_j\tau_i +e_j\geq \log t_{ij})\\
    &= P\left( \sqrt{\varphi_j^2+\lambda_j}\zeta_j + \omega_j \geq \log t_{ij} \right)\\
    &= \Psi \biggl( \frac{\log t_{ij} - \omega_j}{\sqrt{\varphi_j^2+\lambda_j}} \biggr),
\end{align*}
where $\zeta_{j} = (e_j-\varphi_j\tau_i)/\sqrt{\varphi_j^2+\lambda_j}\sim N(0,1)$.

We next calculate the pairwise marginal distributions for response--response, latency--latency, and response--latency pairs. We first compute the following covariances.
\begin{equation*}
    \Cov{\eta_{j_1}}{\eta_{j_2}} = \Cov{\frac{a_{j_1}\theta_i-\varepsilon_{j_1}}{\sqrt{a_{j_1}^2+1}}}{\frac{a_{j_2}\theta_i-\varepsilon_{j_2}}{\sqrt{a_{j_2}^2+1}}} = \frac{a_{j_1}a_{j_2}}{\sqrt{a_{j_1}^2+1}\sqrt{a_{j_2}^2+1}}.
\end{equation*}
Similarly,
\begin{align*}
    \Cov{\zeta_{j_1}}{\zeta_{j_2}} &= \frac{\varphi_{j_1}\varphi_{j_2}}{\sqrt{\varphi_{j_1}^2+\lambda_{j_1}}\sqrt{\varphi_{j_2}^2+\lambda_{j_2}}},\\
    \Cov{\eta_{j_1}}{\zeta_{j_2}} &= -\frac{\rho a_{j_1}\varphi_{j_2}}{\sqrt{a_{j_1}^2+1}\sqrt{\varphi_{j_2}^2+\lambda_{j_2}}}.
\end{align*}

Therefore, the two-component marginal distributions are,
\begin{align*}
    P(R_{i,j_1}=1, R_{i,j_2}=1) &= P(\varepsilon_{j_1}\leq b_{j_1}+a_{j_1}\theta_i, \quad \varepsilon_{j_2}\leq b_{j_2}+a_{j_2}\theta_i)\\
    &= P\left(\sqrt{a_{j_1}^2+1} \eta_{j_1} + b_{j_1}\geq 0, \quad \sqrt{a_{j_2}^2+1} \eta_{j_2}+b_{j_2}\geq 0 \right)\\
    &= \Psi \biggl( -\frac{b_{j_1}}{\sqrt{a_{j_1}^2+1}}, - \frac{b_{j_2}}{\sqrt{a_{j_2}^2+1}}, \frac{a_{j_1}a_{j_2}}{\sqrt{a_{j_1}^2+1}\sqrt{a_{j_2}^2+1}} \biggr),
\end{align*}
where $\Psi(x_1,x_2,\rho)=P(X_1\geq x_1, X_{2}\geq x_2)$, $X_1, X_2\sim N(0,1)$ and $\Cov{X_1}{X_2}=\rho$. Similarly,
\begin{align*}
    &P(\log T_{i,j_1}\geq \log t_{i,j_1},\log T_{i,j_2}\geq \log t_{i,j_2}) \\
    &= P\left(\sqrt{\varphi_{j_1}^2+\lambda_{j_1}}\zeta_{j_1}+\omega_{j_1}\geq \log t_{i,j_1}, \sqrt{\varphi_{j_2}^2+\lambda_{j_2}}\zeta_{j_2}+\omega_{j_2}\geq \log t_{i,j_2} \right)\\
    &= \Psi\biggl( \frac{\log t_{i,j_1}-\omega_{j_1}}{\sqrt{\varphi_{j_1}^2+\lambda_{j_1}}}, \frac{\log t_{i,j_2}-\omega_{j_2}}{\sqrt{\varphi_{j_2}^2+\lambda_{j_2}}}, \frac{\varphi_{j_1}\varphi_{j_2}}{\sqrt{\varphi_{j_1}^2+\lambda_{j_1}}\sqrt{\varphi_{j_2}^2+\lambda_{j_2}}}  \biggr),
\end{align*}
\begin{align*}
    P(R_{ij}=1,\log T_{ij}\geq \log t_{ij}) &= P\left( b_{j_1}+\sqrt{a_{j_1}^2+1}\eta_{j_1} \geq 0, \sqrt{\varphi_{j_1}^2+\lambda_{j_2}} \zeta_{j_2}+\omega_{j_2}\geq \log t_{i,j_2} \right)\\
    &= \Psi\biggl(-\frac{b_{j_1}}{\sqrt{a_{j_1}^2+1}}, \frac{\log t_{i,j_2}-\omega_{j_2}}{\sqrt{\varphi_{j_2}^2+\lambda_{j_2}}}, -\frac{\rho a_{j_1}\varphi_{j_2}}{\sqrt{a_{j_1}^2+1}\sqrt{\varphi_{j_2}^2+\lambda_{j_2}}}\biggr).
\end{align*}
Following this strategy, the full joint distribution of $\Rb_i$ and $\log \Tb_i$ can be characterized by their pairwise covariances. For simplicity, we omit the full expression.

For sufficiency, suppose there are two sets of parameters $\bOmega$ and $\bOmega^{\prime}$ following conditions in Proposition \ref{prop:iden}. Then, note that the joint distribution of $\Rb_i$ and $\log \Tb_i$ only depends on $b_j/\sqrt{a_{j}^2+1}$, $\omega_j$, $\sqrt{\varphi_{j}^2+\lambda_j}$, $\Cov{\eta_{j_1}}{\eta_{j_2}} $, $\Cov{\zeta_{j_1}}{\zeta_{j_2}} $, and $\Cov{\eta_{j_1}}{\zeta_{j_2}} $. When $\bOmega$ and $\bOmega^{\prime}$ satisfy the set of conditions in Proposition \ref{prop:iden}, these quantities are the same. Therefore, $\bOmega$ and $\bOmega^{\prime}$ give rise to the same joint distribution of $\Rb_{i}$ and $\log \Tb_{i}$.

For necessity, suppose $\bOmega$ and $\bOmega^{\prime}$ give rise to the same joint distribution of $\Rb_i$ and $\log \Tb_i$. First, for one-component marginal distribution of $R_{ij}$ and $\log T_{ij}$, $\bOmega$ and $\bOmega^{\prime}$ need to satisfy,
\begin{equation*}
    \frac{b_{j}}{\sqrt{a_j^2+1}} = \frac{b_{j}^{\prime}}{\sqrt{a_j^{\prime 2}+1}},\quad \frac{\log t_{ij}-\omega_j}{\sqrt{\varphi_j^2 + \lambda_j}} = \frac{\log t_{ij}-\omega_j^{\prime}}{\sqrt{\varphi_j^{\prime 2} + \lambda_j^{\prime}}},
\end{equation*}
for all $\log t_{ij}\in \RR$. Therefore,
\begin{equation}
    \label{eqn:mar_impli}
    \omega_j = \omega_{j}^{\prime},\quad \varphi_j^2 + \lambda_j = \varphi_j^{\prime 2} + \lambda_j^{\prime}.
\end{equation}

Then, consider the two-component marginals, $\bOmega$ and $\bOmega^{\prime}$ giving rise to the same distribution asks for the following equalities,
\begin{align*}
    \frac{a_{j_1}a_{j_2}}{\sqrt{a_{j_1}^2+1}\sqrt{a_{j_2}^2+1}} &= \frac{a_{j_1}^{\prime} a_{j_2}^{\prime}}{\sqrt{a_{j_1}^{\prime 2}+1} \sqrt{a_{j_2}^{\prime 2}+1}},\\
    \frac{\varphi_{j_1}\varphi_{j_2}}{\sqrt{\varphi_{j_1}^2 + \lambda_{j_1}} \sqrt{\varphi_{j_2}^2 + \lambda_{j_2}}} &= \frac{\varphi_{j_1}^{\prime} \varphi_{j_2}^{\prime}}{\sqrt{\varphi_{j_1}^{\prime 2} + \lambda_{j_1}^{\prime}} \sqrt{\varphi_{j_2}^{\prime 2} + \lambda_{j_2}^{\prime}}},\\
    \frac{\rho a_{j_1}\varphi_{j_2}}{\sqrt{a_{j_1}^2+1}\sqrt{\varphi_{j_2}+\lambda_{j_2}}} &= \frac{\rho^{\prime} a_{j_1}^{\prime}\varphi_{j_2}^{\prime}}{\sqrt{a_{j_1}^{\prime 2}+1}\sqrt{\varphi_{j_2}^{\prime 2}+\lambda_{j_2}^{\prime}}}.
\end{align*}
Combining with (\ref{eqn:mar_impli}), this requires
\begin{align*}
    \frac{a_{j_1}a_{j_2}}{\sqrt{a_{j_1}^2+1}\sqrt{a_{j_2}^2+1}} = \frac{a_{j_1}^{\prime} a_{j_2}^{\prime}}{\sqrt{a_{j_1}^{\prime 2}+1} \sqrt{a_{j_2}^{\prime 2}+1}}, \quad \varphi_{j_1}\varphi_{j_2} = \varphi_{j_1}^{\prime} \varphi_{j_2}^{\prime}, \quad \frac{\rho a_{j_1}\varphi_{j_2}}{\sqrt{a_{j_1}^2+1}} = \frac{\rho^{\prime} a_{j_1}^{\prime}\varphi_{j_2}^{\prime}}{\sqrt{a_{j_1}^{\prime 2}+1}}.
\end{align*}
Therefore, conditions in Proposition \ref{prop:iden} are necessary.

\subsection{Proof of Theorem 1 in the main article}
\label{append:thm_iden}
First, we show the identifiability of the probit model part. Suppose there are two sets of parameters $\bOmega$ and $\bOmega^{\prime}$ that give rise to the same distribution for $\Rb$ and $\log \Tb$. Define $\tilde{\ba}=(\tilde{a}_1,\ldots ,\tilde{a}_J)$, where $\tilde{a}_j=a_j/\sqrt{a_j^2+1}$. Then, from Proposition \ref{prop:iden}, we have
\begin{equation*}
    \tilde{\ba}\tilde{\ba}^\top + \Sbb = \tilde{\ba}^{\prime}\tilde{\ba}^{\prime \top} + \Sbb^{\prime},
\end{equation*}
where $\Sbb=\diag\{ b_j/\sqrt{a_j^2+1}-a_j^2/(a_j^2+1) \}_{j=1}^J$ and similarly for $\Sbb^{\prime}$.

If any row of $\tilde{\ba}$ is deleted, $\tilde{\ba}$ still ranks 1 because there are at least 3 non-zero entries in $\tilde{\ba}$. Then, from Theorem 5.1 in \citet{anderson1956statistical}, $\Sbb^{\prime}=\Sbb$ and $\tilde{\ba}\tilde{\ba}^\top=\tilde{\ba}^{\prime}\tilde{\ba}^{\prime \top}$. The diagonal entries of $\tilde{\ba}\tilde{\ba}^\top$ and $\tilde{\ba}^{\prime} \tilde{\ba}^{\prime \top}$ being equal implies
\begin{equation*}
    \frac{a_j^2}{a_j^2+1} = \frac{a_j^{\prime 2}}{a_j^{\prime 2}+1}.
\end{equation*}
Hence, $a_j^2=a_j^{\prime 2}$ for all $j\in [J]$. Combining with the definition of $\Sbb$, we have $b_j=b_j^{\prime}$ for all $j\in [J]$.

Additionally, by Lemma 5.1 in \citet{anderson1956statistical}, we have
\begin{equation*}
    \frac{a_j}{\sqrt{a_{j}^2+1}}=\frac{ca_j^{\prime}}{\sqrt{a_{j}^{\prime 2}+1}},
\end{equation*}
where $c\in \{-1,1\}$. Since $\sum_{j=1}^Ja_j>0$, $c$ can only be $1$. Therefore, $\ba=\ba^{\prime}$.

For the parameters $\bvarphi$, $\bomega$, the proof is the same as the probit model case. For $\blambda$, since $\bvarphi=\bvarphi^{\prime}$, $\varphi_j^2=\varphi_j^{\prime 2}$ for all $j\in [J]$, and thus $\lambda_j=\lambda_j^{\prime}$ for all $j\in [J]$. For $\rho$, since every other parameter is identified, following Proposition \ref{prop:iden}, $\rho=\rho^{\prime}$.

\section{Proof of APN}
\label{append:apn}
\subsection{Required Assumptions}
\label{append:assump}
\begin{assumption}
    \label{assump:finite}
    If restricted to any compact set $K\subseteq \Theta$, $|a_j\phi(a_j\theta+b_j)|$, $|a_j^2\phi^{\prime}(a_j\theta+b_j)|$, and $|a_j^3\phi^{\prime\prime}(a_j\theta+b_j)|$ are uniformly bounded for all $j\in \mathbb{N}$.
    Additionally, there exists constants $0< d_0(K) < d_1(K) <1$, such that for all $j\in \mathbb{N}$, we have
    \begin{equation*}
        d_0(K)\leq \inf _{(j,\theta)\in \mathbb{N}\times K}\Phi(a_j\theta+b_j)\leq \sup_{(j,\theta)\in \mathbb{N} \times K} \Phi(a_j\theta+b_j)\leq d_1(K).
    \end{equation*}
    Moreover, $\sup_{j\in \mathbb{N}}\varphi_j^2/\lambda_j\leq C <\infty$, where $C$ is a constant.

\end{assumption}

\begin{assumption}
    \label{assump:iden_theta}
    For every $(\theta,\tau)\neq (\theta_0,\tau_0)$, there is $c_1(\theta), c_2(\tau) < 0$, such that
    \begin{align*}
        &\limsup_{J\to \infty} \frac{1}{J}\sum_{j=1}^J \Biggl[
        \Phi(a_j\theta_0+b_j)\log \frac{\Phi(a_j\theta+b_j)}{\Phi(a_j\theta_0+b_j)}\\
        &\qquad\qquad+ \Phi(-a_j\theta_0-b_j)\log
        \frac{\Phi(-a_j\theta-b_j)}{\Phi(-a_j\theta_0-b_j)} \Biggr]
        \leq c_1(\theta),\\
        &\limsup_{J\to \infty}-\frac{1}{J}\sum_{j=1}^J \frac{1}{2\lambda_j} \Bigl[
        2\varphi_j(\tau-\tau_0)(\omega_j-\varphi_j\tau_0)+\varphi_j^2(\tau^2-\tau_0^2)\\
        &\qquad\qquad- 2\varphi_j\omega_j(\tau-\tau_0) \Bigr]
        \leq c_2(\tau),
    \end{align*}
    and additionally
    \begin{align*}
        \sup_{(\theta,\tau)\in \Theta\setminus B_{\delta}(\bxi_0)} c_1(\theta) <0,
        \quad \sup_{(\theta,\tau)\in \Theta\setminus B_{\delta}(\bxi_0)} c_2(\tau) < 0,
    \end{align*}
    for all $\delta>0$.
\end{assumption}

\begin{assumption}
    \label{assump:fisher_info}
    Assume
    \begin{equation*}
        \sum_j\frac{\varphi_j^2}{\lambda_j} > 0,\quad \sum_j\frac{a_j^2\phi(a_j\theta_0+b_j)^2}{\Phi(a_j\theta_0+b_j)[1-\Phi(a_j\theta_0+b_j)]} > 0, \quad |\rho|<1.
    \end{equation*}
\end{assumption}

Assumption \ref{assump:finite} ensures that there is randomness in each entry, and the variance of the individual latent variables are bounded. Assumption \ref{assump:iden_theta} serves as an identifiability condition for $\bxi$ from the log posterior. That is, the ground truth $\bxi_0$ is unique and maximizes the log posterior when $J\to \infty$. Assumption \ref{assump:fisher_info} makes sure the Fisher information of both the likelihood and posterior is non-degenerate. These assumptions are mild and can be satisfied in common scenarios with finite population parameters in $\bOmega$. For a more detailed discussion, one can refer to Section 6 in \citet{kornely2022Asymptotic}. Next,  we establish the asymptotic distribution of the latent traits in Theorem 2 of the main text.

\subsection{Auxiliary Lemmas}
First, we present there the Kolmogorov's strong law of large numbers \citep{serfling2009approximation} for completeness.
\begin{theorem}[Kolmogorov's Strong Law of Large Numbers]
    \label{thm:k_slln}
    Let $\{X_i\}_{i\in\mathbb{N}}$ a sequence of independent random variables with $\EE{X_i}=\mu_i\in \mathbb{R}$ and $0<\Var{X_i}=\sigma_i^2< \infty$. If $\sum_{i=1}^{\infty}\sigma_i^2/i^2<\infty $, then almost surely
    \begin{equation*}
        \frac{1}{d}\sum_{i=1}^dX_i-\frac{1}{d}\sum_{i=1}^d\mu_i\to 0,
    \end{equation*}
    for $d\to \infty$.
\end{theorem}

Denote the probabilistic model as $P_{\xi_0}$, where $\xi_0=(\theta_0,\tau_0)$ the true value.
\begin{lemma}
    \label{lem:fix_xi}
    Let $\{R_j,T_j\}_{j\in \NN}$ be a set of data generated by fixed $\xi_0=(\theta_0,\tau_0)$. Under Assumption \ref{assump:finite} and \ref{assump:iden_theta},
    \begin{equation*}
        \limsup_{J\to \infty} \frac{1}{J}\left[ l^{(J)}(\xi\mid \Rb,\Tb)-l^{(J)}(\xi_0\mid \Rb, \Tb) \right] \leq c_1(\theta) + c_2(\tau)<0.
    \end{equation*}
\end{lemma}
\begin{proof}
    First, note that $l^{(J)}$ can be decomposed into two parts.
    \begin{align*}
        l^{(J)}(\xi\mid \Rb,\Tb) &= l^{(J)}_R(\theta) + l_T^{(J)}(\tau)\\
        &= \sum_{j=1}^J\left[R_{j}\log \Phi(a_j\theta+b_j)+(1-R_j)\log\Phi(-a_j\theta-b_j)\right]\\
        &- \sum_{j=1}^J \frac{1}{2\lambda_j}(\log T_{ij}+\varphi_j\tau-\omega_j)^2.
    \end{align*}

    Lemma W.2 in the web-appendix of \citet{kornely2022Asymptotic} shows that
    \begin{equation*}
        \limsup_{J\to \infty} \frac{1}{J}\left[ l_R^{(J)}(\theta\mid\Rb)-l_R^{(J)}(\theta_0\mid \Rb) \right]\leq c_1(\theta)<0.
    \end{equation*}
    Then, we focus on proving $l_T^{(J)}$. Define $Z_j=-[2\varphi_j(\tau-\tau_0)\log T_{ij} + \varphi_j^2(\tau^2-\tau_0^2) - 2\varphi_j\omega_j(\tau-\tau_0)]/2\lambda_j $, then $l_T^{(J)}(\tau\mid \Tb)-l_T^{(J)}(\tau_0\mid \Tb) = \sum_{j=1}^J Z_j $. Since $\log T_{ij}\sim N(\omega_j-\varphi_j\tau_0,\lambda_j)$, we have
    \begin{equation*}
        \EE{Z_j} = -\frac{1}{2\lambda_j}\left[2\varphi_j(\tau-\tau_0)(\omega_j-\varphi_j\tau_0) + \varphi_j^2(\tau^2-\tau_0^2) - 2\varphi_j\omega_j(\tau-\tau_0)\right],
    \end{equation*}
    \begin{align*}
        \Var{Z_j} = \frac{\varphi_j^2(\tau-\tau_0)^2}{\lambda_j^2}\Var{\log T_{j}} = \frac{\varphi_j^2(\tau-\tau_0)^2}{\lambda_j} \leq (\tau-\tau_0)^2\sup_{j\in \mathbb{N}}\frac{\varphi_j^2}{\lambda_j}.
    \end{align*}

    Therefore, under Assumption \ref{assump:finite},
    \begin{equation*}
        \sum_{j=1}^J \frac{\Var{Z_j}}{j^2} \leq (\tau-\tau_0)^2\sup_{j\in \mathbb{N}}\frac{\varphi_j^2}{\lambda_j} \sum_{j=1}^J\frac{1}{j^2} \leq \infty.
    \end{equation*}

    Then, by Kolmogorov's strong law of large numbers, we have,
    \begin{equation*}
        \frac{1}{J}\sum_{j=1}^J Z_j - \frac{1}{J}\sum_{j=1}^J\EEE_{\tau_0}[Z_j]\stackrel{a.s.}{\to} 0,\quad J\to \infty.
    \end{equation*}

    Under Assumption \ref{assump:iden_theta}, $\limsup_{J\to\infty}\sum_{j=1}^J\EEE_{\tau_0}[Z_j]/J\leq c_2(\tau)$, and hence
    \begin{equation*}
        \limsup_{J\to \infty } \frac{1}{J}\left[ l^{(J)}_T(\tau\mid \Tb) - l^{(J)}_T(\tau_0\mid \Tb) \right] \leq c_2(\tau) <0.
    \end{equation*}
\end{proof}

\begin{lemma}
    \label{lem:any_xi}
    Under Assumption \ref{assump:finite} and \ref{assump:iden_theta}, for any $\delta>0$, there exists a $k(\delta)<0$ so that
    \begin{equation*}
        \lim_{J\to \infty}P_{\xi_0}\left( \sup_{\xi\in \Theta\setminus B_{\delta}(\xi_0)} \frac{1}{J} \left( l^{(J)}(\xi\mid \Rb^{(J)},\Tb^{(J)}) - l^{(J)}(\xi_0\mid \Rb^{(J)},\Tb^{(J)}) \right) < k(\delta) \right) = 1.
    \end{equation*}
\end{lemma}
\begin{proof}
    Similarly as in the proof of Lemma \ref{lem:fix_xi}, the log likelihood can be decomposed into $l^{(J)}_R$ and $l_T^{(J)}$. The bound for $l^{(J)}_R$ is shown in Lemma 1 in \citet{kornely2022Asymptotic}. Here, we focus on the proving the following argument,
    \begin{equation}
        \label{eqn:any_xi_tau_target}
        \lim_{J\to \infty} P_{\tau_0}\left(\sup_{\xi \in \Theta\setminus B_\delta(\xi_0)}\frac{1}{J}\left( l_T^{(J)}(\tau\mid \Tb^{(J)}) - l_T^{(J)}(\tau_0\mid \Tb^{(J)}) \right) < k_2(\delta) \right) = 1.
    \end{equation}

    Before digging into the proof, we first show combining Lemma 1 in \citet{kornely2022Asymptotic} and (\ref{eqn:any_xi_tau_target}), we obtain the desired result. For simplicity, denote $A_J$ the series of events in Lemma 1 in \citet{kornely2022Asymptotic}, and $B_J$ the series of events in (\ref{eqn:any_xi_tau_target}). We know $\lim_{J\to \infty} P(A_J)=1 $, $\lim_{J\to\infty} P(B_J)=1$. Then,
    \begin{equation*}
        \lim_{J\to \infty} P(A_J^{C}\cup B_J^C)\leq \lim_{J\to\infty} P(A_J^C) + \lim_{J\to \infty}P(B_J^C) = 0.
    \end{equation*}
    Therefore, $\lim_{J\to\infty}P(A_J\cap B_J)=1$, and we obtain the desired result.

    To prove (\ref{eqn:any_xi_tau_target}), we have the following decomposition. For any $\tau_i\neq \tau_0$, sufficiently small $\delta_i>0$,
    \begin{align*}
        \frac{1}{J}\left( l_T^{(J)}(\tau\mid \Tb^{(J)}) - l_T^{(J)}(\tau_0\mid \Tb^{(J)}) \right) &= \underbrace{\frac{1}{J}\left( l_T^{(J)}(\tau\mid \Tb^{(J)}) - l_T^{(J)}(\tau_i\mid \Tb^{(J)}) \right)}_{\alpha_1}\\
        &+ \underbrace{\frac{1}{J} \left( l_T^{(J)}(\tau_i\mid \Tb^{(J)}) - l_T^{(J)}(\tau_0\mid \Tb^{(J)}) \right)}_{\alpha_2}.
    \end{align*}

    For $\alpha_2$, from Lemma \ref{lem:fix_xi}, we have
    \begin{equation*}
        \limsup_{J\to\infty} \frac{1}{J} \left( l_T^{(J)}(\tau_i\mid \Tb^{(J)}) - l_T^{(J)}(\tau_0\mid \Tb^{(J)}) \right) \leq c_2(\tau_0) < 0,\quad P_{\xi_0}-a.s.
    \end{equation*}

    In the sequel, we will bound $\alpha_1$. Consider $\tau\in \bar{B}_{\delta_i}(\tau_i) $, we first bound $\sup_{\tau \in B_{\delta_i}(\tau_i)}(l_T^{(J)}(\tau\mid\Tb^{(J)})-l_T^{(J)}(\tau_i\mid\Tb^{(J)}))/J$. Define $Z_j=[2\varphi_j(\tau-\tau_i)\log T_{j}+\varphi_j^2(\tau^2-\tau_i^2)-2\varphi_j\omega_j(\tau-\tau_i)]/(2\lambda_j) $, and $l_T^{(J)}(\tau\mid\Tb^{(J)})-l_T^{(J)}(\tau_i\mid\Tb^{(J)}) = \sum_{j=1}^JZ_j$. Then, we bound $|\sum_{j=1}^JZ_j|/J$.
    \begin{align*}
        \frac{1}{J}\left|\sum_{j=1}^JZ_j\right|\leq \frac{1}{J}|\tau-\tau_i|\Biggl[ \underbrace{\left|\sum_{j=1}^J\frac{\varphi_j}{\lambda_j}(\log T_{j}-\omega_j+\varphi_j\tau) \right|}_{\beta_1} + \underbrace{|\tau-\tau_i|\left|\sum_{j=1}^J\frac{\varphi_j^2}{2\lambda_j} \right|}_{\beta_2} \Biggr].
    \end{align*}

    For $\beta_1$, note that $\varphi_j(\log T_j-\omega_j+\varphi_j\tau)/\lambda_j\sim N(0,\varphi_j^2/\lambda_j)$, and $\log T_j$s' are independent. Hence,
    \begin{equation*}
        \sum_{j=1}^J\frac{\varphi_j}{\lambda_j}(\log T_j-\omega_j+\varphi_j\tau) \sim N\left(0,\sum_{j=1}^J \frac{\varphi_j^2}{\lambda_j}\right).
    \end{equation*}
    By Assumption \ref{assump:finite}, $\sum_{j=1}^J\varphi_j^2/\lambda_j\leq CJ$.
    Then, by standard Gaussian tail bound, with probability $1-O(J^{-8})$,
    \begin{equation*}
        \beta_1 \leq 4\sqrt{CJ\log J}.
    \end{equation*}

    For $\beta_2$, by Assumption \ref{assump:finite},
    \begin{equation*}
        \beta_2\leq \frac{CJ}{2}|\tau-\tau_i|.
    \end{equation*}

    Moreover, $|\tau-\tau_i|\leq \delta_i$. Thus, with probability at least $1-O(J^{-8})$,
    \begin{equation*}
        \frac{1}{J}\left|\sum_{j=1}^JZ_j\right| \leq \delta_i\left[4\sqrt{\frac{C\log J}{J}} + \frac{C}{2}\delta_i \right].
    \end{equation*}

    Therefore,
    \begin{equation*}
        \lim_{\delta\to 0}\sup_{\xi\in \bar{B}_{\delta}(\xi_i)}\frac{1}{J}\left| l_T^{(J)}(\tau\mid\Tb^{(J)})-l_T^{(J)}(\tau_i\mid\Tb^{(J)}) \right| = 0.
    \end{equation*}
    Let $\varepsilon=-c_2(\tau_0)/2$, $\exists \delta_i>0$ and $c_i=c_2(\tau_0)/2$,
    \begin{equation*}
        \lim_{J\to\infty} P_{\bxi_0}\left( \sup_{\xi\in \bar{B}_{\delta_i}(\bxi_i)} \frac{1}{J}\left(l_T^{(J)}(\tau\mid\Tb^{(J)})-l_T^{(J)}(\tau_i\mid\Tb^{(J)})\right)<c_i<0\right) = 1.
    \end{equation*}

    Next, we first show the result assuming $\Theta$ is compact. Then, we extend the result to unbounded $\Theta$. For all $\delta>0$, $\Theta\setminus B_{\delta}(\bxi_0) $ is still compact. For each $\delta^{\prime}<\delta$, $\cup_{\bxi\in\Theta\setminus B_{\delta}(\bxi_0) } B_{\delta^{\prime}}(\bxi)$ is a cover for $\Theta\setminus B_{\delta}(\bxi_0)$. Hence, there exists a finite cover $B_{\delta}(\bxi_1), \ldots , B_{\delta}(\bxi_n) $ that form a cover of $\Theta\setminus B_{\delta}(\bxi_0) $.

    For each $B_{\delta}(\bxi_{k})$, there exists $c_k<0$, such that
    \begin{equation*}
        \lim_{J\to \infty} P_{\bxi_0}\left( \sup_{\xi\in \bar{B}_{\delta_i}(\bxi_i)} \frac{1}{J}\left(l_T^{(J)}(\tau\mid\Tb^{(J)})-l_T^{(J)}(\tau_i\mid\Tb^{(J)})\right)<c_k<0\right) = 1.
    \end{equation*}

    let $k = \max_{m\in [n]}c_m$, $\forall \bxi \in \Theta\setminus B_{\delta}(\bxi_0) $, by union bound
    \begin{align*}
        &\lim_{J\to \infty} P\left(\sup_{\bxi\in \Theta\setminus B_{\delta}(\bxi_0)} \frac{1}{J}\left(l_T^{(J)}(\tau\mid\Tb^{(J)})-l_T^{(J)}(\tau_i\mid\Tb^{(J)})\right)\geq k\right)\\&\leq \lim_{J\to \infty}\sum_{m=1}^n P\left(\sup_{\bxi\in B_{\delta^{\prime}}(\bxi_m)} \frac{1}{J}\left(l_T^{(J)}(\tau\mid\Tb^{(J)})-l_T^{(J)}(\tau_i\mid\Tb^{(J)})\right)\geq k\right)\\
        &= 0.
    \end{align*}

    Hence, for every compact $\Theta$, we have the desired result. Then, we extend the result to unbounded $\Theta$. Define $\Theta^{(j)}=\{ (\theta,\tau)\in\Theta: \delta+j\leq |\theta-\theta_0|\leq \delta+j+1, \quad \delta+j\leq|\tau-\tau_0|\leq \delta+j+1 \}$, $j\in \mathbb{N}$. Each $\Theta^{(j)}$ is compact and enjoys the above property. Recall the definition of $k\leq \sup_{\bxi\in\Theta\setminus B_{\delta}(\bxi_0)}c_2(\tau)/2$. Hence, let $k_j=\sup_{\bxi\in \Theta^{(j)}}c_2(\tau)/2 $, we have
    \begin{align*}
        &\sup_{\bxi\in \Theta\setminus B_{\delta}(\bxi_0)} \frac{1}{J}
        \left(l_T^{(J)}(\tau\mid\Tb^{(J)})-l_T^{(J)}(\tau_i\mid\Tb^{(J)})\right)\\
        &\qquad= \sup_{j\in \mathbb{N}}\left\{ \sup_{\bxi\in \Theta^{(j)}} \frac{1}{J}
        \left(l_T^{(J)}(\tau\mid\Tb^{(J)})-l_T^{(J)}(\tau_i\mid\Tb^{(J)})\right)\right\},\\
        \sup_{j\in \mathbb{N}} k_j &\leq \sup_{\bxi\in\Theta\setminus B_{\delta}(\bxi_0)}c_2(\tau)/2:=k_2(\delta).
    \end{align*}

    Therefore,
    \begin{align*}
        &\lim_{J\to \infty} P_{\bxi_0}\left(\sup_{\bxi \in \Theta\setminus B_\delta(\bxi_0)}\frac{1}{J}\left( l_T^{(J)}(\tau\mid \Tb^{(J)}) - l_T^{(J)}(\tau_0\mid \Tb^{(J)}) \right) < k_2(\delta) \right)\\
        &\geq \lim_{J\to\infty} P_{\bxi_0}\left(\sup_{j\in \mathbb{N}}\left( \sup_{\bxi\in \Theta^{(j)}} \frac{1}{J}\left(l_T^{(J)}(\tau\mid\Tb^{(J)})-l_T^{(J)}(\tau_i\mid\Tb^{(J)})\right)\right)\leq \sup_{j\in\NN}k_j\right)=1.
    \end{align*}
    \end{proof}

    \begin{lemma}
    \label{lem:consistency}
        \begin{enumerate}
            \item[(1)] There exists $\hat{\bxi}=(\hat{\theta},\hat{\tau})$ such that
            \begin{align}
                \label{eqn:MLE_cond}
                \lim_{J\to \infty} P_{\bxi_0}\left( \nabla l^{(J)}(\hat{\theta}, \hat{\tau}\mid \Rb^{(J)}, \Tb^{(J)})=0 \right) &= 1,\\
                \label{eqn:MLE_reach}
                \lim_{J\to\infty}P_{\bxi_0}\left( l^{(J)}(\hat{\theta},\hat{\tau}\mid \Rb^{(J)},\Tb^{(J)}) = \max_{\bxi\in\Theta}l^{(J)}(\bxi\mid\Rb^{(J)},\Tb^{(J)}) \right)&=1,\\
                \label{eqn:consistent}
                (\hat{\theta},\hat{\tau}) &\stackrel{p}{\to} (\theta,\tau), \quad J\to \infty.
            \end{align}
            \item[(2)] There exists $\tilde{\bxi}=(\tilde{\theta},\tilde{\tau})$ such that when the log likelihood $l$ is replaced by log posterior $\tilde{l}$, the above result still holds.
        \end{enumerate}
    \end{lemma}
    \begin{proof}
        \textbf{(1)} First, we show the existence of such solution.

        Define
        \begin{equation*}
            A_{\delta,\varepsilon,J} = \left\{ \sup_{\bxi:\|\bxi-\bxi_0\|\geq \delta} \frac{1}{J}\left( l^{(J)}(\bxi\mid \Rb^{(J)}, \Tb^{(J)})- l^{(J)}(\bxi_0\mid \Rb^{(J)}, \Tb^{(J)}) \right) < \varepsilon \right\}.
        \end{equation*}
        By Lemma \ref{lem:any_xi}, $\lim_{J\to\infty}P_{\bxi_0}(A_{\delta,\varepsilon,J})=1$, for all $\delta > 0$ and $\varepsilon>0$. Given $A_{\delta,\varepsilon,J} $, the global minimum of the log likelihood must lie in $B_{\delta}(\bxi_0)$. Next, we construct a measurable mapping from $\left(\{0,1\}^J\times \RR^{J},\text{Pow}(\{0,1\}^J)\otimes \cB(\RR^J)\right) \to (\Theta, \cB(\Theta)) $. $\text{Pow}(\{0,1\}^J)$ denotes the power set of $\{0,1\}^J$.

        Note that $l^{(J)}(\cdot \mid \Rb^{(J)}, \Tb^{(J)}) $ is continuous for every fixed $\Rb^{(J)} $, $\Tb^{(J)}$, and $l^{(J)}(\bxi\mid \cdot) $ is continuous for every fixed $\bxi$. Let $\Theta_{\delta}=\bar{B}_{\delta}(\bxi_0)\cap \Theta$. For simplicity, we assume $\Theta_{\delta}$ is compact. If $\Theta$ is unbounded, similar techniques as in the proof of Lemma \ref{lem:any_xi} can be applied similarly, and we omit it here.

        By continuity, there exists $\bxi^*$ such that,
        \begin{equation*}
            l^{(J)}(\bxi^*\mid \Rb^{(J)},\Tb^{(J)})=\sup_{\bxi \in \Theta_{\delta}} l^{(J)}(\bxi \mid \Rb^{(J)},\Tb^{(J)}).
        \end{equation*}
        Then by Lemma W.3 in \citet{kornely2022Asymptotic}, there exists a measurable mapping $\check{\boldsymbol{\xi}}_J$, such that $\bxi^*=\check{\boldsymbol{\xi}}_J(\Rb^{(J)},\Tb^{(J)})$. By Lemma \ref{lem:any_xi}, let $\hat{\bxi}_J = \check{\bxi}_{J}(\Rb^{(J)},\Tb^{(J)})$, we have a sequence $\hat{\bxi}_J$ that satisfies (\ref{eqn:MLE_cond}) and (\ref{eqn:MLE_reach}).

        For (\ref{eqn:consistent}), we prove by contradiction. Suppose $\hat{\bxi}_J$ is not consistent. There exists $\varepsilon_0>0$, for all $\delta_0>0$, $\forall J\in \NN$, $P(\|\hat{\bxi}_J-\bxi_0\|>\delta_0) \geq \varepsilon_0$. Let $\tilde{\delta} = \delta_0/2$, from Lemma \ref{lem:any_xi}, we have
        \begin{equation*}
            \lim_{J\to \infty} P\left( \sup_{\bxi\in\Theta \setminus B_{\tilde{\delta}}(\bxi_0)} \frac{1}{J}\left( l^{(J)}(\bxi\mid \Rb^{(J)},\Tb^{(J)}) - l^{(J)}(\bxi_0\mid \Rb^{(J)}, \Tb^{(J)})\right) < c(\tilde{\delta})<0 \right) = 1.
        \end{equation*}

        Let
        \begin{equation*}
            A_{J,\tilde{\delta}}=\left\{ \frac{1}{J}\left( l^{(J)}(\bxi\mid \Rb^{(J)},\Tb^{(J)}) - l^{(J)}(\bxi_0\mid \Rb^{(J)}, \Tb^{(J)}) \right)< c(\tilde{\delta})<0 \right\}.
        \end{equation*}
        For all $\tilde{\varepsilon}$, there exists $J_0>0$, $\forall J>J_0$, $P(A_{J,\tilde{\delta}})>1-\tilde{\varepsilon}$. Let $\tilde{\varepsilon} = \varepsilon_0/2$. Since $A_{J,\tilde{\delta}}\cap \{\|\bxi^*_J-\bxi_0\| > \delta_0\}=\emptyset$, $\{\|\bxi^*_J-\bxi_0\| > \delta_0\}\subseteq A_{J,\delta_0}^c$. Thus,
        \begin{equation*}
            P\left( \|\hat{\bxi}_J-\bxi_0\| > \delta_0 \right) \leq P(A_{J,\tilde{\delta}}^c) < \varepsilon_0,
        \end{equation*}
        for all $J>J_0$. There is a contradiction and $\hat{\bxi}_J$ is consistent.

        \textbf{(2)} The key difference for $\tilde{\bxi}_J$ to satisfy (\ref{eqn:MLE_cond})-(\ref{eqn:consistent}) is to show an equivalent version of Lemma \ref{lem:any_xi} for $\tilde{l}^{(J)}$. For $\tilde{l}^{(J)} $, there is the following decomposition,
        \begin{align*}
            &\frac{1}{J}\left[ \tilde{l}^{(J)}(\bxi\mid \Rb^{(J)}, \Tb^{(J)})-\tilde{l}^{(J)}(\bxi_0\mid \Rb^{(J)}, \Tb^{(J)}) \right]\\
            &= \underbrace{\frac{1}{J}\left[ {l}^{(J)}(\bxi\mid \Rb^{(J)}, \Tb^{(J)})-{l}^{(J)}(\bxi_0\mid \Rb^{(J)}, \Tb^{(J)}) \right]}_{\alpha_1} + \underbrace{\frac{1}{J}\left[ -\frac{1}{2}\bxi^\top\bSigma^{-1}\bxi + \frac{1}{2}\bxi_0^\top \bSigma^{-1}\bxi_0 \right]}_{\alpha_2} .
        \end{align*}

        We have shown $\alpha_1$ in Lemma \ref{lem:any_xi}. For $\alpha_2$, since $\bSigma$ is positive definite, $-\bxi^\top\bSigma^{-1}\bxi/2\leq 0 $. Because $\bxi_0$ and $\bSigma$ are constants, $\forall \tilde{\varepsilon} > 0$, $\exists J_0>0$, for $\forall J > J_0$,
        \begin{equation*}
            \sup_{\bxi \in \Theta\setminus B_{\delta}(\bxi_0)}\frac{1}{J}\left[ -\frac{1}{2}\bxi^\top\bSigma^{-1}\bxi + \frac{1}{2}\bxi_0^\top \bSigma^{-1}\bxi_0 \right] \leq \tilde{\varepsilon}.
        \end{equation*}

        Let $\tilde{\varepsilon}=k(\delta)/2$, $\tilde{k}(\delta)=k(\delta)/2$, then
        \begin{equation*}
            \lim_{J\to \infty}P_{\bxi_0}\left( \sup_{\xi\in \Theta\setminus B_{\delta}(\xi_0)} \frac{1}{J} \left( \tilde{l}^{(J)}(\bxi\mid \Rb^{(J)},\Tb^{(J)}) - \tilde{l}^{(J)}(\bxi_0\mid \Rb^{(J)},\Tb^{(J)}) \right) < \tilde{k}(\delta) \right) = 1.
        \end{equation*}

        The following proof is the same as the proof in (1).
    \end{proof}

    \begin{lemma}
        \label{lem:taylor_exp}
        \begin{enumerate}
            \item[1.] For any $\bxi\in \Theta$, there exists $\{a_J\}_{J\in \NN}$, $a_J\in [0,1] $, such that
            \begin{align*}
                &\tilde{l}^{(J)}(\bxi\mid \Rb^{(J)},\Tb^{(J)}) - \tilde{l}^{(J)}(\tilde{\bxi}_J\mid \Rb^{(J)}, \Tb^{(J)})\\ &= \frac{1}{2}\left( \bxi - \tilde{\bxi}_J \right)^{\top} \tilde{\Hb}_J(\bxi_J^*)\left( \bxi - \tilde{\bxi}_J \right)\\
                &= -\frac{1}{2}\left( \bxi - \tilde{\bxi}_J \right)^{\top} \left[ \cI_{J}(\tilde{\bxi}_J)(\Ib_2-E_J(\bxi))+\bSigma^{-1} \right]\left( \bxi - \tilde{\bxi}_J \right),
            \end{align*}
            where $\bxi_J^*=a_J\tilde{\bxi}_J+(1-a_J)\bxi$, $\tilde{\Hb}_J$ is the Hessian of the log posterior, $E_J = \Ib_K + \cI_{J}(\tilde{\bxi}_J)^{-1}\Hb_{J}(\bxi_J^*)$, and $\Hb_J$ is the Hessian of the log likelihood.
            \item[2.] For any $\varepsilon>0$, there is $\delta>0$, such that
            \begin{equation*}
                \lim_{J\to \infty} P_{\bxi_0}\left( \sup_{\bxi\in B_{\delta}(\bxi_0)}\|E_J(\bxi)\|<\varepsilon \right) = 1.
            \end{equation*}
            \item[3.] $\forall \varepsilon>0$, $\exists \delta>0$, for all $\bxi \in B_{\delta}(\bxi_0) $,
            \begin{align*}
                \lim_{J\to \infty} P_{\bxi_0}\biggl( (1+\varepsilon) \tilde{V}_J(\bxi)&\leq -\frac{1}{2}\left( \bxi - \tilde{\bxi}_J \right)^{\top} \left[ \cI_{J}(\tilde{\bxi}_J)(\Ib_2-E_J(\bxi))+\bSigma^{-1} \right]\left( \bxi - \tilde{\bxi}_J \right)\\
                &\leq (1-\varepsilon)\tilde{V}_J(\bxi) \biggr) = 1,
            \end{align*}
            where $\tilde{V}_J(\bxi) = -\frac{1}{2}\left( \bxi - \tilde{\bxi}_J \right)^{\top} \tilde{\cI}_{J}(\tilde{\bxi}_J)\left( \bxi - \tilde{\bxi}_J \right)$.
        \end{enumerate}
    \end{lemma}
    \begin{proof}
        \textbf{(1)} The inequality directly comes from Taylor Expansion with Cauchy form of the remainder. We omit the detailed algebraic computation here.

        \textbf{(2)} First, since
        \begin{equation*}
            \frac{\partial^2 l^{(J)}}{\partial \theta \partial \tau} = 0,
        \end{equation*}
        the Hessian of the log likelihood is a diagonal matrix. Hence, the Fisher information of the log likelihood is also a diagonal matrix. By Assumption \ref{assump:fisher_info}, because the diagonal entries are both greater than 0, the Fisher information of the log likelihood is full rank of 2. Additionally, since $\|\cI_J(\tilde{\bxi}_J)/J\|^{-1} = 1/\sigma_{\min}(\cI_J(\tilde{\bxi}_J)/J)$, there exists constant $C_0>0$, such that $\|\cI_J(\tilde{\bxi}_J)/J\|^{-1}\leq 1/C_0$, when $J$ is sufficiently large due to consistency of $\tilde{\bxi}_J$ shown in Lemma \ref{lem:consistency}.

        Hence,
        \begin{align*}
            \|E_J(\bxi)\| &= \left\| \left(\frac{1}{J}\cI_J(\tilde{\bxi}_J)\right)^{-1} \frac{1}{J}\left(\cI_J(\tilde{\bxi}_J)+\nabla^2l^{(J)}(\bxi\mid \Rb^{(J)}, \Tb^{(J)}) \right) \right\| \\
            &\leq \left\|\frac{1}{J}\cI_J(\tilde{\bxi}_J)\right\|^{-1} \left\| \frac{1}{J}\left(\cI_J(\tilde{\bxi}_J)+\nabla^2l^{(J)}(\bxi\mid \Rb^{(J)}, \Tb^{(J)}) \right) \right\|\\
            &\leq \frac{1}{C_0}\left\| \frac{1}{J}\left(\cI_J(\tilde{\bxi}_J)+\nabla^2l^{(J)}(\bxi\mid \Rb^{(J)}, \Tb^{(J)}) \right) \right\|\\
            &\leq \frac{1}{C_0}\max\left\{\alpha_1,\alpha_2\right\},
        \end{align*}
        where
        \begin{align*}
            \alpha_1 &= \frac{1}{J}\left[\sum_{j=1}^J \frac{a_j^2\phi(a_j\tilde{\theta}+b_j)}{\Phi(a_j\tilde{\theta}+b_j)[1-\Phi(a_j\tilde{\theta}+b_j)]} - \sum_{j=1}^J \frac{a_j^2\phi(a_j{\theta}+b_j)}{\Phi(a_j{\theta}+b_j)[1-\Phi(a_j{\theta}+b_j)]}\right]\\
            \alpha_2 &= \frac{1}{J}\left[ \sum_{j=1}^J\frac{\varphi_j^2}{\lambda_j} - \sum_{j=1}^J\frac{\varphi_j^2}{\lambda_j} \right] = 0.
        \end{align*}
        The last inequality follows because the Fisher information and the Hessian of the log likelihood are diagonal matrices.

        Since $\alpha_2=0$, we only need to bound $\alpha_1$. Following the same proof of Lemma 2 in \citet{kornely2022Asymptotic}, for any $\varepsilon>0$, when $J$ is sufficiently large, there exists $\delta>0$, such that
        \begin{equation*}
            \lim_{J\to\infty} P_{\bxi_0}\left( \sup_{\bxi \in B_{\delta}(\bxi_0)} \frac{1}{C_0}\alpha_1 < \varepsilon \right) = 1.
        \end{equation*}

        Therefore, the second part of Lemma \ref{lem:taylor_exp} is proven.

        \textbf{(3)} For the result in the third part, we first show $\forall \varepsilon >0$, $\exists \delta > 0$, such that,
        \begin{equation*}
            \lim_{J\to \infty} P_{\bxi_0}\left( \left| (\bxi-\tilde{\bxi}_J)^\top\cI_J(\tilde{\bxi}_J)E_J(\bxi)(\bxi-\tilde{\bxi}_J) \right|\leq -2\varepsilon V_J(\bxi) \right) = 1,
        \end{equation*}
        where ${V}_J(\bxi) = -\frac{1}{2}( \bxi - \tilde{\bxi}_J )^{\top} {\cI}_{J}(\tilde{\bxi}_J)( \bxi - \tilde{\bxi}_J )$

        By Lemma W.5 in \citet{kornely2022Asymptotic},
        \begin{equation*}
            \left|  \frac{1}{2}(\bxi-\tilde{\bxi}_J)^\top\cI_J(\tilde{\bxi}_J)E_J(\bxi)(\bxi-\tilde{\bxi}_J) \right| \leq -\kappa\bigl(\cI_J(\tilde{\bxi}_J)\bigr) \|E_J(\bxi)\|V_J(\bxi).
        \end{equation*}

        By Assumption \ref{assump:fisher_info} and Lemma \ref{lem:consistency} and continuous mapping theorem, there exists $C_1^{\prime}$ such that
        \begin{equation*}
            P_{\bxi_0}\left( \limsup_{J\to \infty} \kappa\bigl(\cI_J(\tilde{\bxi}_J)\bigr) \leq C_1^{\prime} \right) = 1.
        \end{equation*}

        Additionally, from (2) of this Lemma, $\|E_J(\bxi)\|$ converges to 0 in probability. Hence,
        \begin{equation*}
            \lim_{J\to \infty} P_{\bxi_0}\left( \left| (\bxi-\tilde{\bxi}_J)^\top\cI_J(\tilde{\bxi}_J)E_J(\bxi)(\bxi-\tilde{\bxi}_J) \right|\leq -2\varepsilon V_J(\bxi) \right) = 1.
        \end{equation*}

        Therefore, under $\{ | (\bxi-\tilde{\bxi}_J)^\top\cI_J(\tilde{\bxi}_J)E_J(\bxi)(\bxi-\tilde{\bxi}_J) |\leq -2\varepsilon V_J(\bxi)\}$,
        \begin{align*}
            &-\frac{1}{2}\left( \bxi - \tilde{\bxi}_J \right)^{\top} \left[ \cI_{J}(\tilde{\bxi}_J)(\Ib_2-E_J(\bxi))+\bSigma^{-1} \right]\left( \bxi - \tilde{\bxi}_J \right)\\
            &\leq V_J(\bxi)+ \frac{1}{2}\left(\bxi-\tilde{\bxi}_J\right)^\top \cI_{J}(\tilde{\bxi}_J)E_J(\bxi)\left(\bxi-\tilde{\bxi}_J\right) - \frac{1}{2}\left( \bxi-\tilde{\bxi}_J \right)^{\top} \bSigma^{-1} \left( \bxi - \tilde{\bxi}_J \right)\\
            &\leq (1-\varepsilon)V_J(\bxi) - \frac{1}{2}\left( \bxi-\tilde{\bxi}_J \right)^{\top} \bSigma^{-1} \left( \bxi - \tilde{\bxi}_J \right).
        \end{align*}

        Since $\bSigma$ is positive definite,
        \begin{align*}
            &(1-\varepsilon)\left( \bxi-\tilde{\bxi}_J \right)^{\top} \bSigma^{-1}
              \left( \bxi - \tilde{\bxi}_J \right)\\
            &\quad\leq \left( \bxi-\tilde{\bxi}_J \right)^{\top} \bSigma^{-1}
              \left( \bxi - \tilde{\bxi}_J \right)\\
            &\quad\leq (1+\varepsilon)\left( \bxi-\tilde{\bxi}_J \right)^{\top} \bSigma^{-1}
              \left( \bxi - \tilde{\bxi}_J \right).
        \end{align*}
        Also, note that
        \begin{equation*}
            \tilde{V}_J(\bxi) = V_J(\bxi) - \frac{1}{2}\left( \bxi-\tilde{\bxi}_J \right)^{\top} \bSigma^{-1} \left( \bxi - \tilde{\bxi}_J \right).
        \end{equation*}

        Therefore,
        \begin{align*}
            (1-\varepsilon)\tilde{V}_J(\bxi) &\geq (1-\varepsilon)V_J(\bxi) - \frac{1}{2}\left( \bxi-\tilde{\bxi}_J \right)^{\top} \bSigma^{-1} \left( \bxi - \tilde{\bxi}_J \right),\\
            (1+\varepsilon)\tilde{V}_J(\bxi) &\leq (1+\varepsilon)V_J(\bxi) - \frac{1}{2}\left( \bxi-\tilde{\bxi}_J \right)^{\top} \bSigma^{-1} \left( \bxi - \tilde{\bxi}_J \right).
        \end{align*}

        Hence,
        \begin{align*}
            -\frac{1}{2}\left( \bxi - \tilde{\bxi}_J \right)^{\top} \left[ \cI_{J}(\tilde{\bxi}_J)(\cI_2-E_J(\bxi))+\bSigma^{-1} \right]\left( \bxi - \tilde{\bxi}_J \right) \leq (1-\varepsilon)\tilde{V}_J(\bxi).
        \end{align*}
        The other side of the inequality holds similarly.
    \end{proof}

    \begin{lemma}
        \label{lem:converge}
        Let ${\Phi}(B) = P(Z\in B)$, where $Z\sim N(0,\cI)$. Under Assumption \ref{assump:finite}, \ref{assump:iden_theta}, \ref{assump:fisher_info}
        \begin{enumerate}
            \item[1.] For every function $f$ that the integral $\int_{\Theta}f(\bxi)\pi(\bxi)d\bxi$ exists, for every $\delta > 0$, we have,
            \begin{equation*}
                \frac{\int_{\Theta\setminus B_{\delta}(\bxi_0)}f(\bxi)P^{(J)}(\Rb^{J}, \Tb^{(J)}\mid \bxi)\pi(\bxi)d\bxi }{P^{(J)}(\Rb^{(J)},\Tb^{(J)}\mid \tilde{\bxi}_J)}\det(\tilde{\cI}_J(\tilde{\bxi}_J))^{1/2} \stackrel{P_{\bxi_0}}{\to} 0, \quad J\to \infty.
            \end{equation*}
            \item[2.] Consider a sequence of mappings $\{G_J\}_{J\in \NN} $, $G_J: \bigl(\Theta,\cB(\Theta)\bigr)\to \bigl(\Theta, \cB(\Theta)\bigr) $ satisfying either of the following condition
            \begin{align}
                \label{eqn:cond_1}
                \lim_{J\to \infty}P_{\bxi_0}\left( G_J(B) \subseteq B_{\delta}(\bxi_0) \right) &= 1, \quad \forall \delta > 0,\\
                \label{eqn:cond_2}
                \lim_{J\to \infty}P_{\bxi_0}\left( G_J(B) \supseteq B_{\delta}(\bxi_0) \right) &= 1, \quad \forall \delta>0,
            \end{align}
            for all bounded $B\in \cB(\Theta)$. Then,
            \begin{align*}
                &\frac{\int_{G_J(B)} P^{(J)}( \Rb^{(J)}, \Tb^{(J)}\mid \bxi ) \pi(\bxi) d\bxi }
                {P^{(J)}(\Rb^{(J)},\Tb^{(J)}\mid \tilde{\bxi}_J)}
                \det(\tilde{\cI}_J(\tilde{\bxi}_J))^{1/2}\\
                &\quad- {\Phi}\left( \tilde{\cI}(\tilde{\bxi}_J)^{1/2}
                (G_J(B)-\tilde{\bxi}_J) \right)\pi(\bxi_0)(2\pi)
                = o_{P_{\bxi_0}}(1).
            \end{align*}
        \end{enumerate}
    \end{lemma}
    \begin{proof}
        \textbf{(1)} First, note that
        \begin{align*}
            &\frac{\int_{\Theta\setminus B_{\delta}(\bxi_0)} f(\bxi)P^{(J)}(\Rb^{(J)}, \Tb^{(J)} \mid \bxi) \pi(\bxi)d\bxi }{\pi(\tilde{\bxi}_J)P^{(J)}(\Rb^{(J)}, \Tb^{(J)}\mid \tilde{\bxi}_J)}\det(\tilde{\cI}_J(\tilde{\bxi}_J))^{1/2}\\
            &= \exp \left( \tilde{l}^{(J)}(\bxi_0\mid \Rb^{(J)}, \Tb^{(J)}) - \tilde{l}^{(J)}(\tilde{\bxi}_J\mid \Rb^{(J)},\Tb^{(J)}) \right) \tilde{L}_J \det(\tilde{\cI}_J(\tilde{\bxi}_J))^{1/2},
        \end{align*}
        where
        \begin{equation*}
            \tilde{L}_J = \int_{\Theta\setminus B_{\delta}(\bxi_0)}\exp \left( \tilde{l}^{(J)}(\bxi_0\mid \Rb^{(J)}, \Tb^{(J)}) - \tilde{l}^{(J)}(\tilde{\bxi}_J\mid \Rb^{(J)},\Tb^{(J)}) \right) f(\bxi) d\bxi.
        \end{equation*}

        Since $\tilde{\bxi}_J$ is a maximum of $\tilde{l}^{(J)} $, one has
        \begin{equation*}
            \exp \left( \tilde{l}^{(J)}(\bxi_0\mid \Rb^{(J)}, \Tb^{(J)}) - \tilde{l}^{(J)}(\tilde{\bxi}_J\mid \Rb^{(J)},\Tb^{(J)}) \right) \leq 1.
        \end{equation*}

        Hence,
        \begin{equation*}
            \left|\frac{\int_{\Theta\setminus B_{\delta}(\bxi_0)} f(\bxi)P^{(J)}(\Rb^{(J)}, \Tb^{(J)} \mid \bxi) \pi(\bxi)d\bxi }{\pi(\tilde{\bxi}_J)P^{(J)}(\Rb^{(J)}, \Tb^{(J)}\mid \tilde{\bxi}_J)}\det(\tilde{\cI}_J(\tilde{\bxi}_J))^{1/2}\right| \leq \left|\tilde{L}_J\det(\tilde{\cI}_J(\tilde{\bxi}_J))^{1/2}\right|.
        \end{equation*}

        For the determinant,
        \begin{align*}
            \det(\tilde{\cI}_J(\tilde{\bxi}_J))^{1/2} &= \sqrt{\det\left(\cI_J(\tilde{\bxi}_J) + \bSigma^{-1}\right)}\\
            &= J\sqrt{\det\left(\frac{1}{J}\cI_J(\tilde{\bxi}_J) + \frac{1}{J}\bSigma^{-1}\right) }\\
            &\leq J \sigma_1\left(\frac{1}{J}\cI_J(\tilde{\bxi}_J) + \frac{1}{J}\bSigma^{-1}\right)\\
            &\leq J\left[\sigma_1\left(\frac{1}{J}\cI_J(\tilde{\bxi}_J)\right)+\sigma_1\left( \frac{1}{J}\bSigma^{-1} \right)\right],
        \end{align*}
        where the last inequality comes from Weyl's inequality. By Assumption \ref{assump:finite}, $\sigma_1(\cI_J(\tilde{\bxi}_J)/J)$ is bounded by some constant. Since $\bSigma$ is a constant, $\sigma_1(\bSigma^{-1}/J)=O(1/J)$. Hence,
        \begin{equation*}
            \det(\tilde{\cI}_J(\tilde{\bxi}_J))^{1/2} = O_{P_{\bxi_0}}(J).
        \end{equation*}

        Since $\pi(\bxi)$ is proper and has support over $\Theta$,
        \begin{equation*}
            \frac{1}{P^{(J)}\left(\Rb^{(J)}, \Tb^{(J)}\mid \bxi_0\right)} \int_{\Theta\setminus B_{\delta}(\bxi_0)} P^{(J)}(\Rb^{(J)},\Tb^{(J)}\mid \bxi)\pi(\bxi)d\bxi = o_{P_{\bxi_0}}(J^{-1}),\quad \forall \delta > 0.
        \end{equation*}
        For detailed discussion, one can refer to Equation (28) in \citet{kornely2022Asymptotic}. Additionally, suppose there exists constant $C_f>0$, $|f(\bxi)|<C_f$ for $\bxi \in \Theta$ almost everywhere. Then,
        \begin{equation*}
            \tilde{L}_J \leq C_f\left| \frac{\int_{\Theta\setminus B_{\delta}(\bxi_0)}P^{(J)}\left(\Rb^{(J)}, \Tb^{(J)}\mid \bxi \right)\pi(\bxi)d\bxi}{\pi(\bxi_0)P^{(J)}(\Rb^{(J)},\Tb^{(J)}\mid\bxi_0)} \right| = o_{P_{\bxi_0}}(J^{-1}).
        \end{equation*}

        Therefore,
        \begin{equation*}
            \left|\tilde{L}_J\det(\tilde{\cI}_J(\tilde{\bxi}_J))^{1/2} \right| = o_{P_{\bxi_0}}(1).
        \end{equation*}

        \textbf{(2)}
        Let $M_{\delta,J}=B_{\delta}(\bxi_0)$, $U_J=\int_{M_{\delta,J}}P^{(J)}\left(\Rb^{(J)},\Tb^{(J)}\mid \bxi\right)\pi(\bxi)d\bxi $.
        \begin{align*}
            &\frac{U_J\det(\tilde{\cI}_J(\tilde{\bxi}_J))^{1/2}}{\pi(\tilde{\bxi}_J)P^{(J)}(\Rb^{(J)},\Tb^{(J)}\mid \tilde{\bxi}_J)}\\
            &= \frac{\det(\tilde{\cI}_J(\tilde{\bxi}_J))^{1/2}}{\pi(\tilde{\bxi}_J)}
            \int_{M_{\delta,J}}\exp\left\{ -\frac{1}{2} \left(\bxi-\tilde{\bxi}_J\right)^\top \right.\\
            &\qquad\left.\times\left[ \cI_J(\tilde{\bxi}_J)(\Ib_2 - E_J(\bxi)) + \bSigma^{-1} \right]
            \left(\bxi-\tilde{\bxi}_J\right) \right\}d\bxi.
        \end{align*}

        From Lemma \ref{lem:consistency}, for any $\varepsilon > 0$,
        \begin{align*}
            &(1-o_{P_{\bxi_0}}(1))\int_{M_{\delta,J}}\exp\left(-\frac{\varepsilon+1}{2}\left(\bxi-\tilde{\bxi}_J\right)^\top\tilde{\cI}_J(\tilde{\bxi}_J)\left(\bxi-\tilde{\bxi}_J\right)\right)d\bxi\\
            &\leq \int_{M_{\delta,J}}\exp\left( -\frac{1}{2} \left(\bxi-\tilde{\bxi}_J\right)^\top \left( \cI_J(\tilde{\bxi}_J)(\Ib_2 - E_J(\bxi)) + \bSigma^{-1} \right) \left(\bxi-\tilde{\bxi}_J\right) \right)d\bxi\\
            &\leq (1+o_{P_{\bxi_0}}(1))\int_{M_{\delta,J}}\exp\left(-\frac{1-\varepsilon}{2}\left(\bxi-\tilde{\bxi}_J\right)^\top\tilde{\cI}_J(\tilde{\bxi}_J)\left(\bxi-\tilde{\bxi}_J\right)\right)d\bxi.
        \end{align*}

        Therefore,
        \begin{align*}
            &(1-o_{P_{\bxi_0}}(1)){\Phi}\left( \sqrt{1+\varepsilon} \tilde{\cI}_J(\tilde{\bxi}_J)^{1/2}(M_{\delta,J}-\tilde{\bxi}_J) \right)(2\pi)\frac{\pi(\tilde{\bxi}_J)}{1+\varepsilon}\\
            &\leq \frac{U_J\det(\tilde{\cI}_J(\tilde{\bxi}_J))^{1/2}}{P^{(J)}(\Rb^{(J)},\Tb^{(J)}\mid\tilde{\bxi}_J)}\\
            &\leq (1+o_{P_{\bxi_0}}(1)){\Phi}\left( \sqrt{1-\varepsilon} \tilde{\cI}_J(\tilde{\bxi}_J)^{1/2}(M_{\delta,J}-\tilde{\bxi}_J) \right)(2\pi)\frac{\pi(\tilde{\bxi}_J)}{1-\varepsilon}.
        \end{align*}
        By setting $\varepsilon$ arbitrarily small and continuous mapping theorem from the consistency of $\tilde{\bxi}_J$, the desired result is obtained.
    \end{proof}

    \begin{corollary}
        \label{coro:converge}
        Suppose a sequence $\{R_j,T_j\}_{j\in \NN}$ generated from a fixed $\bxi_0\in \Theta$, under Assumption \ref{assump:finite}-\ref{assump:fisher_info}, for $J\to\infty$,
        \begin{equation*}
            \left( \frac{P^{(J)}(\Rb^{(J)},\Tb^{(J)})\det(\tilde{\cI}_J (\tilde{\bxi}_J))^{1/2}}{P^{(J)}(\Rb^{(J)},\Tb^{(J)}\mid\tilde{\bxi}_J)} \right)^{-1} \stackrel{P_{\bxi_0}}{\to} \frac{1}{(2\pi)\pi(\bxi_0)}.
        \end{equation*}
    \end{corollary}
    \begin{proof}
        By setting $G_J(B)=\mathbb{R}^2$, $J\in \NN$, from Lemma \ref{lem:converge}, we have the desired result.
    \end{proof}

    \subsection{Proof of Theorem 2}
    First, we show the result for all bounded $B$. Then we extend to unbounded $B$ and convergence in probability. Let $G_J(B)=\{\tilde{\cI}_J(\tilde{\bxi}_J)^{-1/2}\bx+\tilde{\bxi}_J:\bx\in B \} = \tilde{\cI}_J(\tilde{\bxi}_J)^{-1/2}B + \tilde{\bxi}_J $. Note that
    \begin{align*}
        &P(\tilde{\cI}_J(\tilde{\bxi}_J)^{-1/2}(\bxi-\tilde{\bxi}_J)\in B\mid \Rb^{(J)}, \Tb^{(J)})\\
        &= \underbrace{\frac{\det(\tilde{\cI}_J) \int_{G_J(B)}P^{(J)}(\Rb, \Tb\mid \bxi)\pi(\bxi)d\bxi}{P^{(J)}(\Rb^{(J)},\Tb^{(J)}\mid \tilde{\bxi}_J)}}_{\alpha_1} \cdot \underbrace{\left(\frac{P^{(J)}(\Rb^{(J)},\Tb^{(J)}) \det(\tilde{\cI}_J (\tilde{\bxi}_J))}{P^{(J)}(\Rb^{(J)},\Tb^{(J)}\mid \tilde{\bxi}_J)}\right)^{-1}}_{\alpha_2}.
    \end{align*}

    From Lemma \ref{lem:consistency}, $\tilde{\bxi}_J\stackrel{P_{\bxi_0}}{\to}\bxi_0$, hence,
    \begin{align*}
        \left\|\tilde{\cI}_J^{-1}(\tilde{\bxi}_J)\right\| = \frac{1}{J}\left\| \left( \frac{1}{J}\tilde{\cI}_J(\tilde{\bxi}_J) \right)^{-1} \right\| = O_{P_{\bxi_0}}\left(\frac{1}{J}\right).
    \end{align*}
    Thus, $\tilde{\cI}_J^{-1}(\tilde{\bxi}_J)\stackrel{P_{\bxi_0}}{\to} 0$.

    From Lemma \ref{lem:converge} (2),
    \begin{equation*}
        \alpha_1 \stackrel{P_{\bxi_0}}{\to} {\Phi}\left( \tilde{\cI}_J(\tilde{\bxi}_J)^{-1/2}(G_J(B) - \tilde{\bxi}_J) \right) \pi(\bxi_0)(2\pi) = {\Phi}\left( B \right) \pi(\bxi_0)(2\pi), \quad J\to \infty.
    \end{equation*}

    From Corollary \ref{coro:converge},
    \begin{equation*}
        \alpha_2 \stackrel{P_{\bxi_0}}{\to} \frac{1}{(2\pi)\pi(\bxi_0)}.
    \end{equation*}

    Therefore, for every bounded $B$, combining limit distribution of $\alpha_1$ and $\alpha_2$, we have the desired result. Then, for unbounded $B\in \cB(\Theta)$, define the posterior probability measure as
    \begin{equation*}
        \tilde{\Psi}_J(A) = \int_{G_J(A)}P(\bxi\mid \Rb^{(J)},\Tb^{(J)})d\bxi.
    \end{equation*}

    For an unbounded Borel set $B$, it can be written as $B=\cup_{m=1}^{\infty}B_m $, where $B_m\cap B_n = \emptyset$, $\forall m\neq n$, and $B_m$'s are bounded. Hence, for any $\varepsilon>0$,
    \begin{equation*}
        \lim_{J\to\infty} P_{\bxi_0}\left( \left| \tilde{\Psi}_J(B_m) - \Phi_2(B_m) \right| < \varepsilon \right) = 1.
    \end{equation*}

    Let $\varepsilon=6\varepsilon^\prime/(\pi m^2)$, we have
    \begin{align*}
        |\tilde{\Psi}_J(B) - \Phi_2(B)| &\leq \sum_{m=1}^{\infty}|\tilde{\Psi}_J(B_m)-\Phi_2(B_m)|\\
        &< \sum_{m=1}^{\infty} \frac{6\varepsilon^{\prime}}{\pi^2 m^2}\\
        &= \varepsilon^{\prime}.
    \end{align*}
    Hence, the result holds for arbitrary $B$.

    Let $H_{d,\epsilon}(\bxi^\prime)=P(|\tilde{\Psi}_J(B)-\Phi_2(B)| > \epsilon\mid \bxi_0=\bxi^{\prime})$. Since $H_{d,\epsilon}\leq 1$ uniformly, by dominated convergence theorem,
    \begin{align*}
        \lim_{J\to\infty}P(|\tilde{\Psi}_J(B)-\Phi_2(B)|>\epsilon) &= \lim_{J\to\infty} \int_{\Theta} H_{d,\epsilon}(\bxi)d\cG(\bxi)\\
        &=\int_{\Theta}\lim_{J\to\infty} H_{d,\epsilon}(\bxi)d\cG(\bxi)=0,
    \end{align*}
    where $\cG(\bxi)$ is any proper probability measure on $\Theta$. The last inequality comes from convergence in $P_{\bxi_0}$.

\section{Additional Simulations}
\color{black}
\subsection{Robustness to Model Misspecification}
We examine the robustness of LaRT under link-function misspecification. We present the results for the doubly-misspecified Logistic-Poisson generative model, as it represents the most challenging scenario. In generalized linear models (GLMs), utilizing a misspecified link function typically shifts the parameter estimates by an unknown scalar constant \citep{li1989regression}. Therefore, rather than evaluating RMSE, we assess performance using the Spearman rank correlation between the estimated quantities and the ground truth. This rank-based evaluation aligns with two practical objectives of LLM evaluation. First, relative rankings are prioritized over absolute ability scores. Second, the true parameters under link function misspecification only shift by an unknown constant, and the true ordering should be preserved. For a baseline comparison, we fit a correctly specified logistic-link IRT model using Marginal Maximum Likelihood (MML) via the \texttt{girth} Python package with default settings, chosen for its algorithmic similarity to our approach.

The true parameters for this doubly-misspecified setting are generated similarly to the previous experiment. For the response accuracy, each entry of $\ba$ is drawn from $\text{Unif}(0.5,1)$ and $\bb$ from $N(0,0.5)$. For the CoT length, we draw each entry of $\bomega$ from $N(0,40^2)$ and $\bvarphi$ from $\text{Unif}(30,70)$. To reflect the empirically large token counts observed in CoT reasoning, we generate the lengths from a shifted Poisson distribution $T_{ij} \sim \text{Poisson}(\omega_j - \tau_i\varphi_j + 800)$. Across all simulations in this setting, we maintain $J=50$ and $\rho=-0.8$, varying the sample size $N \in \{100, 200, 500\}$. 200 independent replications are performed for each setting.

The result is presented in Figure~\ref{fig:logistic_poisson}. We compare the performance in Spearman correlation with the correctly specified IRT model. Notably, LaRT yields uniformly higher rank-estimation accuracy for both $\btheta$ and $\ba$, and performs only marginally worse for $\bb$. This robust performance stems from two primary factors. First, our proposed SAEM algorithm avoids the numerical integration approximations inherent to the baseline. Under a logistic link, the absence of a closed-form posterior for $\btheta$ necessitates Gauss-Hermite quadrature during the MML algorithm's E-step, which introduces approximation biases into the estimates of both latent traits and item parameters. Second, the joint modeling in LaRT allows the latent ability estimation to borrow strength from the CoT lengths,  which offers an informational advantage that persists even under link misspecification.

Furthermore, Figure~\ref{fig:logistic_poisson_time} illustrates the estimation accuracy for the CoT-related parameters. All parameters achieve high Spearman correlations, with the correlation $\rho$ performing comparably to the correctly specified case shown in Figure 3 of the main article. This confirms that when token counts are sufficiently large, the continuous log-normal distribution serves as a highly effective approximation for the discrete Poisson data-generating process.

\begin{figure}[h!]
    \centering
    \includegraphics[width=\linewidth]{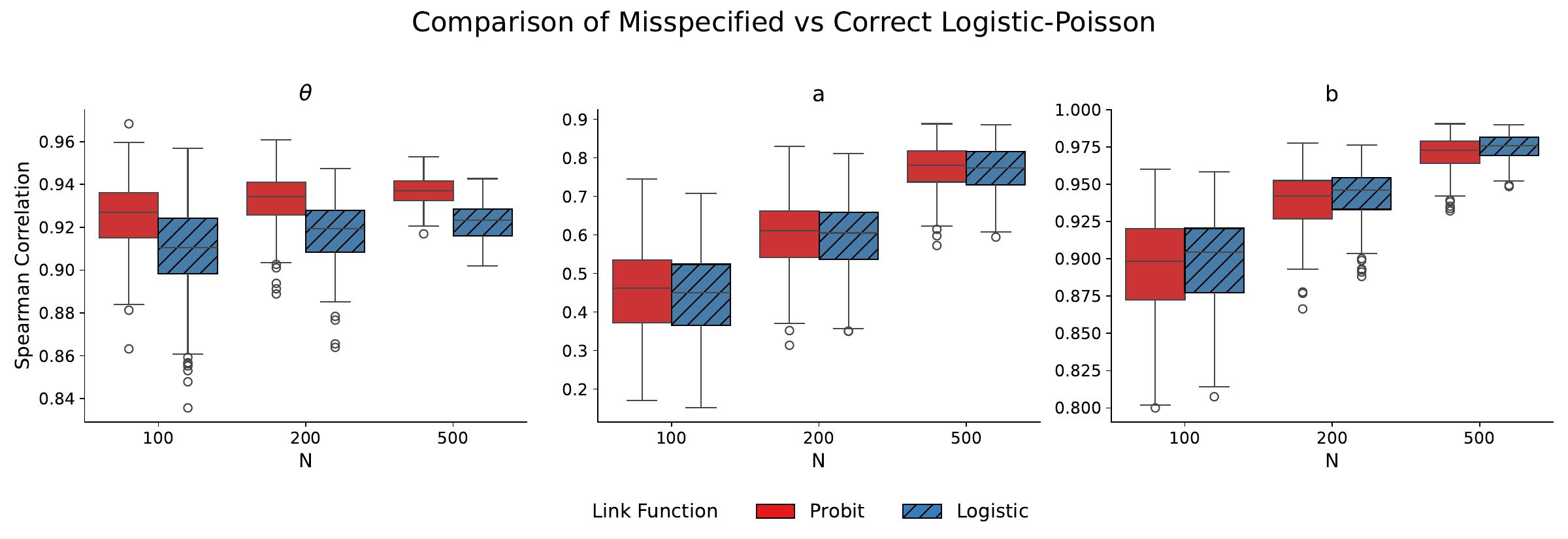}
    \caption{\textcolor{black}{Spearman rank correlation of parameter estimates from the doubly-misspecified LaRT model versus a correctly specified standard IRT model. Despite the link-function misspecification, LaRT achieves higher ranking accuracy for $\btheta$ and $\ba$, and remains highly competitive for $\bb$.}}
    \label{fig:logistic_poisson}
\end{figure}

\begin{figure}[h!]
    \centering
    \includegraphics[width=0.66\linewidth]{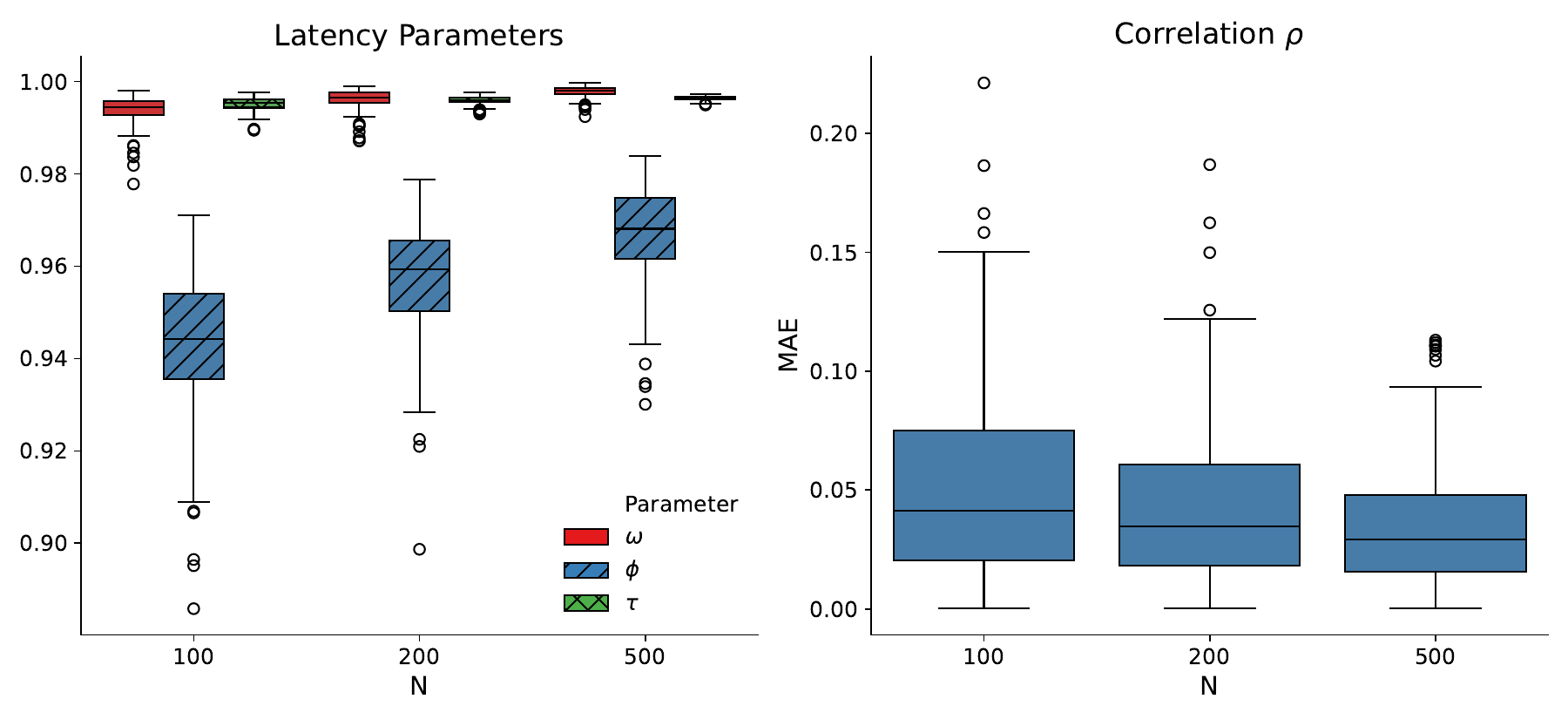}
    \caption{\textcolor{black}{Spearman rank correlation of the CoT-related time parameter estimates using the misspecified LaRT model. High correlations across all parameters demonstrate the robustness of the log-normal approximation for large-count Poisson data.}}
    \label{fig:logistic_poisson_time}
\end{figure}

\subsection{Converging Scenarios between LaRT and IRT}
To further illustrate the conditions under which LaRT and IRT align, we conduct additional simulations. First, we vary the correlation parameter $\rho \in \{-0.2, -0.4, -0.6, -0.8\}$ while fixing $N=200$ and $J=50$. We restrict $\rho$ to negative values to reflect the empirical observation that longer Chain-of-Thought (CoT) sequences generally correlate with higher latent ability. The data generation process for the remaining parameters $\bOmega=\{\ba,\bb,\bomega,\bvarphi,\blambda\}$ follows the simulation setup in the simulation section.

Figure~\ref{fig:rho_diff} presents the estimation accuracy for $\btheta$, $\ba$, and $\bb$. While LaRT consistently outperforms standard IRT across all settings, the performance gap narrows as the magnitude of the correlation decreases. Fortunately, our real-data analysis demonstrates a strong empirical correlation between latent ability and latent speed, underscoring the practical advantage of LaRT in applied settings.

\begin{figure}[h!]
    \centering
    \includegraphics[width=\linewidth]{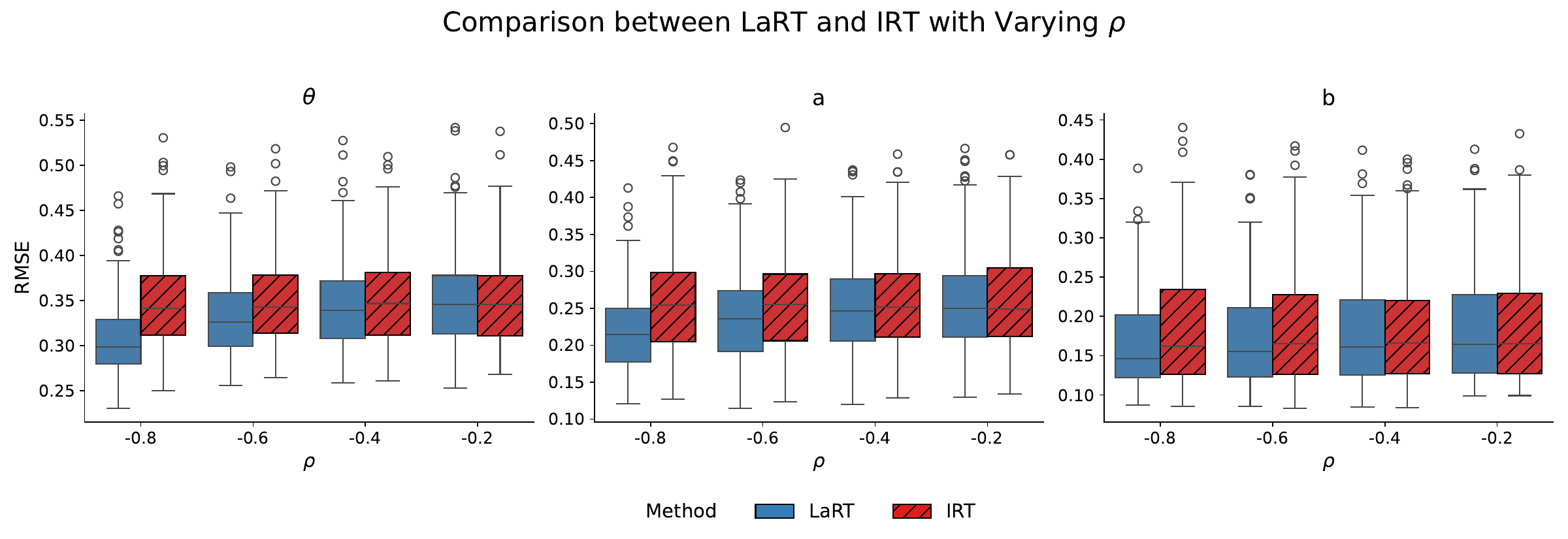}
    \caption{\textcolor{black}{Performance comparison between LaRT and IRT across varying $\rho$. As the magnitude of $\rho$ increases, the advantage of LaRT becomes more significant.}}
    \label{fig:rho_diff}
\end{figure}

Second, we evaluate the impact of the test length by varying $J \in \{20, 50, 100, 200\}$, while fixing $N=200$ and $\rho=-0.8$ (with $\bOmega$ generated as before). As shown in Figure~\ref{fig:J_diff}, while LaRT maintains a significant edge over IRT, this comparative advantage diminishes as $J$ increases. This empirical trend aligns directly with our asymptotic variance analysis. Specifically, the Fisher Information for $\theta_i$ is given by:
\begin{equation*}
    \tilde{\cI}_J(\theta) = \underbrace{\frac{1}{1-\rho^2}}_{\cI_{\rho}} + \underbrace{\sum_{j=1}^J \frac{a_j^2\phi(a_j\theta_i+b_j)^2}{\Phi(a_j\theta_i+b_j)[1-\Phi(a_j\theta_i+b_j)]}}_{\cI_{\text{item}}} .
\end{equation*}
For a fixed $\rho$, as $J\to\infty$, $\tilde{\cI}_J(\theta)$ is dominated by $\cI_{\text{item}}$. Thus, the performance of LaRT and IRT becomes closer when $J$ grows, justifying the empirical observation.

\begin{figure}[h!]
    \centering
    \includegraphics[width=\linewidth]{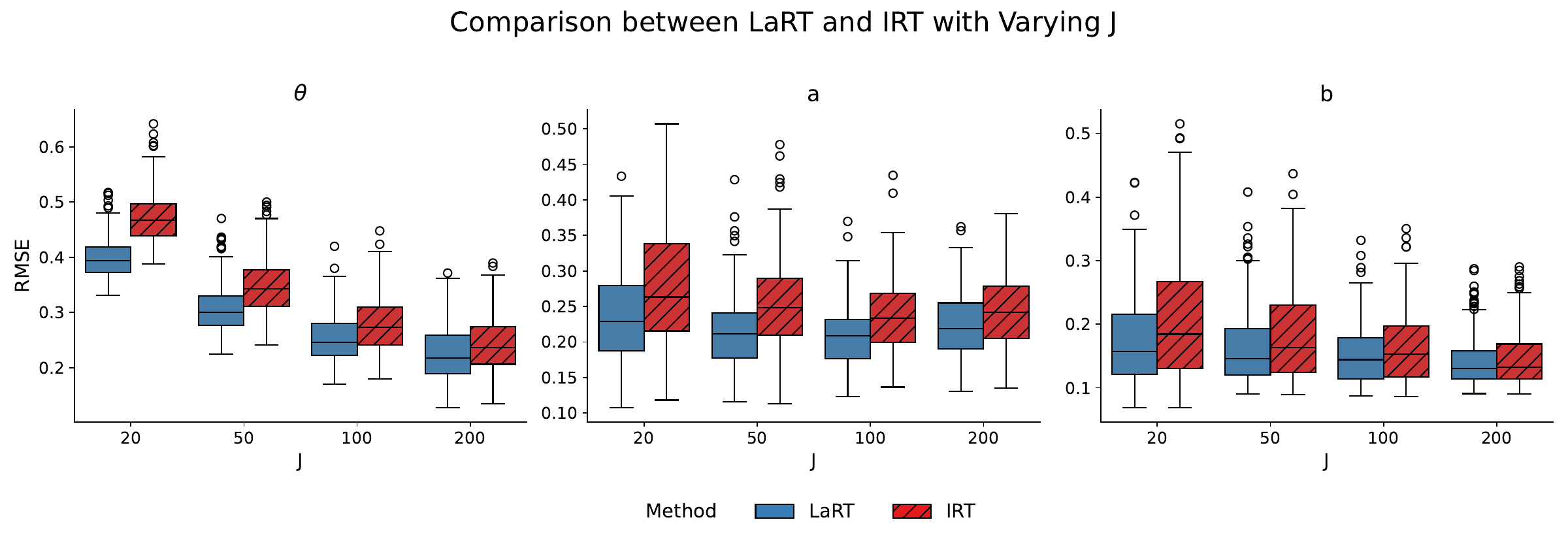}
    \caption{\textcolor{black}{Performance comparison between LaRT and IRT across varying $J$. As $J$ increases, the advantage of LaRT diminishes.}}
    \label{fig:J_diff}
\end{figure}

\color{black}

\subsection{Simulation comparison with traditional SAEM implementation}
\label{append:SAEM}

This section compares the parameter estimation performance of smart initialized SAEM against traditional SAEM. The traditional SAEM is implemented with the first 20 steps with weight $\alpha_t=1$, followed by a decay of $\alpha_t=1/(t-20)$. The smart initialized SAEM starts with $\alpha_t=1/t$. Parameters are generated in the same configuration as in Section 6 of the main text. We conduct 200 parallel simulations.

The result is presented in Figure \ref{fig:sim_com_saem_1} and \ref{fig:sim_com_saem_2}. For the probit part, there are significant outliers in the estimation of $\ba$ and $\bb$. The outlier results from instability of the stochastic approximation E-step, when $\alpha_t=1$. The suboptimal optimization target drives the estimate away from the true optimal region. Even excluding the outlier, the smart initialized SAEM consistently yields lower estimation error across all parameters. When the number of "burn-in" steps increases, the estimation accuracy exacerbates. We do not present the result because the many outliers when the "burn-in" steps become larger, and not ideal for presentation.

\begin{figure}[h!]
    \centering
    \includegraphics[width=\linewidth]{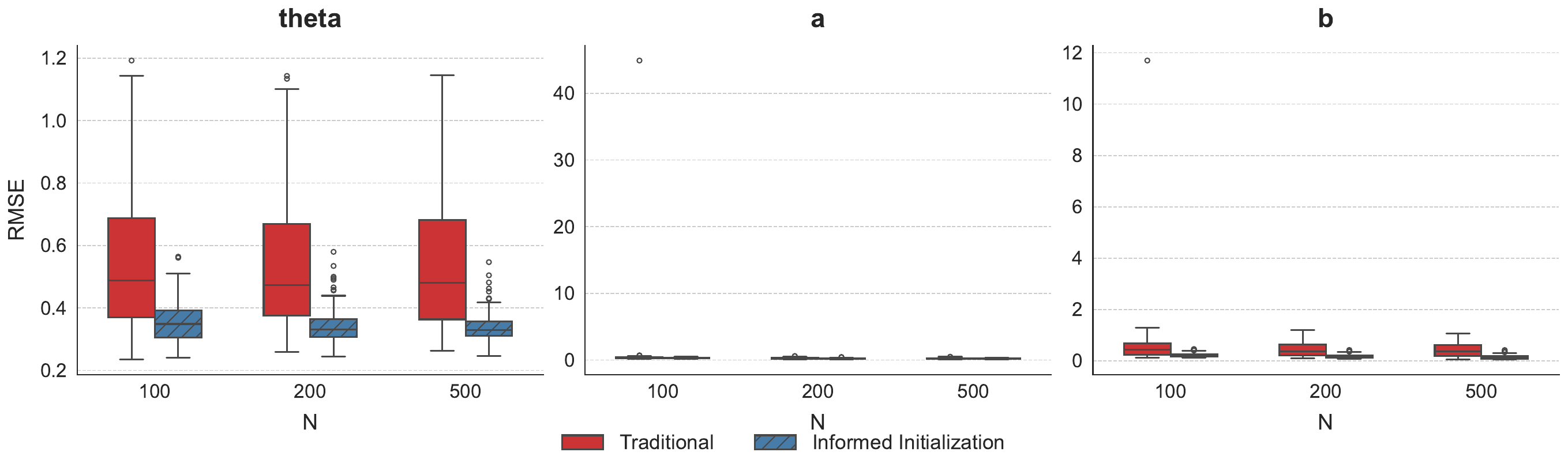}
    \caption{Comparison between smart initialized SAEM and traditional SAEM in $\btheta$, $\ba$, and $\bb$.}
    \label{fig:sim_com_saem_1}
\end{figure}

\begin{figure}[h!]
    \centering
    \includegraphics[width=0.6\linewidth]{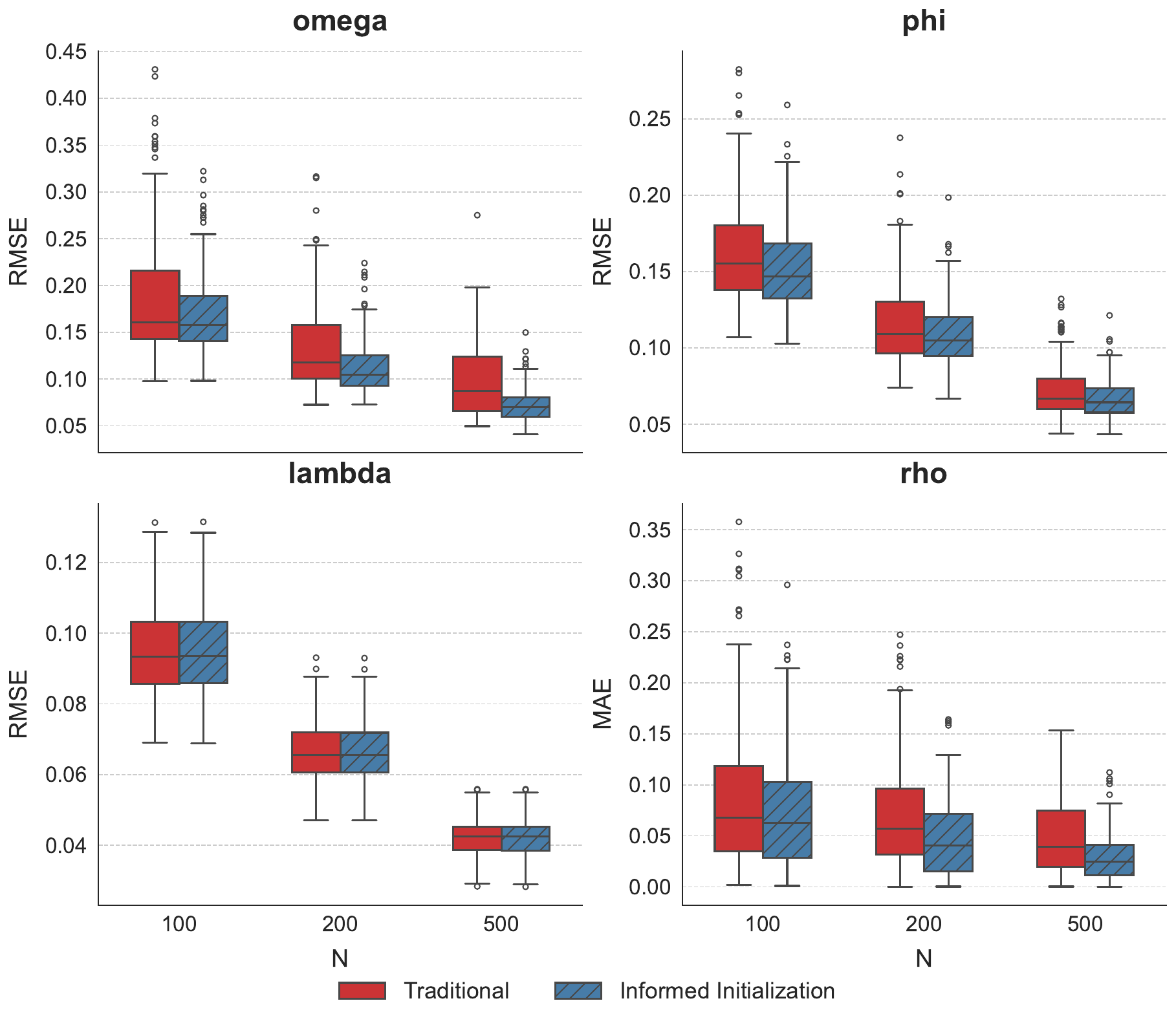}
    \caption{Comparison between smart initialized SAEM and traditional SAEM in $\bomega$, $\bvarphi$, $\blambda$, and $\rho$.}
    \label{fig:sim_com_saem_2}
\end{figure}

\section{List of Evaluated LLMs and Hyperparameters for generation}
\label{append:llm}
\subsection{List of Evaluated LLMs}
The list of models is as follows. For MATH500, we have an additional LLM google/gemma-2-27b-it.
\footnotesize
\begin{multicols}{2}
    \begin{itemize}[noitemsep, topsep=0pt]
        \item 01-ai/Yi-34B
        \item baidu/ERNIE-4.5-21B-A3B-PT
        \item baidu/ERNIE-4.5-21B-A3B-Thinking
        \item deepseek-ai/DeepSeek-R1-0528-Qwen3-8B
        \item deepseek-ai/DeepSeek-R1-Distill-Llama-8B
        \item deepseek-ai/DeepSeek-R1-Distill-Qwen-1.5B
        \item deepseek-ai/DeepSeek-R1-Distill-Qwen-14B
        \item deepseek-ai/DeepSeek-R1-Distill-Qwen-32B
        \item deepseek-ai/DeepSeek-R1-Distill-Qwen-7B
        \item dphn/dolphin-2.9.1-yi-1.5-34b
        \item dphn/Dolphin-Mistral-24B-Venice-Edition
        \item google/gemma-2b-it
        \item google/gemma-3-1b-it
        \item google/gemma-3-1b-pt
        \item google/gemma-7b-it
        \item google/vaultgemma-1b
        \item HuggingFaceTB/SmolLM3-3B
        \item huihui-ai/Huihui-gpt-oss-20b-BF16-abliterated
        \item huihui-ai/Huihui-Qwen3-8B-abliterated-v2
        \item ibm-granite/granite-3.3-2b-instruct
        \item internlm/internlm2-chat-20b
        \item LGAI-EXAONE/EXAONE-4.0.1-32B
        \item LLM360/K2-Think
        \item meta-llama/Llama-2-7b-chat-hf
        \item meta-llama/Llama-2-7b-hf
        \item meta-llama/Llama-3.1-8B-Instruct
        \item meta-llama/Llama-3.2-1B
        \item meta-llama/Llama-3.2-1B-Instruct
        \item meta-llama/Llama-3.2-3B
        \item meta-llama/Llama-3.2-3B-Instruct
        \item meta-llama/Meta-Llama-3-8B
        \item meta-llama/Meta-Llama-3-8B-Instruct
        \item microsoft/Phi-3.5-mini-instruct
        \item microsoft/Phi-3.5-MoE-instruct
        \item microsoft/phi-4
        \item microsoft/Phi-4-mini-instruct
        \item microsoft/Phi-4-reasoning
        \item microsoft/Phi-4-reasoning-plus
        \item mistralai/Magistral-Small-2507
        \item mistralai/Magistral-Small-2509
        \item mistralai/Mistral-7B-Instruct-v0.1
        \item mistralai/Mistral-7B-Instruct-v0.2
        \item mistralai/Mistral-7B-Instruct-v0.3
        \item mistralai/Mistral-Small-3.2-24B-Instruct-2506
        \item mistralai/Mistral-Small-Instruct-2409
        \item moonshotai/Moonlight-16B-A3B
        \item moonshotai/Moonlight-16B-A3B-Instruct
        \item nvidia/AceReason-Nemotron-1.1-7B
        \item nvidia/AceReason-Nemotron-14B
        \item nvidia/Llama-3.1-Nemotron-8B-UltraLong-4M-Instruct
        \item nvidia/Nemotron-Research-Reasoning-Qwen-1.5B
        \item nvidia/NVIDIA-Nemotron-Nano-12B-v2
        \item nvidia/NVIDIA-Nemotron-Nano-9B-v2
        \item nvidia/OpenReasoning-Nemotron-1.5B
        \item nvidia/OpenReasoning-Nemotron-7B
        \item openai-community/gpt2
        \item openai/gpt-oss-20b
        \item openbmb/MiniCPM4.1-8B
        \item Qwen/Qwen1.5-32B
        \item Qwen/Qwen2-7B-Instruct
        \item Qwen/Qwen2.5-0.5B-Instruct
        \item Qwen/Qwen2.5-1.5B-Instruct
        \item Qwen/Qwen2.5-14B-Instruct
        \item Qwen/Qwen2.5-32B-Instruct
        \item Qwen/Qwen2.5-3B-Instruct
        \item Qwen/Qwen2.5-7B-Instruct
        \item Qwen/Qwen3-0.6B
        \item Qwen/Qwen3-1.7B
        \item Qwen/Qwen3-14B
        \item Qwen/Qwen3-30B-A3B
        \item Qwen/Qwen3-30B-A3B-Instruct-2507
        \item Qwen/Qwen3-30B-A3B-Thinking-2507
        \item Qwen/Qwen3-32B
        \item Qwen/Qwen3-4B
        \item Qwen/Qwen3-4B-Instruct-2507
        \item Qwen/Qwen3-4B-Thinking-2507
        \item Qwen/Qwen3-8B
        \item Qwen/QwQ-32B
        \item swiss-ai/Apertus-8B-Instruct-2509
        \item THUDM/GLM-4-9B-0414
        \item TinyLlama/TinyLlama-1.1B-Chat-v1.0
        \item zai-org/GLM-4-32B-0414
    \end{itemize}
\end{multicols}
\normalsize

\subsection{Hyperparameters and Prompts for generation}
For the hyperparameters of LLMs generation, we set the temperature to be 0.5, top p 0.95, max output tokens as 10,240, and repetition penalty of 1.05. Without the repetition penalty, some LLMs will keep repeat until reach the maximum output token. Therefore, we set a mild repetition penalty such that the CoT is not repeated, and the CoT length will be a better summary of the thinking quality of an LLM.

For the prompts, we use CoT zero-shot prompting and one-shot prompting. The specific forms of the prompts are as follows. The \{problem\} provides the detailed question of the item. The one-shot example comes from a question item in MATH dataset \citep{hendrycksmath2021} that is not included in MATH500.
\begin{enumerate}
    \item \textbf{Zero-shot Prompt}: Solve the following math problem. Be clear and concise.
    Problem: "\{problem\}"
    Provide a \textbf{step-by-step solution}. Start each step with a number followed by a period (e.g., '1.', '2.', etc.).
    Use basic LaTeX for mathematical expressions, such as for fractions, exponents, and variables. Avoid complex formatting.
    At the very end of your entire response, and only at the very end, state the final answer.
    This final answer must be enclosed in a single LaTeX box, like so: \boxed{{Your Answer}}.
    \item \textbf{One-shot Prompt}: Solve the following math problem. Please think \textbf{step-by-step} to obtain the solution. Use basic LaTeX for mathematical expressions, such as for fractions, exponents, and variables. Avoid complex formatting. At the very end of your entire response, and only at the very end, state the final answer. This final answer must be enclosed in a single LaTeX box, like so: \boxed{{Your Answer}}.

    Here is an example of how to format your response and think about solving the problem:
    Example Problem:
    What is the sum of the two values of $x$ for which $(x+3)^2 = 121$?

    Example Solution:
    Expanding the left side, we have $x^2+6x+9=121 \Rightarrow x^2+6x-112=0$. For a quadratic with the equation $ax^2+bx+c=0$, the sum of the roots is $-b/a$. Applying this formula to the problem, we have that the sum of the two roots is $-6/1=\boxed{-6}$.

    Solution: \boxed{{-6}}

    --- New Problem:
    \{problem\}.
\end{enumerate}

\section{Supplementary Materials for Application}
\label{append:applied}
\subsection{Confidence Interval Estimate of Latent Ability of LLMs}
Additionally, we present the asymptotic confidence intervals estimated by Theorem 2 of the main text in MATH500 for zero-shot models in Figure~\ref{fig:lart_math500_ci}. The LLMs are ordered by their estimated $\theta$. These confidence intervals allow us to assess whether one LLM outperforms another with statistical significance. Since $\sqrt{\text{sd}_1^2+\text{sd}_2^2}\leq \text{sd}_1+\text{sd}_2$, if the upper limit of the confidence interval of LLM A is smaller than the lower limit of the confidence interval of LLM B, then we can conclude at the 0.05 significance level that LLM B has better latent ability in this dataset than LLM A. For example,
we can conclude with 95\% confidence that Qwen3-30B-A3B-Instruct-2507 is better than Nvidia-Acereason-Nemotron-1.1-7B.

\begin{figure}[h!]
    \centering
    \includegraphics[width=\linewidth]{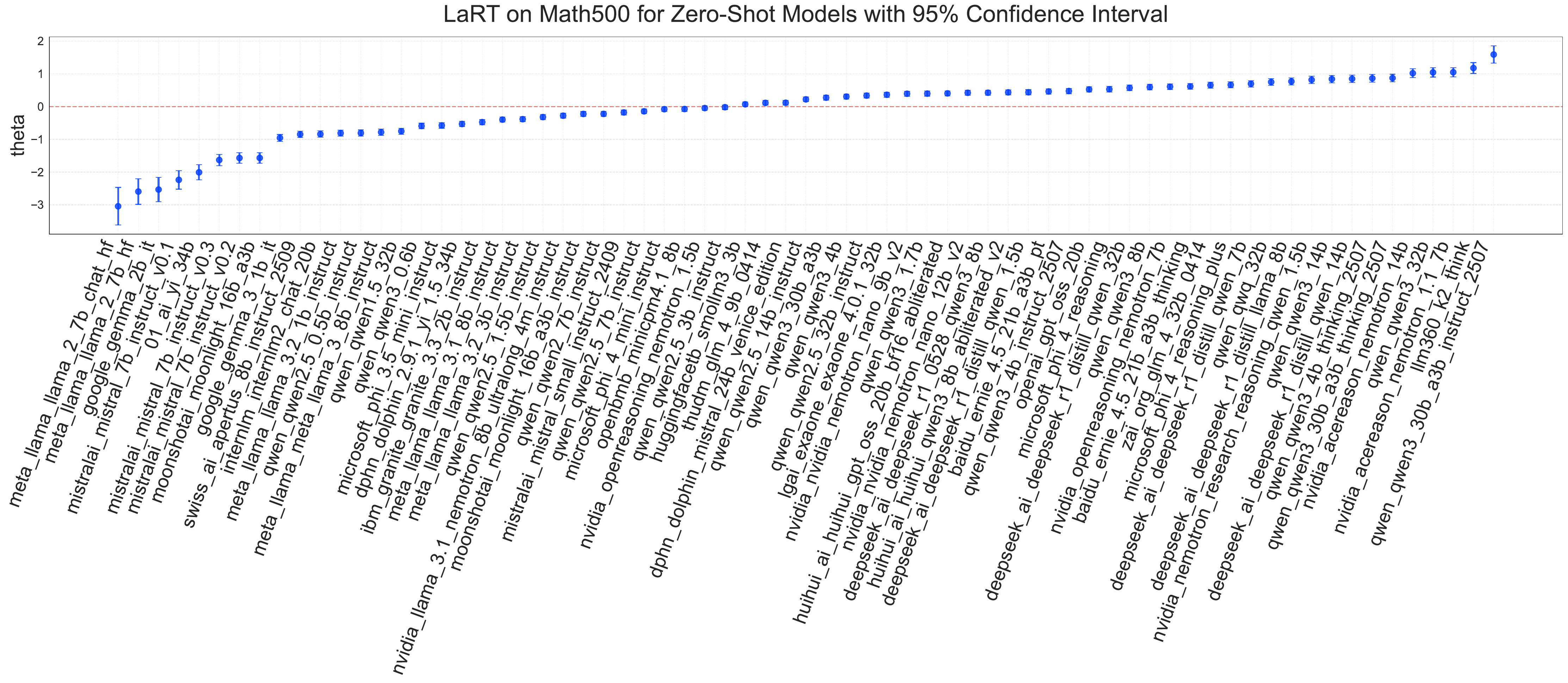}
    \caption{Latent ability estimates for zero-shot models on the MATH500 dataset using LaRT. The error bars represent 95\% asymptotic confidence intervals. Non-overlapping intervals indicate statistically significant differences in performance between LLMs.}
    \label{fig:lart_math500_ci}
\end{figure}

\subsection{Complete Ranking of Zero-shot and  One-shot Models}
\begin{figure}[h!]
    \centering
    \includegraphics[width=0.74\linewidth]{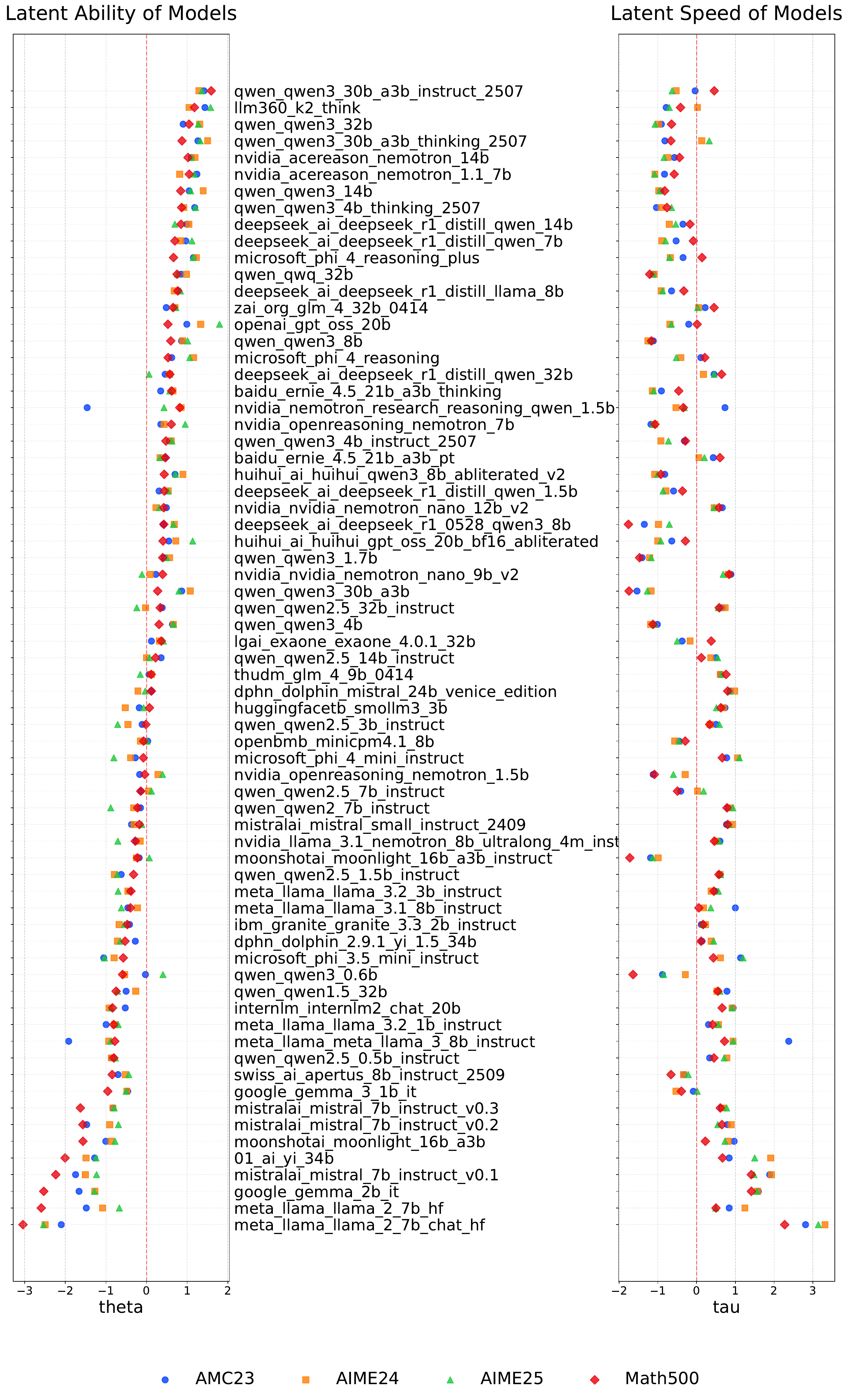}
    \caption{Estimated latent ability and latent speed of the LLMs with zero-shot prompting by LaRT. The left panel shows the estimated latent ability, while the right panel shows the estimated latent speed. The legend identifies the corresponding benchmark datasets. As the latent ability becomes larger, the latent speed becomes smaller (longer CoT), which matches intuition.}
    \label{fig:zero_model_all}
\end{figure}
\begin{figure}[h!]
    \centering
    \includegraphics[width=0.74\linewidth]{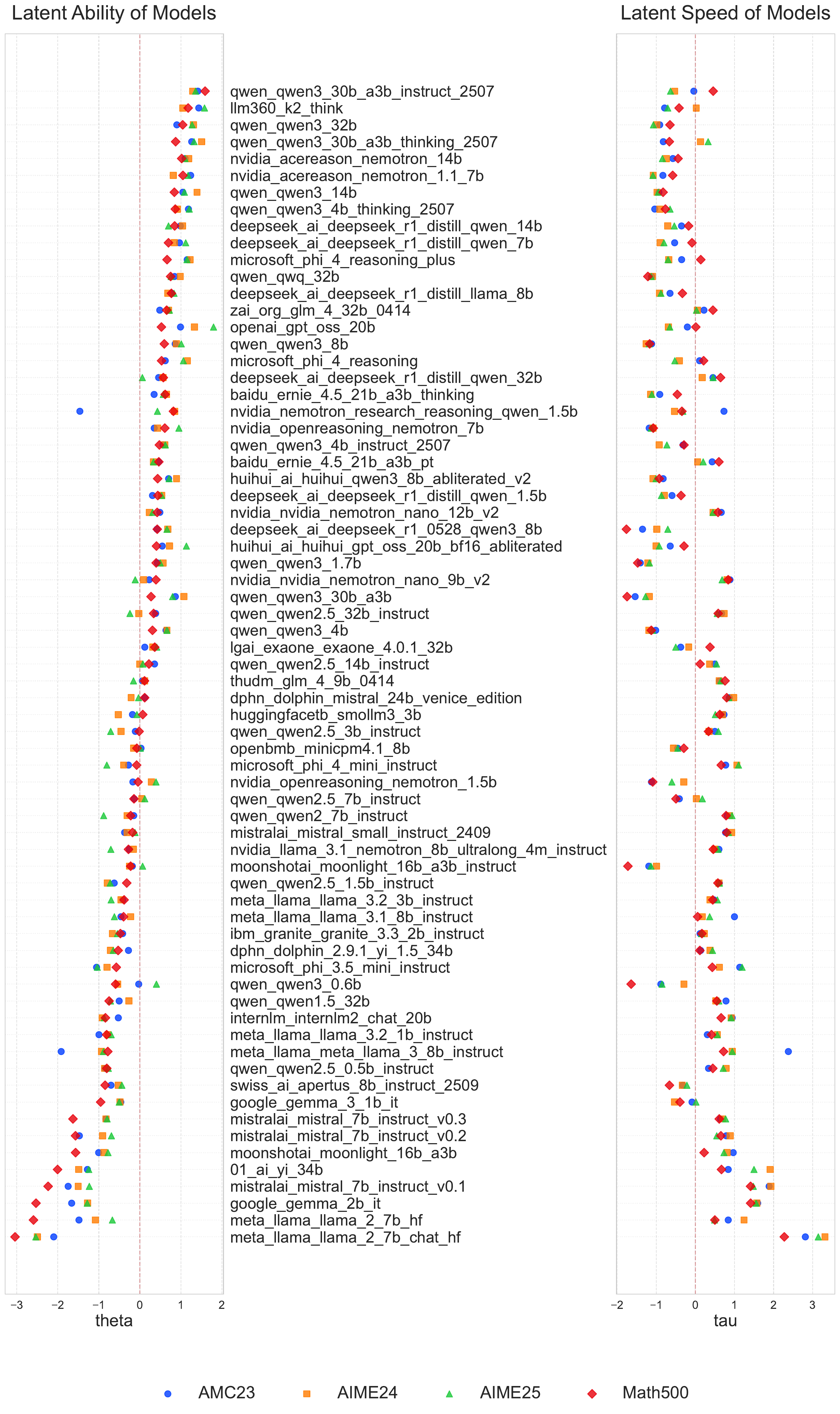}
    \caption{Estimated latent ability and latent speed of the LLMs with one-shot prompting by LaRT. The left panel shows the estimated latent ability, while the right panel shows the estimated latent speed. The legend identifies the corresponding datasets.}
    \label{fig:model_one_all}
\end{figure}

\subsection{Ranking differences of Other Datasets}
\begin{figure}[h!]
    \centering
    \includegraphics[width=\linewidth]{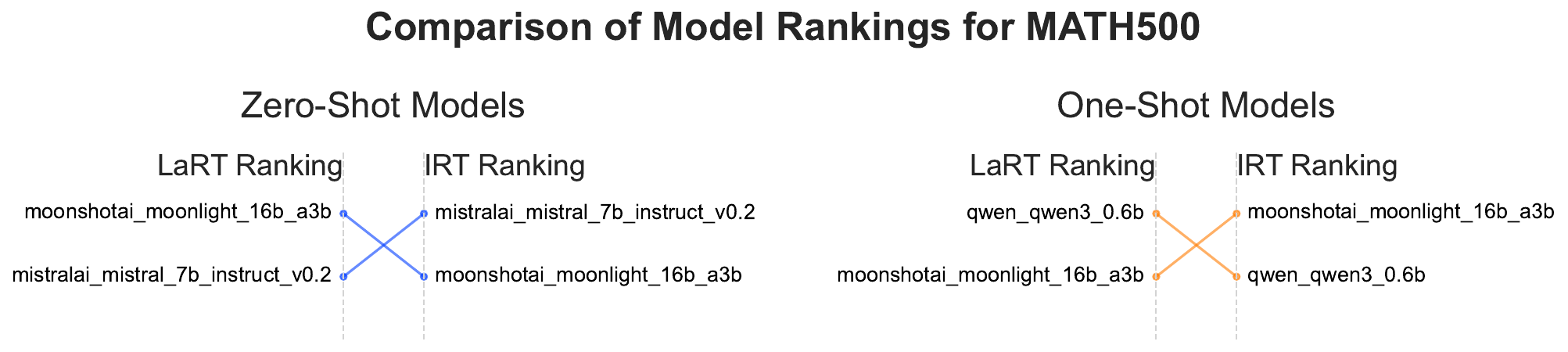}
    \caption{Differences in LLM rankings for both zero-shot models and one-shot models for MATH500. The left panel is for zero-shot models, and the right panel is for one-shot models. For each panel, rankings by LaRT are on the left, and rankings by IRT are on the right. LLMs higher in the plot have higher rankings. The lines connect the same models with different rankings by LaRT and IRT.}
    \label{fig:ranking_shift_math500}
\end{figure}

\begin{figure}[h!]
    \centering
    \includegraphics[width=\linewidth]{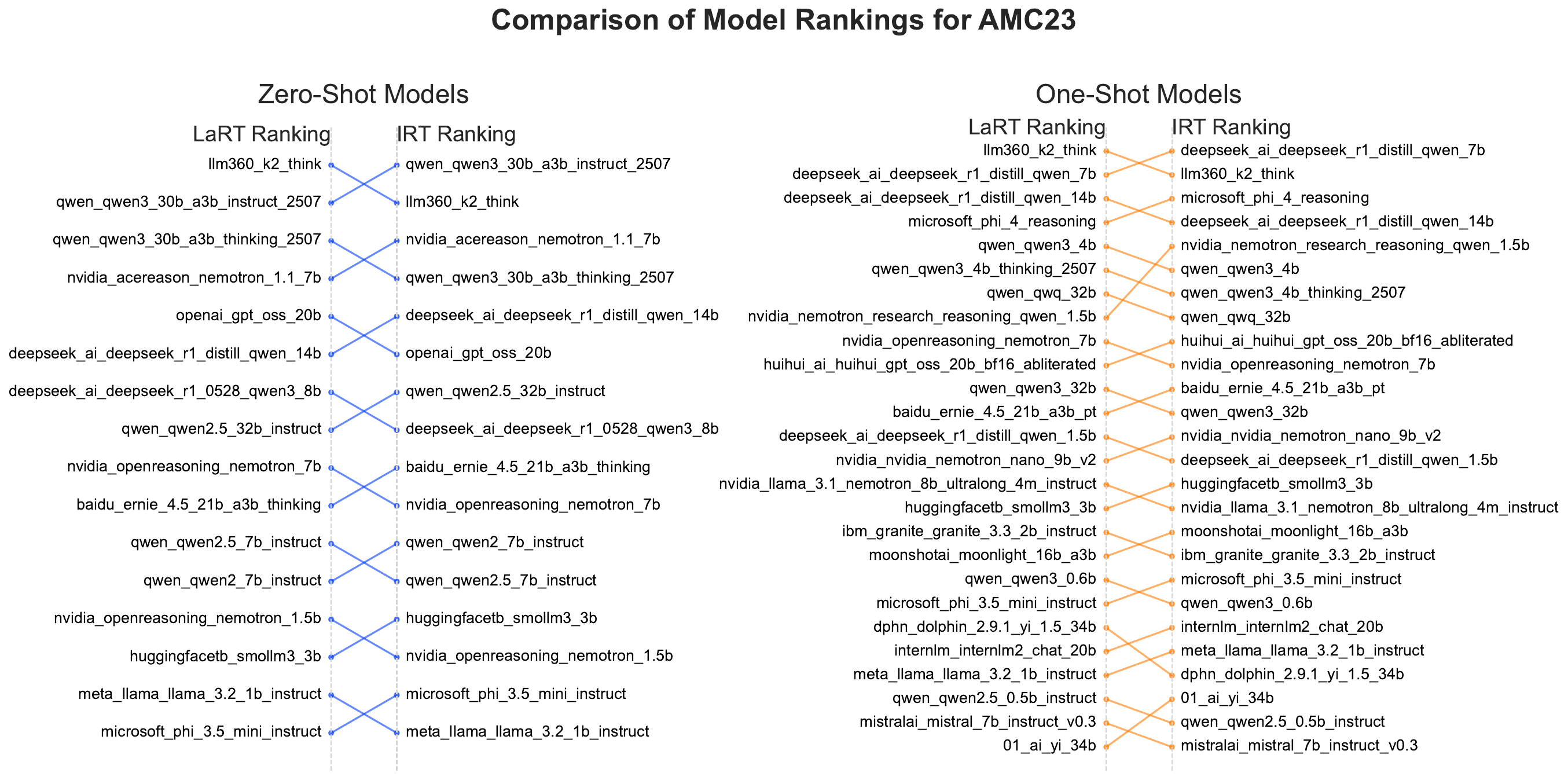}
    \caption{Differences in LLM rankings for both zero-shot models and one-shot models for AMC23. The left panel is for zero-shot models, and the right panel is for one-shot models. For each panel, rankings by LaRT are on the left, and rankings by IRT are on the right. LLMs higher in the plot have higher rankings. The lines connect the same models with different rankings by LaRT and IRT.}
    \label{fig:ranking_shift_amc23}
\end{figure}

\begin{figure}[h!]
    \centering
    \includegraphics[width=\linewidth]{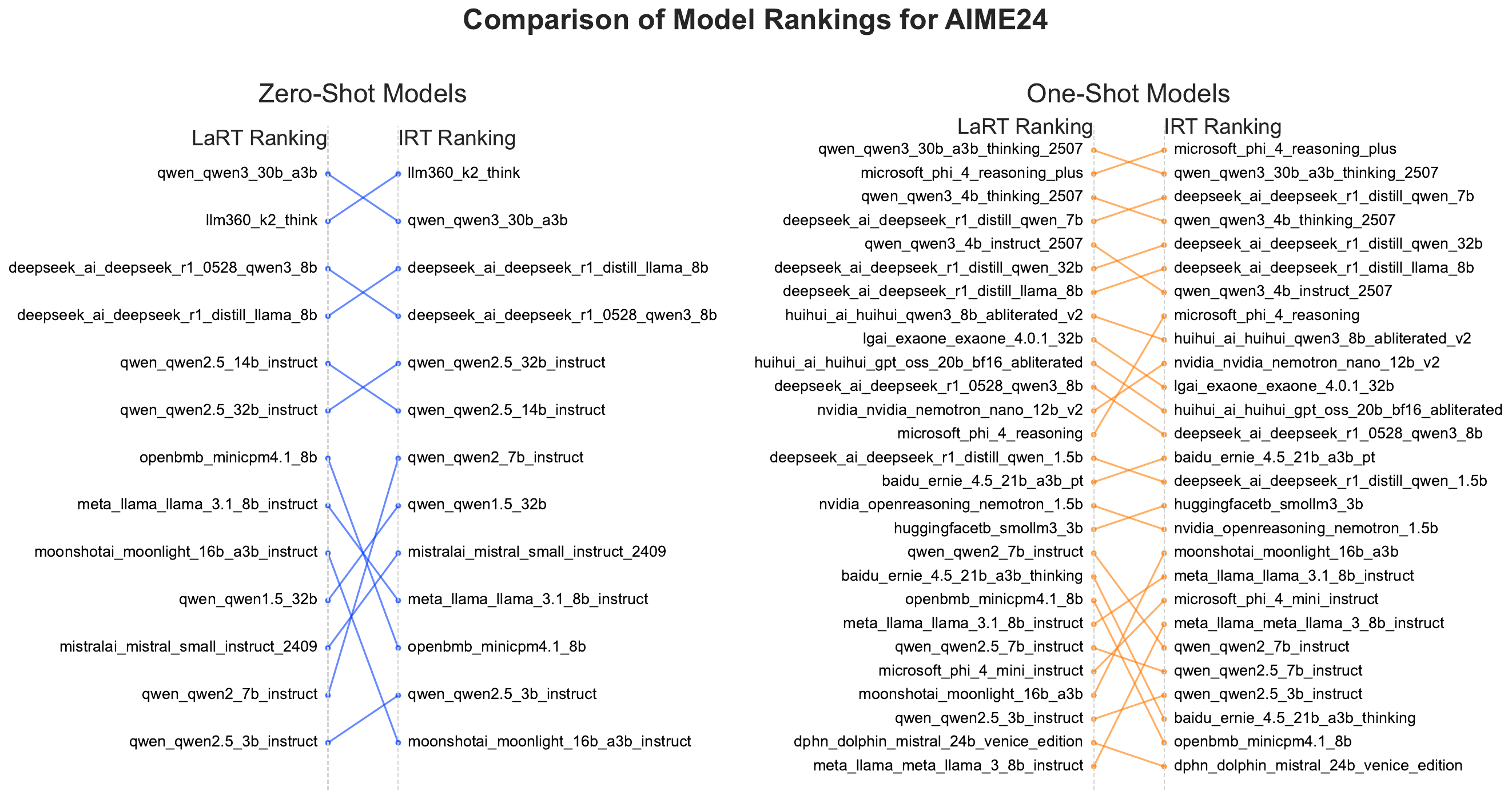}
    \caption{Differences in LLM rankings for both zero-shot models and one-shot models for AIME24. The left panel is for zero-shot models, and the right panel is for one-shot models. For each panel, rankings by LaRT are on the left, and rankings by IRT are on the right. LLMs higher in the plot have higher rankings. The lines connect the same models with different rankings by LaRT and IRT.}
    \label{fig:ranking_shift_aime24}
\end{figure}

\end{document}